%% file: main.tex
\documentclass[11pt,letterpaper]{article}

\usepackage[margin=1in]{geometry}
\usepackage{latexsym,graphicx}

\usepackage{amsmath,amssymb,amsfonts,amsthm}
\let\oldemptyset\emptyset

\usepackage{enumerate}

\usepackage{wrapfig}
\usepackage{subfigure}
\usepackage{float}
\usepackage{stmaryrd} 
\usepackage{epsfig}
\usepackage{xspace}
\usepackage{paralist}
\usepackage{mathtools}
\usepackage{enumitem}
\usepackage[T1]{fontenc}

\usepackage{tikz}
\usetikzlibrary{positioning, calc}
\usetikzlibrary{decorations.pathmorphing, arrows.meta}

\usepackage{cases}
\usepackage{caption}
\usepackage{algorithm}
\usepackage{algorithmicx}
\usepackage[colorlinks,linkcolor=blue]{hyperref}
\usepackage[noend]{algpseudocode}

\usepackage{graphicx}
\usepackage{wrapfig}
\usepackage{color}
\usepackage{thmtools}
\usepackage{thm-restate}

\numberwithin{figure}{section}
\newtheorem{definition}{Definition}

\newtheorem{corollary}{Corollary}

\newtheorem{lemma}{Lemma}
\newtheorem{claim}{Claim}
\newtheorem{observation}{Observation}

\usepackage{textcomp}
\usepackage{graphicx}
\usepackage{setspace}
\newcommand{\eat}[1]{}

\newcommand{\ti}{\widetilde{O}}
\usepackage{setspace}
\usepackage[utf8]{inputenc}
\usepackage{bbm}
\usepackage{comment}
\usepackage[maxnames=99]{biblatex}
\usepackage{xcolor}
\usepackage{hyperref}

\renewcommand{\textsc}[1]{{\textup{\scshape#1}}}

\usepackage{tcolorbox}

\newcommand*\widebar[1]{%
  \hbox{%
    \vbox{%
      \hrule height 0.1pt % The actual bar
      \kern0.3ex%         % Distance between bar and symbol
      \hbox{%
        \kern-0.1em%      % Shortening on the left side
        \ensuremath{#1}%
        \kern-0.1em%      % Shortening on the right side
      }%
    }%
  }%
} 

\usepackage{titlesec}
\titleformat{\paragraph}
{\normalfont\normalsize\bfseries}{\theparagraph}{1em}{}
\titlespacing*{\paragraph}
{0pt}{3.25ex plus 1ex minus .2ex}{1.5ex plus .2ex}

\definecolor{mygreen}{RGB}{10,110,230}
\definecolor{myred}{RGB}{10,110,230}

\hypersetup{
     colorlinks=true,
     citecolor= mygreen,
     linkcolor= myred,
     urlcolor= black
}

\newcommand{\eps}{\varepsilon}
\renewcommand{\b}[1]{\overline{#1}}
\newcommand{\vs}{V^{\text{sparse}}}
\newcommand{\vd}{V^{\text{dense}}}
\newcommand{\glr}{G_{\mathcal{L}, \mathcal{R}}}
\newcommand{\lst}{\frac{n}{\Delta}}
\DeclareMathOperator{\polylog}{polylog}

\renewcommand{\epsilon}{\varepsilon}
\let\oldemptyset\emptyset

\newcommand{\myparagraph}[1]{\vspace{0.2cm} \noindent\textbf{#1}}

\title{\Large Fully Dynamic $(\Delta+1)$-Coloring Against Adaptive Adversaries}
\author{Soheil Behnezhad\thanks{Northeastern University.  Email: {\tt s.behnezhad@northeastern.edu}.  Partially supported by NSF CAREER Award CCF-2442812 and a Google Faculty Research Award.}
\and Rajmohan Rajaraman\thanks{Northeastern University.  Email: {\tt \{r.rajaraman,wasim.o@northeastern.edu\}}.  Partially supported by NSF grant CCF-2335187.} \and Omer Wasim\footnotemark[2]}

\date{}
\begin{document}
\maketitle
\begin{abstract} 

Over the years, there has been extensive work on fully dynamic algorithms for classic graph problems that admit greedy solutions. Examples include $(\Delta+1)$ vertex coloring, maximal independent set, and maximal matching. For all three problems, there are randomized algorithms that maintain a valid solution after each edge insertion or deletion to the $n$-vertex graph by spending $\polylog n$ time, provided that the adversary is \textbf{oblivious}. However, none of these algorithms work against \textbf{adaptive} adversaries whose updates may depend on the output of the algorithm. In fact, even breaking the trivial bound of $O(n)$ against adaptive adversaries remains open for all three problems. For instance, in the case of $(\Delta+1)$ vertex coloring, the main challenge is that an adaptive adversary can keep inserting edges between vertices of the same color, necessitating a recoloring of one of the endpoints. The trivial algorithm would simply scan all neighbors of one endpoint to find a new available color (which always exists) in $O(n)$ time. 

    \smallskip\smallskip

    In this paper, we break this linear barrier for the $(\Delta+1)$ vertex coloring problem. Our algorithm is randomized, and maintains a valid $(\Delta+1)$ vertex coloring after each edge update by spending $\widetilde{O}(n^{8/9})$  time with high probability. 

    \smallskip\smallskip
    To achieve this result, we build on a powerful sparse-dense decomposition of graphs developed in previous work. While such a decomposition has been applied to several sublinear models, this is its first application in the dynamic setting. A major challenge in applying this framework to our setting is that it relies on maintaining a {\em perfect} matching of a certain graph. While maintaining a perfect matching (conditionally) requires $n^{1-o(1)}$ time per update, we prove several structural properties of this graph (of possible independent interest) to achieve an update-time that is sublinear in $n$.\end{abstract}

\clearpage

\begingroup
\renewcommand{\baselinestretch}{0.4}\normalsize
\setcounter{tocdepth}{3}
\tableofcontents
\thispagestyle{empty}
\endgroup

\clearpage

\section{Introduction}

Over the years, significant research has been dedicated to developing {\em fully dynamic} algorithms for graph problems such as $(\Delta+1)$ vertex coloring \cite{dynamicalgsforgraphcol,bhattacharya2022fully,HenzingerP22}, maximal independent set (MIS) \cite{behnezhad2019fully,chechik2019fully}, and maximal matching \cite{baswana2015fully,solomon2016fully,behnezhad2019fully}. In a static setting, all three problems can be solved trivially via simple greedy  algorithms that take linear time. But the situation is very different when it comes to fully dynamic graphs. Let us first formalize the model. A fully dynamic graph undergoes a sequence of edge insertions and deletions. The goal is to maintain a desired graph property while optimizing the {\em update time}, i.e., the time spent per update.

When the adversary is {\em oblivious}---meaning that the sequence of updates is independent of our output---an extensive body of work over the last decade has led to randomized algorithms that only spend $(\polylog n)$ amortized time per update for all three problems 
\cite{dynamicalgsforgraphcol,bhattacharya2022fully,HenzingerP22, behnezhad2019fully,chechik2019fully, baswana2015fully,solomon2016fully}, where $n$ is the number of vertices in the graph. In fact, for $(\Delta+1)$ vertex coloring, the independent works of Bhattacharya, Grandoni, Kulkarni, Liu, and 
Solomon \cite{bhattacharya2022fully} and Henzinger and Peng \cite{HenzingerP22} achieve constant amortized update-time.

In sharp contrast, when it comes to {\em adaptive} adversaries, no algorithm with sublinear in $n$ update-time is known for either of the three problems above. This is unfortunate, since in many applications of dynamic graph algorithms (such as the growing body of work on fast static graph algorithms \cite{ChenKLPGS22,ChuzhoyK24}) the output of the dynamic algorithm influences its future updates.

While the significant gap between algorithms for adaptive and oblivious adversaries may seem surprising at first, there is an intuitive explanation for it. Take the $(\Delta + 1)$ vertex coloring problem for example. In this problem, the goal is to maintain an assignment of colors from $\{1, \ldots, \Delta+1\}$ to the vertices of a graph of maximum degree $\Delta$  such that any two adjacent vertices receive different colors. Since the number of colors is greater than than the maximum degree, any vertex has one available color to choose no matter how its neighbors are colored. Thus, in a static setting, we can go over vertices one by one and color them greedily. In a dynamic setting, the challenging updates are edge insertions between vertices of the same color, as these require at least one endpoint to be recolored. If the vertex colors are sufficiently random, an oblivious adversary's edge insertions will likely involve vertices of different colors, in which case no recoloring is necessary. However, an adaptive adversary, aware of the current coloring, can deliberately insert edges between vertices of the same color, forcing a recoloring of one endpoint. The trivial algorithm here scans all neighbors of one endpoint of the edge inserted to find a valid new color to assign to it in $\Theta(\Delta)$ time which can be as large as $\Theta(n)$. Breaking the $\Theta(n)$ bound remains open not just for $(\Delta+1)$ vertex coloring, but also for maximal matching and MIS, raising the following natural question:

\begin{quote}
    {\em Is $\Omega(n)$ time needed for maintaining a $(\Delta+1)$ vertex coloring, MIS, or maximal matching of a fully dynamic graph against adaptive adversaries?}\footnote{We refer interested readers to Saranurak's talk where this problem (along with many others) is mentioned \url{https://www.youtube.com/live/1cAv-A6EbZE?si=kGYvFHdjXfljZG8Y&t=2579}.}
\end{quote}

\noindent \textbf{Our Contribution:} In this work, we  show that the trivial $O(n)$ update-time algorithm is {\em not} optimal at least for the $(\Delta + 1)$ vertex coloring problem by proving the following theorem.

\begin{tcolorbox}
\begin{restatable}{theorem}{mainthm}
\label{thm:mainthm}
    There exists a fully dynamic algorithm for $(\Delta+1)$-coloring taking amortized update time of $\ti(n^{8/9})$ w.h.p., against an adaptive adversary.
\end{restatable}
\end{tcolorbox}

To prove Theorem~\ref{thm:mainthm}, we build on a powerful sparse-dense decomposition which has its root in the seminal work of Reed \cite{Reed98} (see Chapter~15 of the book by Molloy and Reed \cite{molloy2002graph}). While this decomposition has been successfully applied to various sublinear models over the recent years (see in particular the celebrated work of Assadi, Chen, and Khanna \cite{ack} as well as the works of \cite{hss,junlipettie}) this is its first application in the dynamic setting.\footnote{Independently and concurrent to our work, Braverman et al. \cite{correlation-clustering-arXiv} also develop a dynamic algorithm for a variant of sparse-dense decompositions tailored to the correlation clustering problem. The problems studied in the two papers, and consequently the techniques employed, are completely disjoint.} 

The key challenge in applying this framework to the dynamic setting is that it relies on computing {\em perfect} matchings of a certain auxiliary graph with vertices on one side and available colors on the other side. In our setting, this graph undergoes both edge and vertex updates. In general, maintaining a perfect matching even under just edge updates conditionally requires near linear in $n$ update-time  \cite{AbboudW14,HenzingerKNS15,Dahlgaard16}. We get around this by proving several properties of this auxiliary graph and showing, essentially, that any non-maximum matching in it can be augmented by an augmenting path of length at most five that can be identified in sublinear time. We provide a more detailed overview of our algorithm and these structural properties in Section~\ref{sec:techniques}.

\section{Technical Overview}\label{sec:techniques}

In this section, we give an overview of our algorithm for Theorem~\ref{thm:mainthm}. We emphasize that our discussion here over-simplifies many parts of the algorithm and the technical challenges that arise along the way to provide a high level intuition about our approach. The formal proofs and arguments are provided in subsequent sections. 

As discussed, our algorithm builds on a sparse-dense decomposition technique. While this decomposition has become a standard tool for the $(\Delta+1)$ vertex coloring problem across various settings \cite{hss,ack,junlipettie}, this is its first application for coloring in the dynamic setting. We first provide some background about this decomposition and why it is useful for $(\Delta+1)$ vertex coloring in Section~\ref{sec:tech-decomposition}. We then focus on the challenges that arise with dynamic graphs and discuss how we color the sparse and dense parts of the graph in Sections~\ref{sec:tech-sparse} and \ref{sec:tech-dense} respectively.

\vspace{-0.3cm}
\subsection{Background on Sparse-Dense Decompositions for \texorpdfstring{$(\Delta+1)$}{D+1} Vertex Coloring}\label{sec:tech-decomposition}

Given a graph $G=(V, E)$ of maximum degree $\Delta$ and a parameter $\epsilon > 0$, one can always decompose the vertices of $G$ into {\em sparse vertices} $V_S$ and {\em dense vertices} $V_D$ with $V_D$ being further partitioned into vertex disjoint subsets $C_1, \ldots, C_k$ such that:
\begin{itemize}[leftmargin=15pt, itemsep=0pt, topsep=0pt]
    \item Each $C_i$ is an almost-clique of size $(1 \pm \epsilon) \Delta$ that's nearly disjoint from the rest of the graph. Namely, each vertex $v \in C_i$ has at most $\epsilon \Delta$ non-neighbors in $C_i$, and at most $\epsilon \Delta$ neighbors in $V \setminus C_i$.
    \item For each sparse vertex $v \in V_S$, the total number of edges between its neighbors is at most $(1-\epsilon^2)\binom{\Delta}{2}$. In words, the neighborhood of a sparse vertex is $\epsilon^2$-far from a $\Delta$-clique.
\end{itemize}

% \begin{wrapfigure}[11]{r}{0.45\textwidth}
%     \centering
%     \resizebox{0.4\textwidth}{!}{%
%         \input{figs/clique-decomposition}
%     }
% \end{wrapfigure}

\begin{figure}[h]
    \centering
    \resizebox{0.4\textwidth}{!}{%
        \input{figs/clique-decomposition}
    }
    \caption{An example of sparse-dense decomposition.}
    \label{fig:decomp}
\end{figure}
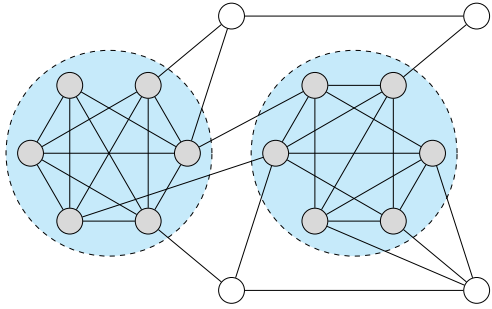

That the decomposition always exists is not hard to prove, and the argument is also constructive and efficient. Figure~\ref{fig:decomp} is an illustration of the sparse-dense decomposition, with the blue clusters being the almost-cliques and the white vertices being the sparse ones.

Now how is this decomposition helpful for $(\Delta+1)$ vertex coloring? The challenge with $(\Delta+1)$ vertex coloring is that given a partial coloring of the graph, an uncolored vertex may only have one available color to choose. If instead we had some $\Omega(\epsilon \Delta)$ excess colors for all uncolored vertices, then the problem would have been much simpler. For sparse vertices, as we will soon discuss in Section~\ref{sec:tech-sparse}, it is relatively straightforward to create enough excess colors. Doing so for dense vertices is generally impossible (e.g. if $C_i$ is an almost-clique of size $\Delta+1$, the last vertex to be colored has exactly one color available). But given that the almost-cliques are highly structured, depending on the setting, they usually have other nice properties to use.

\vspace{-0.3cm}
\subsection{Our Dynamic Coloring: Basic Invariants and Parameters}\label{sec:tech-basic}

Note that the trivial $(\Delta+1)$ coloring algorithm which upon insertion of an edge checks all neighbors of one endpoint to find a new feasible color takes $\Theta(\Delta)$ time per update. So if $\Delta \leq n^{8/9}$, the trivial algorithm yields Theorem~\ref{thm:mainthm}. Thus, we assume for the rest of the section that $\Delta > n^{8/9}$.

Instead of an arbitrary $(\Delta+1)$ vertex coloring, we will ensure at any point during the execution of our dynamic algorithm, that {\em every} color $c \in \{1, \ldots, \Delta+1\}$ is assigned to at most $\widetilde{O}(n/\Delta)$ vertices in the graph. The advantage of maintaining this property is that we can check if a color $c$ is available for a vertex $v$ in time $\widetilde{O}(n/\Delta)$ by going over the list of vertices of color $c$ and checking if any of them is a neighbor of $v$. Note that this is much faster than scanning all neighbors of $v$ in $\Theta(\Delta)$ time given our assumption that $\Delta \geq n^{8/9}$.

We also point out that, crucially, in the sparse-dense decomposition that we will employ, the parameter $\epsilon$ is sub-constant (particularly $\epsilon = 1/n^{c}$ for some sufficiently small constant $c > 0$). With these parameters, running times of say  $O(\epsilon n)$, $O(\frac{n^2}{\epsilon^2 \Delta^2})$, or $O(1/\epsilon^4)$ will all be $n^{1-\Omega(1)}$. We will use these facts in our forthcoming discussions.

\vspace{-0.3cm}
\subsection{Coloring Sparse Vertices Dynamically}\label{sec:tech-sparse}

% As discussed earlier, one can create some $\Omega(\epsilon^2 \Delta)$ excess colors for sparse vertices in a simple way. 

As discussed earlier, the nice property of sparse vertices is that we can create many excess colors for them. For static algorithms, the standard approach is to run a \textsc{One-Shot-Coloring} algorithm (see Algorithm~\ref{alg:one-shot}) which samples a random color $c_v$ for each vertex $v$ and assigns it to $v$ if no other neighbor of $v$ picks the same color. It is not hard to show that a constant fraction of vertices will be colored and, in fact, many neighbors of each sparse vertex will be assigned identical colors, guaranteeing some $\Omega(\epsilon^2 \Delta)$ excess colors for each sparse vertex.

Suppose that we maintain this property dynamically. Namely, that at any point during the execution of our dynamic algorithm, every sparse vertex has at least $\Omega(\epsilon^2 \Delta)$ excess colors. Now when the time comes to re-color a sparse vertex $v$, we keep sampling a new random color for $v$. The probability that we hit one of those excess colors with each sample is $\frac{\Omega(\epsilon^2 \Delta)}{\Delta+1} = \Omega(\epsilon^2)$. So after some $O(1/\epsilon^2)$ samples we expect to find a feasible color for $v$. Given that checking feasibility can be done in $\widetilde{O}(n/\Delta)$ time based on our discussion of Section~\ref{sec:tech-basic}, it would take $\widetilde{O}(\frac{n}{\epsilon^2\Delta}) = n^{1-\Omega(1)}$ time to recolor sparse vertices.

From our discussion above, it remains to show that we can maintain $\Omega(\epsilon^2 \Delta)$ excess colors for sparse vertices dynamically. Note that once we run \textsc{One-Shot-Coloring}, the adaptive adversary sees which colors are used by more than one neighbor of a sparse vertex $v$, and so can enforce us to recolor these neighbors, reducing the number of excess colors for $v$. To avoid this, we re-run \textsc{One-Shot-Coloring} from scratch after every $t = \Theta(\epsilon^2 \Delta)$ edge updates. Since each edge update only reduces the number of excess colors of each vertex by one (as we only recolor one vertex), all sparse vertices will retain $\Omega(\epsilon^2 \Delta) - t = \Omega(\epsilon^2 \Delta)$ excess colors throughout the phase. We will show that \textsc{One-Shot-Coloring} can be implemented in $\widetilde{O}(n^2/\Delta)$ time, thus the amortized update time of running \textsc{One-Shot-Coloring} will be $\widetilde{O}(n^2/\Delta)/t = \widetilde{O}(\frac{n^2}{\epsilon^2\Delta^2}) = n^{1-\Omega(1)}$.

Finally, we note that \textsc{One-Shot-Coloring} does not color all sparse vertices, but only a constant fraction of them. The remaining sparse vertices at the beginning of a phase can be colored by iteratively sampling a color for them until hitting one of their $\Omega(\epsilon^2 \Delta)$ excess colors. But an important technical difficulty arises here: if each sparse vertex picks its color from among $\Theta(\epsilon^2 \Delta)$ colors uniformly, then a single color $c$ (which say is among the excess colors of all $n$ vertices) will be used some $\Omega(\frac{n}{\epsilon^2 \Delta})$ times. While this may not seem very different from our desired upper bound of $\widetilde{O}(n/\Delta)$ discussed in Section~\ref{sec:tech-basic} on the total number of vertices assigned the same color, the problem is the sub-constant $\epsilon^2$ factor in the denominator (we will get back to this when discussing dense vertices). To get around this, we utilize an elegant greedy algorithm for coloring the remaining sparse vertices due to Assadi and Yazdanyar~\cite{assadi-talk} and prove that this algorithm guarantees only $\widetilde{O}(n/\Delta)$ vertices will be assigned the same color, without any dependency on $\epsilon$ (see Lemma~\ref{lemma: colorsparse}).

\subsection{Coloring Dense Vertices Dynamically}\label{sec:tech-dense}

Our algorithm maintains a proper coloring $c$ for all sparse vertices \textit{independently} of dense vertices, i.e. when feasibility of any color is checked for a sparse vertex $v$,  our algorithm ignores its dense neighbors. In a {\em static} setting, Assadi, Chen, and Khanna \cite{ack} showed that this can be done for an almost-clique $C_i$ in the following two-step process:
\begin{description}[itemsep=0pt, topsep=5pt, leftmargin=40pt]
    \item[Step I] Find a $\Theta(1)$-approximate maximum matching $\mathcal{M}_N$ of the {\em non-edges} inside $C_i$ and assign the same colors to the endpoints of each matched pair. This ``saves'' on the colors needed to color $C_i$ as it colors the vertices of the matching at a rate of two vertices per color.
    \item[Step II] This step is responsible for coloring vertices in $C_i$ which are not colored in Step I. Let $\mathcal{L}$ and $\mathcal{R}$ denote the set of uncolored vertices and colors not used to color endpoints of any non-edge in $C_i$ respectively. Construct a bipartite graph $H$ between $\mathcal{L}$ and $\mathcal{R}$, adding an edge between $v \in \mathcal{L}$ and $c \in \mathcal{R}$ if color $c$ is not used by any neighbor of $v$ outside $\mathcal{L}$. Now a {\em perfect} matching from $\mathcal{L}$ to $\mathcal{R}$ corresponds to a valid coloring of the whole clique. We note that existence of this perfect matching is not trivial, and follows from the nice properties of dense vertices.
\end{description}

To dynamize Step I, we utilize existing fully dynamic algorithms for constant approximate maximum matching against adaptive adversaries \cite{bhattacharya2016new}. As these algorithms require adjacency list access to the edge-set of the graph, we have to maintain the list of non-edges inside the almost-cliques explicitly; an issue we will come back to when discussing dynamic maintenance of the decomposition. Another challenge is that we have to be careful about the {\em adjustment complexity} of this matching, as each change to it may correspond to a recoloring of the rest of the almost-clique in Step II. We provide more details about dynamizing Step I in Section~\ref{sec: non-edge-matching}.

Dynamizing Step II is the most technically challenging part of our paper. The first immediate challenge is that maintaining a {\em perfect} matching, even against oblivious adversaries,  conditionally requires linear time in the number of vertices \cite{AbboudW14,HenzingerKNS15,Dahlgaard16}. Since there are $\Theta(\Delta)$ vertices in graph $H$, this would correspond to an update-time of, say, near-linear in $\Delta$ which is not sublinear in $n$. To get around this major challenge, we rely heavily on the structure of dense vertices and (essentially) prove that any non-perfect matching in $H$ admits an {\em augmenting path} of length at most 5 which can additionally be found fast in sublinear time. We end up using very different arguments for ``large'' almost-cliques of size at least $\Delta+1$, and ``small'' almost-cliques of size at most $\Delta$.

\myparagraph{Large Almost-Cliques.} Take an almost-clique $C_i$ of size $|C_i| \geq \Delta+1$. Note that if a vertex $v \in C_i$ has $d$ edges to $V \setminus C_i$, then it must have at least $d$ non-edges inside $C_i$. This property crucially relies on the almost-clique having size at least $\Delta+1$ and does not hold for smaller almost-cliques. Intuitively, it helps us because having more edges outside $C_i$ imposes more feasibility constraints for vertices of $C_i$, but on the other hand, implies a large number of non-edges inside and so the non-edge matching in $C_i$ will be more effective in saving colors.

For each almost-clique $C_i$ we classify colors into \textit{heavy} and \textit{light}: a color $c\in [\Delta+1]$ is heavy for $C_i$ if some $\Omega(\Delta)$ edges exist which have exactly one endpoint in $C_i$ and their other endpoint is colored $c$.  Intuitively, heavy colors are infeasible for a large number of vertices. Our analysis which builds on the intuition above about the non-edge matching reveals that even after excluding the heavy colors $\mathcal{H}$ from $\mathcal{R}$, vertices in $\mathcal{L}$ can still be colored by $\mathcal{R} \setminus \mathcal{H}$ perfectly. The nice thing about colors in $\mathcal{R} \setminus \mathcal{H}$ is that they are feasible for most vertices in $C_i$ and so it is easier to recolor vertices with them. 
\begin{wrapfigure}{r}{0.33\textwidth}
    \centering
    \vspace{-0.2cm}
    \resizebox{0.34\textwidth}{!}{%
        \input{figs/large-clique}
    }
\end{wrapfigure}
Consider an uncolored vertex $v\in \mathcal{L}$.  We pick an unassigned color $c$ in $\mathcal{R} \setminus \mathcal{H}$. If $c$ is feasible for $v$, we immediately assign $c(v) \gets c$. If not, we find a length 3 augmenting path as follows. First, we pick a random vertex $w\in \mathcal{L}$. Since $c$ is not heavy, it is feasible for $w$ with large constant probability. Additionally, since $w$ is also assigned a color that is not heavy, $c(w)$ must be feasible for $v$ with large constant probability. Putting the two together, we can assign  $c(v)\gets c(w)$ and $c(w) \gets c$ corresponding to a length three augmenting path illustrated on the right.

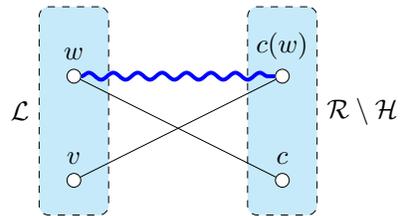
\myparagraph{Small Almost-Cliques.}  To color a vertex $v$ in an almost-clique $C_i$ which has size at most $\Delta$, we proceed as follows. Let $|C_i| = \Delta + 1 -k$ noting that $k \geq 1$ for small almost-cliques. If the size of the non-edge matching is large, we employ a similar approach as in the case of large almost-cliques. 

For the case of small non-edge matchings, we analyze the number of available colors not assigned to any vertex in $C_i$, and bound the number of edges incident to any subset $D\subseteq C_i$ to vertices outside $C_i$. Based on this, we call a color \textit{good} if $\Omega(\Delta)$ vertices in $\mathcal{L}$ do not have a neighbor outside $C_i$ colored $c$. Effectively, we show that a large majority of vertices in $\mathcal{L}$ are assigned good colors at any point. Thus, a random vertex is likely to be assigned a good color. Our algorithm does not maintain the set of good colors as such, and they are only used in the analysis.

To recolor a vertex $v$, we first pick a random color $c$ from the colors not assigned to any vertex of $C_i$ (we maintain this set explicitly). If $c$ is available for $v$ we can assign it to $v$ immediately, but $c$ may not be available for $v$ due to sparse neighbors of $v$. To get around this, we take a random vertex $u\in \mathcal{L}$. We show that $c$ is feasible for $u$ with probability at least $\Omega(1/k)$. Thus, repeating this process $\widetilde{O}(k) = \widetilde{O}(\epsilon \Delta)$ times ensures that a vertex $u$ for which $c$ is feasible is found. 
\begin{wrapfigure}{r}{0.3\textwidth}
    \centering
    \resizebox{0.3\textwidth}{!}{%
        \input{figs/small-cliques}
    }
\end{wrapfigure}
Now if the color $c(u)$ assigned to $u$ is feasible for $v$, we can assign $c$ to $u$ and $c(u)$ to $v$ (corresponding to a length three augmenting path). However, $c(u)$ may not be feasible to $v$ again due to its sparse neighbors. Nonetheless, we show that $u$ is assigned a good color $c(u)$ with constant probability, so it is feasible for a constant fraction of vertices in $C_i$. 
To utilize this, we pick yet another random vertex $w\in \mathcal{L}$. It so happens that with large constant probability $c(w)$ must be feasible for $v$ and $c(u)$ must be feasible for $w$ (due to $c(u)$ being good). This results in a length five augmenting path illustrated on the right that can be used to color $v$.

The run-time of process above is dominated by the time required to find $u$. We need $\widetilde{O}(k)$ random candidates for $u$ and we need to check feasibility of $c$ for each candidate which takes $\widetilde{O}(k \cdot \frac{n}{\Delta})$ time overall. Since $k \leq \epsilon \Delta$ as the size of almost-cliques is at least $(1-\epsilon)\Delta$, this sums up to $\widetilde{O}(\epsilon n)$ which is sublinear. Note that this crucially relies on checking feasibility in $\widetilde{O}(n/\Delta)$ time and not say $\widetilde{O}(n/\epsilon^2 \Delta)$ or even $\widetilde{O}(n/\epsilon \Delta)$ time. This is why we needed the strong upper bound of $\widetilde{O}(n/\Delta)$ on the number of vertices assigned the same color for sparse vertices in Section~\ref{sec:tech-sparse}.

\vspace{-0.3cm}
\subsection{Putting Everything Together}\label{sec:tech-wrap-up}

We note that our discussion above completely ignores maintenance of the sparse-dense decomposition in a dynamic graph. In fact, as apparent from our discussion above, in addition to the decomposition itself, we need to maintain the non-edges inside almost-cliques and guarantee that the almost-cliques have size close to $\Delta$ at all times. Our algorithm maintains this decomposition along with our desired properties in $O(\log n/\epsilon^4)$ time per update. We provide a high-level overview of our dynamic sparse-dense decomposition in Section~\ref{sec: tech-decomposition}. 

Finally, by balancing $\epsilon$ as a function of $\Delta$ and $n$ we arrive at our update-time of $\widetilde{O}(n^{8/9})$ (see Section~\ref{sec:wrap-up} for the final update-time as a function of $n, \Delta, \epsilon$ for various parts of the algorithm).

\section{Preliminaries} 
Let $V = [n]=\{1,2,...,n\}$ denote the fixed set of vertices, where $[k]=\{1,2,..,k\}$ for any $k\in \mathbb{Z}^{>0}$. For any vertex $v\in V$, let $N(v)$ denote the set of neighbors of $v$ at any given point during the algorithm, and $d(v)=|N(v)|$ be the degree of $v$. Let $E(U,W)=\{(u,w) \mid u\in U, w\in W\}$ denote the set of edges with one endpoint in $U$ and the other in $W$. Let $E(U)$ denote the set of edges with both endpoints in $U$. 
A pair of vertices $(u,v)$ is called a \textit{non-edge} if $(u,v)\notin E$. A \textit{non-neighbor} of any vertex $v$ is a vertex $u$ such that $(u,v)\notin E$. For a vertex $v$, and a set $U$ of vertices, we let $N_U(v)$ denote the set of neighbors of $v$ in $U$.

Let $\Delta$ denote the maximum degree of any vertex in the graph throughout the update sequence, which is known in advance. The coloring on $V$ maintained by our algorithm is denoted by $c:\,V\rightarrow [\Delta+1]$. Throughout the paper, we use the $\ti(\cdot)$ notation to hide a multiplicative $\text{polylog}(n)$ term. Without loss of generality, we assume that initially, $G$ is an empty graph on $n$ vertices (i.e, $E=\oldemptyset$ at the beginning of the algorithm). We say that an event $E$ happens \textit{with high probability} (abbreviated as w.h.p.) if $\Pr[E]\geq 1-\frac{1}{\text{poly}(n)}$. 

By choosing large constants during various sampling procedures and algorithms and, recomputing every polynomially many (in $n$) updates yields that all desired events happen w.h.p. We do not stress on the number of \textit{undesirable} events or low-order constants arising in probabilistic analysis, which allows us to stress more on the core technical ideas in this paper. Occasionally in our proofs, when deriving a high probability bound for an event $E$, we implicitly condition on a small number of desirable events (which happen with high probability); by taking a union bound over the failure probability of such desirable events, we are able to claim a high probability bound for the event $E$. 

\subsection{Organization of the Paper}
We give a high level technical overview of our paper in Section \ref{sec:techniques}. In Section \ref{sec: fullydynamicsparsedensedecomp}, we give a technical overview of sparse-dense decompositions and state properties of our \textit{fully dynamic} decomposition. The algorithm to maintain the decomposition and the analysis are deferred to Section \ref{sec: fdalgorithmdecomposition}, which can be read independently of other sections.

In Section \ref{sec: fullydynamiccoloring}, we present our fully dynamic algorithm to maintain a $(\Delta+1)$-coloring. This is divided into multiple subsections. We first highlight our approach of working with a fixed sparse-dense decomposition for a phase of updates in Section \ref{sec: constadjcomplexity}. We present our algorithm for coloring sparse vertices in Section \ref{sec: staticalgsparse}, and our fully dynamic algorithm to maintain a coloring on sparse vertices in Section \ref{sec: fdalgsparse}. 

Section \ref{sec: coloringdensevertices} describes our approach to color dense vertices. In Section \ref{sec: non-edge-matching}, we show how to compute and maintain non-edge matching in almost-cliques, together with maintaining a proper coloring. In Sections \ref{sec: large almost cliques} and \ref{sec: small almost cliques}, we give our perfect matching algorithms for small and large almost-cliques respectively. We present our final algorithm in Section \ref{sec:wrap-up} which unifies various algorithms and approaches we develop in the preceding sections to obtain our final sublinear in $n$ update time.
% \subsection{Related Work}

\input{decomposition}
\section{A Fully Dynamic \texorpdfstring{$(\Delta+1)$}{D+1}-Coloring Algorithm}\label{sec: fullydynamiccoloring}
In this section, we describe our $(\Delta+1)$-coloring algorithm, given the sparse-dense decomposition in Theorem \ref{fd-decomposition}. First, we describe an algorithm which maintains a $(\Delta+1)$-coloring on the subgraph induced on the set of sparse vertices $V_S$. Given the coloring on sparse vertices, we show how to maintain a $(\Delta+1)$ coloring on dense vertices in $V_D$.

To obtain a sublinear (in $n$) update time, our algorithms work with a fixed sparse-dense decomposition for every $t=\Omega(\eps^2\Delta)$ updates. While we maintain data structures such as non-edges in almost-cliques, the set of sparse and dense vertices $V_S$ and $V_D$ respectively remain fixed throughout these $t$ updates. After every $t$ updates, we utilize our fully dynamic algorithm \textsc{Update-Decomposition} to update the sparse-dense decomposition. The next section formalizes our approach.

\subsection{Working with a Fixed Decomposition for \texorpdfstring{$\Theta(\eps^2\Delta)$}{e2D} updates}\label{sec: constadjcomplexity} 
Our approach is simple. Within $t=\frac{\varepsilon^2\Delta}{18e^6}$ updates, the decomposition of $G$ remains fixed. More precisely, we consider the sequence of updates partitioned into contiguous subsequences of $t$ updates each; each such subsequence is called a \textit{phase}. We denote the $j^{th}$ phase by $\mathcal{U}_j$. 

For every phase $\mathcal{U}_j$, we maintain a $(\Delta+1)$ coloring of $G$ such that the decomposition of $V$ into $V_S$ and $V_D$, together with all almost-cliques $C_j$ for $i\in [\ell]$ remains fixed. Edge updates during phase $\mathcal{U}_j$ are reflected in data structures, including $N_S(\cdot), N_D(\cdot),$ $\b{E_C(\cdot)}, N_C(\cdot), \b{E(C)}$ for all $C$ (see Section \ref{sec: fullydynamicdecomposition}), which takes $O(1)$ update time. However, no vertex moves from $V_S$ to $V_D$ or vice-versa, and every dense vertex stays in the same almost-clique throughout $\mathcal{U}_j$. Since vertices do not move from $V_S$ to $V_D$ or vice versa, it holds that for any almost-clique $C_i$, the set of non-edges $\b{E(C_i)}$ changes by at most one per update.

Before the next phase $\mathcal{U}_{j+1}$ begins, all updates to the aforementioned data structures are undone. Then, \textsc{Update-Decomposition} is invoked on each update $(u,v)\in \mathcal{U}_j$. This ensures that our algorithm works with an updated decomposition after each phase. Thereafter, all vertices in $G$ are recolored from scratch given the updated decomposition.

The following lemma shows that the resulting properties of the sparse-dense decomposition throughout a phase by our algorithm are \textit{essentially} similar to the ones guaranteed by Theorem \ref{fd-decomposition}.

\begin{lemma}\label{lemma: fd-decompositionphase}
    For any constant $0< \varepsilon<\frac{3}{50}$, there exists a sparse-dense decomposition for a graph $G=(V,E)$ undergoing edge insertions or deletions which can be maintained in $ O(\frac{\log n}{\varepsilon^4})$ amortized update time against an adaptive adversary, such that throughout every phase of $t=\frac{\eps^2\Delta}{18e^6}=\Theta(\eps^2\Delta)$ updates to $G$ the following properties hold, w.h.p.: 
    \begin{enumerate}
    \item the decomposition of $V=(V_S, V_D)$ and $V_D=(C_1, C_2,..., C_{\ell})$ remains fixed, where $V_S\subseteq \vs_{\varepsilon/3}$ and $V_D\subseteq \vd_{15\eps/4}$. 
    \item For all $i\in [\ell]$: 
        \begin{itemize}
        \item $(1-4\varepsilon)\Delta\leq |C_i|\leq (1+10\varepsilon)\Delta$.
        \item Each vertex in $C_i$ has at least $(1-5\varepsilon)\Delta$ neighbors in $C_i$.
        \item $\b{E(C_i)}$ changes by at most 1 per update.
    \end{itemize}
    \end{enumerate}
\end{lemma}

\begin{proof}
    Before a new phase begins, $\textsc{Update-Decomposition}(\eps, \frac{\eps}{3}, (u,v))$ is invoked for all updates in the current phase. By the proof of Theorem \ref{fd-decomposition}, the amortized update time for this step is $O(\frac{\log n}{\eps^4})$.
    
    Clearly, the properties stated in Theorem \ref{fd-decomposition} hold before the beginning of any phase. Since $C_i$ for all $i\in [\ell]$ remains fixed throughout phase $\mathcal{U}_j$, the bounds on $|C_i|$ don't change. Since a phase  consists of most $t$ updates, any vertex in $C_i$, $i\in [\ell]$ has at least $(1-4\eps-\frac{\eps^2}{18e^6})\Delta\geq (1-5\eps)\Delta$ neighbors in $C_i$ throughout the phase. Thus, any vertex in $C_i$, $i\in [\ell]$ has at most $5\eps\Delta$ neighbors outside $C_i$. 

    Any update changes $\b{E(C_i)}$ by at most 1. By Theorem \ref{fd-decomposition} it follows that after \textsc{Update-Decomposition} is invoked for all updates in the previous phase, for any vertex $v\in V_S$, $v$ is $\frac{4\eps}{3}$-sparse. Since $t\leq \Delta$, any vertex $v\in V_S$ is $\frac{\eps}{3}$ sparse throughout a phase. Similarly, it holds that any vertex $v\in V_D$ is $\frac{15\eps}{4}$-dense.
\end{proof}

 We note that the number of almost-cliques is $O(\frac{n}{(1-4\eps)\Delta})=O(\frac{n}{\Delta})$ for $0<\eps<1$. In the next two sections, we present algorithms to maintain a $(\Delta+1)$ coloring on the sets $V_S$ and $V_D$, respectively.

\subsection{An Algorithm for Sparse Vertices}
In this section, we present a fully dynamic algorithm to maintain a $(\Delta+1)$-coloring in $\ti(\frac{n}{\eps^2\Delta})$ amortized update time on the subgraph $G_S=(V_{S}, E(V_{S}))$ induced on the set of sparse vertices. For this section, let $V_{S}$ denote the set of $\varepsilon$-sparse vertices and $E(V_{S})$ denotes edges in $G$ with both endpoints in $V_{S}$ in a given phase. We assume that $0<\varepsilon<1$.

Before presenting the fully dynamic algorithm, we give a static algorithm to obtain a $(\Delta+1)$-coloring on sparse vertices, at the beginning of a phase. 

\subsubsection{Recoloring Sparse Vertices at the Beginning of a Phase}\label{sec: staticalgsparse}
We exploit sparsity of vertices in a manner similar to $(\Delta+1)$-coloring algorithms in sublinear and distributed settings \cite{hss, ack, junlipettie}. Recall that any $\varepsilon$-sparse vertex has at most $(1-\varepsilon)\Delta$, $\varepsilon$-friends by definition, and thus, the number of non-edges $(u,w)$ s.t. $u,w\in N(v)$ is at least $\Omega(\varepsilon^2\Delta^2)$. Let $V_S$ denote the set of $\eps$-sparse vertices in $V$.

Sparse vertices are colored using two simple subroutines, which we call \textsc{One-Shot-Coloring} and \textsc{Greedy-Coloring} respectively. We maintain a list $L(c)$ for each colors $c\in [\Delta+1]$ which contains all vertices assigned color $c$. Initially, $L(c)=\emptyset$ for all $c\in [\Delta+1]$ and all vertices are assigned a blank color $\perp$.

\noindent\textbf{One Shot Coloring.} The subroutine \textsc{One-Shot-Coloring} receives a set $U$ of vertices to be colored and proceeds as follows. First, each vertex $v\in U$ independently picks a single color $c$ uniformly at random from $[\Delta+1]$. For any vertex $v\in U$, if no neighbor of $u\in U$ picks the color $c$ which is picked by $v$, $v$ is assigned color $c$. A constant fraction of vertices in $U$ are assigned a valid color with high probability in this manner. The set $U'\subseteq U$ of vertices successfully colored after $\textsc{One-Shot-Coloring}$ concludes is returned.

\begin{algorithm}[h]
    \caption{\textsc{One-Shot-Coloring}$(U)$}
    \begin{algorithmic}[1]
        \State $U'=\emptyset$.
        \For {all $v\in U$}
            \State Sample a color $c'(v)$ uniformly at random from $[\Delta+1]$.
        \EndFor 
        \For {all $v\in U$}
            \State $c\leftarrow c'(v)$.
            \If {no neighbor of $v$ is in $L(c)$}
                \State $c(v)\leftarrow c$, $L(c)\leftarrow L(c)\cup \{v\}$,  $U'\leftarrow U'\cup \{v\}$.
            \EndIf 
        \EndFor 
    \State \textbf{return} $U'$.
    \end{algorithmic}
    \label{alg:one-shot}
\end{algorithm}

 Let $A(v)$ denote the set of colors not assigned to any other sparse neighbor of $v$ at any point in time, i.e. $A(v)=[\Delta+1]\setminus \{c(u) \mid  \, u\in N(v)\cap V_S\}$. We remark that the set of available colors is \emph{not maintained} by our algorithm, since this can be costly. 

Let $d_r(v)$ denote the set of remaining uncolored sparse neighbors of $v$ after $\textsc{One-Shot-Coloring}$ finishes. We assume that $\varepsilon\leq \frac{1}{5000}$, $\alpha\geq 5000^3$, and $(\Delta+1)>\frac{\alpha\log n}{\varepsilon^2}$. The following lemma can be derived from Lemma 3.1 and Lemma A.1 in \cite{ack}.

\begin{lemma}\label{lemma: one-shot} 
For every uncolored vertex $v\in V\setminus U'$ after \textsc{One-Shot-Coloring}$(U)$ is completed, the number of available colors $|A(v)|\geq d_r(v)+\frac{\varepsilon^2\Delta}{9e^6}$ with probability at least $1-\exp(-\Omega(\varepsilon^2\Delta))$.
\end{lemma}

\noindent\textbf{Greedy Coloring.} The remaining uncolored vertices in $S\coloneq V_S\setminus U'$ are colored in a sequential greedy manner using the \textsc{Greedy-Coloring} subroutine. This subroutine is due to Assadi and Yazdanyar~\cite{assadi-talk} and works as follows. We consider vertices in $S$ in \textit{random} order. Upon considering an uncolored vertex $v$, a random color $c$ is drawn uniformly at random from $[\Delta+1]$. If none of $v$'s neighbors have been assigned $c$, $v$ is assigned color $c$. Else, the previous step is repeated. 

We give the pseudo-code of the subroutine as follows.

\begin{algorithm}[h]
\caption{\textsc{Greedy-Coloring}$(S)$}
    \begin{algorithmic}[1]
        \State Generate a uniformly random permutation $\pi$: $[|S|]\rightarrow [|S|]$.
        \For {$i=1$ to $|S|$}
            \State $v\leftarrow S[\pi(i)]$. \Comment{i.e. $v$ is the $\pi(i)^{th}$ vertex in list $S$}.
            \While{$c(v)=\perp$}
                \State Draw a color $c$ uniformly at random from $[\Delta+1]$.
                \If{no neighbor of $v$ is in $L(c)$}
                    \State $c(v)\leftarrow c$, $L(c)\leftarrow L(c)\cup \{v\}$.
                \EndIf 
            \EndWhile 
        \EndFor 
    \end{algorithmic}
\end{algorithm}

Finally, we give our algorithm $\textsc{Color-Sparse}$ which takes as input the list $V_S$ of sparse vertices and obtains a proper coloring for vertices in $V_S$. The algorithm works as follows. Each vertex $v$ is added to a set $U$ with probability $\frac{1}{2}$ independently. Thereafter, \textsc{One-Shot-Coloring}$(U)$ is invoked. Let $S\coloneq V\setminus U'$ denote the set of uncolored vertices after this step. Thereafter, \textsc{Greedy-Coloring}$(S)$ is invoked.

\begin{algorithm}[h]
\caption{\textsc{Color-Sparse}$(V_S)$}
    \begin{algorithmic}[1]
        \State $U\coloneq \emptyset$.
        \For{all $v\in V_S$}
            \State Add $v$ to $U$ with probability $\frac{1}{2}$.
        \EndFor 
        \State $U'\leftarrow \textsc{One-Shot-Coloring}(U)$.
        \State \textsc{Greedy-Coloring}$(V_S\setminus U')$.
    \end{algorithmic}
\end{algorithm}

 We give some key properties of Algorithm $\textsc{Color-Sparse}$ as follows. We prove that the total number of colors drawn by all vertices throughout the execution of \textsc{Color-Sparse} is $\ti(n)$ with high probability. Based on this, we conclude that the size of list $L(c)$ for any $c\in [\Delta+1]$ is at most $\ti(\frac{n}{\Delta})$ after $\textsc{Color-Sparse}$ is completed. This allows us to bound the total running time of $\textsc{Color-Sparse}$ by $\ti(\frac{n^2}{\Delta})$. Since the proof is a straightforward adaptation of the arguments in \cite{assadi-talk}, we defer it to Appendix~\ref{sec: appendixcolorsparse}.

%\label{cor: excesscolors}
\begin{restatable}{lemma}{colorsparse}
\label{lemma: colorsparse}
     Algorithm \textsc{Color-Sparse} takes $\ti(\frac{n^2}{\Delta})$ time and on termination, $|L(c)|=\ti(\frac{n}{\Delta})$ for any color $c\in [\Delta+1]$ with high probability. Moreover, for every vertex $v\in V$, $|A(v)|\geq \frac{\varepsilon^2\Delta}{9e^6}$ with high probability.
\end{restatable}

In the next section, we present a fully dynamic algorithm to maintain a coloring on sparse vertices during a phase. 

\subsubsection{Recoloring Sparse Vertices During a Phase}\label{sec: fdalgsparse}
Our algorithm in this section is utilized whenever a recoloring is necessitated for a sparse vertex $v$ after an edge update in a given phase. Lemma \ref{lemma: colorsparse} implies that within an arbitrary $t=\frac{\varepsilon^2\Delta}{18e^6}$ \textit{recolorings} of sparse vertices following an invocation of \textsc{Color-Sparse}, the number of available colors for any vertex is at least $t$ with high probability. Thus, during any phase there exist at least $t$ colors for every sparse vertex $v$ which are not used to color any neighbors of $v$ in $V_S$. 

During any phase, our algorithm works as follows. On an edge insertion $(u,v)$ for which $c(u)=c(v)$ and $u,v\in V_S$, $v$ is recolored by invoking subroutine \textsc{Recolor-Sparse} as follows. A color $c$ is sampled from $[\Delta+1]$ uniformly at random and if there exists a sparse neighbor of $v$ assigned $c$, the sampling process is repeated until the sampled color $c$ is not assigned to any of $v's$ sparse neighbors. Thereafter, $v$ is assigned color $c$. We maintain lists $L(c)$ for all $c\in [\Delta+1]$ throughout a phase, such that $L(c)$ consists of vertices colored $c$. 

 We give the pseudo code of our algorithm $\textsc{Recolor-Sparse}$ as follows. A vertex which is currently not assigned a color is assigned a blank color, denoted by $\perp$. Let $N_S(v)$ denotes the list of all neighbors of $v$ in $V_S$.

\begin{algorithm}[h]
\caption{\textsc{Recolor-Sparse$(v)$}}
\begin{algorithmic}[1]
    \State $L(c(v))\leftarrow L(c(v))\backslash \{v\}$, $c(v)\leftarrow \perp$.
    \While{$c(v)= \perp$}
        \State Sample a color $c$ uniformly at random from $[\Delta+1]$.
        \If{$L(c)\cap N_S(v)=\emptyset$}
            \State $c(v)\leftarrow c$.
            \State $L(c)\leftarrow L(c)\cup \{v\}$.
        \EndIf 
    \EndWhile 
\end{algorithmic}
\end{algorithm}
\eat{
\begin{algorithm}[h]
\caption{\textsc{FD-Sparse}$(\varepsilon$)}
\begin{algorithmic}[1]
    \State recolorings $=0$. 
    \For{any edge insertion $(u,v)$}
        \If{$c(u)=c(v)$ and $u\in V_S$}
            \State Call \textsc{Recolor-Sparse}$(v)$.
            \State recolorings $\leftarrow$ recolorings$+1$.
        \EndIf 
        \If {recolorings $=\frac{e^{-6}\varepsilon^2}{2}\Delta$}
            \State Run Algorithm $\textsc{Color-Sparse}$.
            \State recolorings $\leftarrow 0$.
        \EndIf 
    \EndFor
    
\end{algorithmic}
\end{algorithm}
}

By Lemma \ref{lemma: colorsparse}, $\textsc{Color-Sparse}$ satisfies that for any $c\in [\Delta+1]$, $|L(c)|$ is $\ti(\frac{n}{\Delta})$ w.h.p at the beginning of a phase. We show that during any phase, (i.e. between any two calls to $\textsc{Color-Sparse}$), $|L(c)|$ is $\ti(\frac{n}{\Delta})$ w.h.p. Intuitively, this holds since our algorithm picks a random color from the set of available colors for any vertex $v$ whenever $v$ is recolored within a phase. Thus, checking if a color is feasible for a vertex takes $\ti(\frac{n}{\Delta})$ time. 

\begin{lemma}\label{lemma: colorloadsparse}
    During any phase, $|L(c)|=\ti(\frac{n}{\Delta})$ for all $c\in [\Delta+1]$ with high probability.
\end{lemma}

\begin{proof}
By Lemma \ref{lemma: colorsparse}, $|A(v)|\geq \frac{\varepsilon^2\Delta}{9e^6}$ for all $v\in V$ w.h.p. when $\textsc{Color-Sparse}$ is invoked at the beginning of a phase. We analyze $|L(c)|$ for any color $c$ at any point during the phase, conditioned on this high probability event. Moreover, within any phase of length $t=\frac{\varepsilon^2\Delta}{18e^6}$, there can only be $t$ recolorings of sparse vertices.

Fix a color $c$. Let $X_c$ denote the random variable denoting the length of $L(c)$ at the beginning of the phase (i.e. after $\textsc{Color-Sparse}$ is called). By Lemma \ref{lemma: colorsparse}, $X_c=\ti(\frac{n}{\Delta})$ w.h.p. Conditioned on this, let $Y_c$ denote the random variable denoting the length of $L(c)$ after an arbitrary $t$ recolorings of sparse vertices. Moreover, let $Z_1, Z_2,.., Z_t$ be independent random variables where $Z_i=1$ if $c$ is assigned to a sparse vertex at the $i^{th}$ recoloring, and $0$ otherwise. Let $Z=\sum_{i=1}^t Z_i$. We have that $Y_c=X_c+Z$. Let $v_i$ denote the vertex for which $\textsc{Recolor-Sparse}(v_i)$ is invoked at the $i^{th}$ recoloring. 

Note that,
\begin{align*}
    \Pr[Z_i=1]&=\Pr[c \in A(v_i) \cap \text{ }c \text{ is the first color in }A(v_i) \text { sampled}] \\ 
    &\leq \Pr[c \text{ is the first color in }A(v_i) \text { sampled}]\\
    &\leq \frac{18e^6}{\varepsilon^2\Delta}
\end{align*} where the final inequality follows from Lemma \ref{lemma: colorsparse}, and the fact that $i\leq t=\frac{\varepsilon^2\Delta}{18e^6}$. Let $B_1, B_2,...,B_t$ be independent and identically distributed random Poisson random variables s.t. $\Pr[B_i=1]=\frac{18e^6}{\varepsilon^2\Delta}$, and $B=\sum_{i=1}^t B_i$. Note that $B$ stochastically dominates $Z$, i.e. $\Pr[Z\geq a]\leq \Pr[B\geq a]$ for any $a$. Since $E[B]\leq \frac{\varepsilon^2\Delta}{18e^6}\cdot \frac{18e^6}{\varepsilon^2\Delta}=1$, applying a Chernoff bound yields $\Pr[B\geq \sqrt{3d\log n}+1]\leq \frac{1}{n^d}$ for any constant $d>0$. 

Thus, $Y_c=X_c+Z\leq  X_c+B= \ti(\frac{n\log n}{\Delta}) + O(\sqrt{\log n})=\ti(\frac{n}{\Delta})$ with probability at least $1-\frac{1}{n^d}$. Applying a union bound over all $c\in [\Delta]$ and combining this with the high probability bound of Lemma \ref{lemma: colorsparse}, we have that $|L(c)|=\ti(\frac{n}{\Delta})$ for all $c\in [\Delta+1]$ with probability at least $1-\frac{1}{\text{poly}(n)}$ for sufficiently large $d$, completing the proof.
\end{proof}

\begin{lemma}\label{lemma: recolorsparse}
\textsc{Recolor-Sparse} takes $\ti(\frac{n}{\eps^2\Delta})$ with high probability.
\end{lemma}
\begin{proof}
    By Lemma \ref{lemma: colorsparse}, after $\textsc{Color-Sparse}$ is completed, $A(v)\geq \frac{\varepsilon^2\Delta}{18e^6}$ w.h.p. Conditioned on this, within an arbitrary $t$ recolorings of sparse vertices the number of available colors $A(v)$ for any vertex $v$ satisfies $|A(v)|\geq \frac{\varepsilon^2\Delta}{18e^6}$ w.h.p. Moreover, by Lemma \ref{lemma: colorloadsparse}, the length of any list $L(c)$ for $c\in [\Delta]$ is $\ti(\frac{n}{\Delta})$ w.h.p. The probability of sampling an available color in $A(v)$ is at least $\frac{\varepsilon^2}{18e^6}$. Thus after sampling $O(\frac{\log n}{\eps^2})$ colors, an available color $c\in A(v)$ is sampled and assigned to $v$. Concluding, $\textsc{Recolor-Sparse}$ takes $\ti(\frac{n}{\Delta})O(\frac{\log n}{\eps^2})=\ti(\frac{n}{\eps^2 \Delta})$ time with high probability. Taking a union bound over at most $t=\Omega(\eps^2\Delta)$ recolorings of sparse vertices during a phase yields that any call to $\textsc{Recolor-Sparse}$ takes $\ti(\frac{n}{\eps^2\Delta})$ time with high probability.
\end{proof}

\begin{lemma}\label{lemma: amortizedsparserecoloring}
    The amortized update time for recoloring sparse vertices is $\ti(\frac{n^2}{\eps^2\Delta^2})$ with high probability.
\end{lemma}

\begin{proof}
By Lemma \ref{lemma: colorsparse}, the running time of $\textsc{Color-Sparse}$ is $\ti(\frac{n^2}{\Delta})$ with high probability. Thus, the amortized update time incurred as a result of running $\textsc{Color-Sparse}$ at the beginning of any phase is $\ti(\frac{n^2}{\eps^2\Delta^2})$.
\end{proof} 

\subsection{An Algorithm for Dense Vertices}\label{sec: coloringdensevertices}
In this section, we show to maintain a $(\Delta+1)$-coloring on the set of dense vertices $V_D\coloneq V\backslash V_S$. The coloring on sparse vertices is maintained \emph{independently} of dense vertices in the following sense: when coloring a sparse vertex $v$, we only consider sparse neighbors of $v$. As a result, if $v$ is assigned color $c$, this may necessitate recolorings of dense neighbors of $v$ which are assigned color $c$. In this section, we assume that sparse vertices are already colored. Our main result in this section is a fully dynamic algorithm to maintain a $(\Delta+1)$-coloring on a single dense almost-clique $C$ during any phase. Our algorithm follows a two-step framework which we briefly describe as follows. 

Step I ensures that for any almost-clique $C$ which has non-edges, endpoints of such non-edges are assigned the same color whenever possible to \textit{save colors} for Step II. This is accomplished by maintaining a matching on non-edges in $\b{E(C)}$ such that endpoints of \textit{matched} non-edges are assigned the same color. Distinct endpoints of distinct non-edges in the non-edge matching are assigned distinct colors. Thus, each color in $[\Delta+1]$ is used to color exactly two vertices and the number of vertices colored in Step I is twice the number of colors used. After Step I is completed, the number of colors left for the remaining uncolored vertices is shown to be sufficient. These vertices are handled in Step II. We remark that vertices in Step I are colored \textit{independently} of Step II. More precisely, when coloring endpoints $u,v$ of a matched non-edge $(u,v)$, we only check feasibility of a color with respect to neighbors of $u$ and $v$ outside $C$, and endpoints of other matched non-edges.

Step II of our algorithm assigns a color to the remaining vertices in $C$ which are not endpoints of any matched non-edge from Step I. These vertices are handled by maintaining a perfect matching \textit{dynamically} between vertices and the unused colors in $[\Delta+1]$ from Step I. Hence, the coloring of vertices in Step II depends on both the coloring of sparse vertices and the coloring of matched non-edges in Step I. We establish several structural properties of this perfect matching graph to argue that, after vertex updates to this graph, the perfect matching can be quickly updated in $\ti(\eps n)$ time. Effectively, by arguing the existence of small augmenting paths that can be quickly found, we bound the time taken to recolor any vertex under updates to the perfect matching graph.

Finally to ensure that the update time is $\ti(\eps n)=o(n)$ with high probability for any update \textit{during a phase}, it is crucial that the adjustment complexity of our \textit{non-edge matching} is small. If the non-edge matching is significantly large at the beginning of a phase, we do not maintain it, as such during the phase. However if it is small, we use a naive algorithm to maintain a \textit{maximal matching} on non-edges which takes $O(\eps \Delta)$ worst-case update time, and guarantees that the non-edge matching changes by $O(1)$ per update. This ensures a large matching with respect to the number of non-edges $\b{E(C)}$ is maintained at any time during the phase which changes by $O(1)$ and thus limits updates to the perfect matching by $O(1)$ per update. 

We maintain a list $L_D(c)$ for all $c\in [\Delta+1]$, such that $L_D(c)$ contains all dense vertices $v$ in $V_D$ which are assigned color $c$. Our algorithm maintains the property that each color $c\in [\Delta+1]$ is assigned to at most two vertices in any almost-clique $C$. Thus $|L_D(c)|=O(\frac{n}{\Delta})$ for any $c\in [\Delta+1]$. In addition, we have access to lists $L(c)$ which contains sparse vertices assigned color $c$ such that $|L(c)|=\ti(\lst)$ for any $c\in [\Delta+1]$ by Lemma \ref{lemma: colorloadsparse}.

In the next two sections, we present Steps I and II of our algorithm. 

\subsubsection{Step I: Coloring via Non-Edge Matchings}\label{sec: non-edge-matching}
Step I is a pre-processing step for almost-cliques $C$ with the objective of coloring a large number of vertices whenever possible, by assigning identical colors to endpoints of non-edges in $C$. However, non-edges may share endpoints, and thus we maintain a matching $\mathcal{M}_N$ on $\b{E(C)}$, and assign identical colors to endpoints of \textit{matched} non-edges in $\mathcal{M}_N$. Endpoints $u,v$ of every non-edge $(u,v)\in \mathcal{M}_N$ are assigned a color which is not taken by any of $u$ or $v$'s neighbors outside $C$, or endpoints of other matched non-edges in $\mathcal{M}_N$. The reason why endpoints of distinct non-edges are assigned distinct colors is because for edges $(u,v), (w,x)$ in the matching, if $(u,x)\in E$ then assigning identical colors to $u,v,w,x$ does not result in a proper coloring. 

Since we work with a fixed sparse-dense decomposition during a phase, $\b{E(C)}$ changes by at most 1 per update. Note that when a phase concludes, our decomposition may change. The non-edge matching $\mathcal{M}_N$ for an almost-clique $C$ is updated to reflect the updates in this phase, and endpoints of the updated matching are recolored before the next phase begins.

%\noindent\textbf{\underline{Initializing Non Edge Matchings For a New Phase}}\newline
\paragraph{Initializing Non Edge Matchings For a New Phase}
We present our approach to recompute $\mathcal{M}_N$ given the updated decomposition before the beginning of any phase. Although we maintain a non-edge matching $\mathcal{M}_N$ during a phase, the decomposition remains fixed throughout a phase and thus, $\mathcal{M}_N$ may no longer correspond to a valid non-edge matching for the next phase $\mathcal{U}_{i+1}$, since some vertices may move from $V_S$ to $V_D$ or vice versa in the updated decomposition. We show how to initialize $\mathcal{M}_N$ on $\b{E(C)}$ for an almost-clique $C$.

Let $\mathcal{U}_i$ be the phase which has concluded. Recall that a phase consists of $t=O(\eps^2\Delta)$ updates.

\noindent\textbf{Data Structures.} For an almost-clique $C$, we maintain a boolean array, matched$[\cdot]$ such that matched$[v]=1$ if $v$ is an endpoint of a matched edge in the non-edge matching $\mathcal{M}_N$ for $C$ and 0 otherwise at any given point in time. This can be maintained trivially whenever an edge is inserted or deleted from $\mathcal{M}_N$, or $\b{E(C)}$.

Before the next phase $\mathcal{U}_{i+1}$ begins, we do the following. Let $\mathcal{M}_N^F$ denote the current matching on non-edges $\b{E(C)}$ at the end of phase $\mathcal{U}_i$ for an almost-clique $C$. All updates made to $\mathcal{M}_N, \b{E(C)}$ and array matched$[\cdot]$ during $\mathcal{U}_i$ are reversed, and undone to retrieve the non-edge matching $\mathcal{M}_N^I$ and the state of various data structures at the beginning of phase $\mathcal{U}_i$. 

We use a known fully dynamic matching algorithm due to Bhattacharya, Henzinger and Nanongkai \cite{bhattacharya2016new} to recompute a non-edge matching on $\b{E(C)}$ that will be used before by our algorithm before phase $\mathcal{U}_{i+1}$ begins. The following lemma is reproduced from \cite{bhattacharya2016new}. 

\begin{lemma}\cite{bhattacharya2016new}\label{alg: bhatt}
    For every $\delta \in (0,1)$, there is a deterministic algorithm that maintains a $(2+\delta)$-approximate maximum matching in a dynamic graph in $O(\polylog(n,\frac{1}{\delta}))$ amortized update time, where $n$ denotes the number of nodes in the graph. 
\end{lemma}

For our purpose, it suffices to maintain a $(2+\delta)$-approximate matching for some $\delta>0$. Let $(u_j,v_j)$ be the $j^{th}$ update in phase $\mathcal{U}_i$. Let $\mathcal{B}_j$ denote the set of non-edges in almost-cliques which change (i.e. are added or removed) after our algorithm $\textsc{Update-Decomposition}(\eps, \frac{\eps}{3}, (u_j,v_j))$ is invoked for update $(u_j,v_j)$. From Theorem \ref{fd-decomposition}, it follows that the $|\mathcal{B}_j|=O(\frac{1}{\eps^4})$ in an amortized sense. Thus, on average, each edge update in a phase leads to a change of at most $O(\frac{1}{\eps^4})$ non-edges across all almost-cliques as a result of sparse or dense moves by algorithm $\textsc{Update-Decomposition}$. The dynamic matching algorithm \cite{bhattacharya2016new} is run separately for all almost-cliques, and the amortized update time to recompute the matching across all cliques before $\mathcal{U}_{i+1}$ begins is $\ti(\frac{1}{\eps^4})$ for constant $\delta>0$.  

We give the following lemma.

\begin{lemma}\label{lemma: non-edgematchingsize}
    For each almost-clique $C$, a matching $\mathcal{M}_N$ of size at least $\frac{|\b{E(C)}|}{22\varepsilon\Delta}$ can be obtained at the beginning of any phase with high probability. The amortized update time for initializing non-edge matchings across all almost-cliques is $\ti(\frac{1}{\eps^4})$.
\end{lemma}

\begin{proof}
     By our discussion above, it follows that the amortized update time to initialize non-edge matchings across all almost-cliques is $\ti(\frac{1}{\eps^4})$.

    Let $\delta=\frac{1}{6}$. By Lemma  \ref{lemma: fd-decompositionphase}, any vertex in $C$ has at most $5\varepsilon\Delta$ non-neighbors in $C$ at any point in time w.h.p. Conditioned on this event, every non-edge can be incident to at most $10\varepsilon\Delta$ non-edges w.h.p. Thus, there is a matching of size at least $\frac{|\b{E(C)}|}{10\varepsilon\Delta}$ w.h.p. Running the dynamic matching algorithm of Lemma \ref{alg: bhatt} yields that a matching $\mathcal{M}_N$ of size at least $\frac{|\b{E(C)}|}{10\varepsilon\Delta(2+\frac{1}{6})}>\frac{|\b{E(C)}|}{22\varepsilon\Delta}$ can be maintained at any given point in time.
\end{proof}

\noindent\underline{\it A Post-Processing Step for Small Matchings:} \newline In the case when the computed matching $|\mathcal{M}_N|$ is small, we employ a post-processing step before the next phase $\mathcal{U}_{i+1}$ begins to ensure that $\mathcal{M}_N$ corresponds to a \textit{maximal matching}. This allows us to maintain a large matching throughout the phase in $O(\eps\Delta)$ update time as outlined in the next section. If $|\mathcal{M}_N|\geq \eps^2\Delta$, we skip this post-processing step. 

If $|\mathcal{M}_N|<\eps^2\Delta$, it follows that $|\b{E(C)}|<22\eps\Delta\cdot \eps^2\Delta=O(\eps^3\Delta^2)$ by Lemma \ref{lemma: non-edgematchingsize}. We iterate over all $(u,v)\in \mathcal{M}_N$ and remove $(u,v)$ from $\mathcal{M}_N$, and set matched$[u]=$matched$[v]=0$. Thus, $\mathcal{M}_N=\emptyset$ at this point. Next, we \textit{greedily} compute a maximal matching in $O(|\b{E(C)}|=O(\eps^3\Delta^2)$ time as follows: go over all edges $(u,v)\in \b{E(C)}$, and if matched$[u]=$matched$[v]=0$, $(u,v)$ is added to $\mathcal{M}_N$, and matched$[u]\coloneq $ matched$[v] \coloneq 1$. 

\begin{lemma}\label{lemma: small-matchings-guarantee}
    The amortized update time for the post-processing step for all almost-cliques which have small non-edge matchings is $O(\eps n)$ with high probability. If $C$ is an almost-clique for which this step is carried out, the non-edge maximal matching $\mathcal{M}_N$ for $\b{E(C)}$ has size at least $\frac{|\b{E(C)}|}{20\eps\Delta}$ with high probability. 
\end{lemma}

\begin{proof}
 Consider an almost-clique $C$ for which this post-processing step is implemented. By Lemma \ref{lemma: fd-decompositionphase}, any vertex in $C$ has at most $5\eps\Delta$ neighbors at any point in time w.h.p., and thus each non-edge can be incident to at most $10\eps\Delta$ non-edges. Thus, there exists a matching of size at least $\frac{|\b{E(C)}|}{10\eps\Delta}$.  Since the post-processing step computes a maximal matching for $\b{E(C)}$, it follows that $\mathcal{M}_N$ has size at least $\frac{|\b{E(C)}|}{20\eps\Delta}$ w.h.p.

 The total time to implement this post-processing step for all almost-cliques is $O(\frac{n}{\Delta})\cdot O(\eps^3\Delta^2)=O(\eps^3n\Delta)$. Since the length of a phase is $\Theta(\eps^2\Delta)$, this yields an amortized update time of $O(\eps n)$. 
\end{proof}

The following Corollary follows from Lemmas  \ref{lemma: non-edgematchingsize} and \ref{lemma: small-matchings-guarantee}.
\begin{corollary}\label{corr: non-edgematching}
     For each almost-clique $C$, a matching $\mathcal{M}_N$ of size at least $\frac{|\b{E(C)}|}{22\varepsilon\Delta}$ can be obtained at the beginning of any phase with high probability. The amortized update time for initializing non-edge matchings across all almost-cliques is $\ti(\frac{1}{\eps^4}+\eps n)$ with high probability.
\end{corollary}

Before the new phase $\mathcal{U}_{i+1}$ begins, we recolor all endpoints of the non-edge matchings sequentially across all almost-cliques $C_1, C_2,...,C_{\ell}$ based on the updated decomposition.

%\noindent\textbf{\underline{Recoloring Endpoints of a Matched Non-Edge}}\newline 
\paragraph{Recoloring Endpoints of a Matched Non-Edge}
We present an algorithm $\textsc{Recolor-Non-Edge}$ which recolors endpoints of a non-edge $(u,v)$ in some almost-clique $C$. We utilize this algorithm in one of two cases: 
\begin{enumerate}
    \item Recoloring non-edges in non-edge matchings across all almost-cliques at the beginning of a new phase.
    \item Recoloring endpoints of a single non-edge during any phase if needed.
\end{enumerate}

\noindent\textbf{Data Structures.} For each almost-clique $C$, we maintain an array $A_N$ of length $[\Delta+1]$ such that $A_N[c]=1$ if color $c$ is used to color an endpoint of a matched non-edge in $\mathcal{M}_N$. For any $c$, if $A_N[c]=1$, a pointer to the non-edge $(u,v)$ whose endpoints are colored $c$ is maintained.

The algorithm is presented as follows.

\begin{algorithm}[h]
\caption{\textsc{Recolor-Non-Edge}($(u,v)$)}
\begin{algorithmic}[1]
    \State $A_N[c(u)]\leftarrow 0$, $A_N[c(v)]\leftarrow 0$.
    \State $c(u)\leftarrow \perp$, $c(v)\leftarrow \perp$. 
    \While{$c(u)=\perp$}
        \State Sample a color $c$ uniformly at random from $[\Delta+1]$.
        \If {$A_N[c]=0$ and $(L(c)\cup (L_D(c) \backslash C))\cap ((N(u)\setminus N_C(u)) \cup (N(v)\setminus N_C(v)))=\emptyset$}
            \State $A_N[c]\leftarrow 1$.
            \State $c(u)\leftarrow c$, $c(v)\leftarrow c$.
        \EndIf 
        \EndWhile 
    \State \textbf{return} $c(u)$.
\end{algorithmic}
\end{algorithm}

On input $(u,v)$, the subroutine $\textsc{Recolor-Non-Edge}$ updates the associated data structures. Next, a random color $c$ in $[\Delta+1]$ is sampled. If there exists a non-edge $(w,x)$ such that $c(w)=c(x)=c$, or a neighbor of either $u$ or $v$ outside $C$ is colored $c$, the process repeats until a proper color is found. Else, $u,v$ are colored $c$.

The following observation follows from Lemma \ref{lemma: fd-decompositionphase}, concerning the upper bound on the size of any almost-clique $C$ at any point in time.
\begin{observation}\label{obs: ubnonedgematchingsize}
    The size of the non-edge matching $|\mathcal{M}_N|\leq \frac{|C|}{2}\leq (\frac{1}{2}+5\varepsilon)\Delta$ with high probability.
\end{observation}

\begin{lemma}\label{lemma: recolornonedge}
    Algorithm \textsc{Recolor-Non-Edge} takes $\ti(\frac{n}{\Delta})$ update time with high probability.
\end{lemma}

\begin{proof}
We first analyze the probability of picking a color $c$ s.t $(L(c)\cup (L_D(c) \backslash C))\cap ((N(u)\backslash N_C(u)) \cup (N(v)\cup N_C(v)))=\emptyset$. Then, we analyze the time complexity.

Note that the number of neighbors of $u,v$ outside $C$ is at most $10\varepsilon\Delta$ by Lemma \ref{lemma: fd-decompositionphase}, both at the beginning of any phase and during any phase, with high probability. Moreover, since $|\mathcal{M}_N|\leq (\frac{1}{2}+5\varepsilon)\Delta$ by Observation \ref{obs: ubnonedgematchingsize} and each non-edge is assigned a distinct color, it follows that there are at least $\Delta+1-(10\varepsilon\Delta+(\frac{1}{2}+5\varepsilon)\Delta)\geq (\frac{1}{2}-15\varepsilon)\Delta$ colors that can be assigned to $u,v$ w.h.p. 
It is worth noting that we have lower bounded the feasible colors for $u,v$ by upper bounding neighbors of $u,v$ outside $C$ and the size of the maximum matching $\mathcal{M}_N$. Thus, conditioned on the high probability event, picking a random color $c$ succeeds with probability at least $\frac{1}{2}-15\varepsilon\geq 0.49$ for $\varepsilon\leq \frac{1}{1500}$. Thus, after $\Omega(3d\ln n)$ trials, $u,v$ are assigned a proper color with probability at least $1-\frac{1}{n^d}$ for a constant $d>0$ by a Chernoff bound. 

 Recall that $L_D(c)$ is the list containing all dense vertices colored $c$ across all almost-cliques. Since each color in $[\Delta+1]$ is assigned to at most 2 vertices in $C$, $|L_D(c)|=O(\frac{n}{\Delta})$. Conditioned on the high probability event given by Lemma \ref{lemma: colorloadsparse}, we know that $|L(c)|=\ti(\lst)$ throughout a phase. Computing the list $(L(c)\cup (L_D(c)\backslash C))$ takes $\ti(\lst)$ time. Checking whether any vertex $w$ in this list is also in $N(u)\backslash N_C(u)$ or $N(v)\backslash N_C(v)$ takes $O(1)$ time. If no neighbors of $u,v$ outside $C$ are assigned $c$, and endpoints of any non-edge in $C$ are not assigned $c$, we assign $c$ to $u,v$; otherwise we repeat. Thus, a single iteration of the while loop takes $\ti(\lst)$ time.
 
 Conditioned on the high probability events including, the properties given in Theorem \ref{fd-decomposition} which hold at the beginning of any phase, the upper bound of $(\frac{1}{2}+5\eps)\Delta$ on $|\mathcal{M}_N|$ by Observation \ref{obs: ubnonedgematchingsize}, and the upper bound on the size of lists $L(c)$ by Lemma \ref{lemma: colorloadsparse} throughout a phase, we obtain that Algorithm \textsc{Recolor-Non-Edge} takes $\ti(\frac{n}{\Delta})$ time with high probability.   
\end{proof}

\begin{corollary}\label{corr: amortizedrecolornonedge}
    The amortized update time to recolor endpoints of matched non-edges across all almost-cliques at the beginning of any phase, is $\ti(\frac{n^2}{\eps^2\Delta^2})$ with high probability.
\end{corollary}

\begin{proof}
    By Lemma \ref{lemma: recolornonedge}, the time to recolor a single non-edge is $\ti(\frac{n}{\Delta})$ with high probability. The total number of matched non-edges across all almost-cliques is bounded by $O(n)$. Thus, the total time to re color all non-edges is $\ti(\frac{n^2}{\Delta})$, and amortized over the length of a phase this yields $\ti(\frac{n^2}{\Delta})\cdot O(\frac{1}{\eps^2\Delta})=\ti(\frac{n^2}{\eps^2\Delta^2})$ amortized update time.
\end{proof}

%\noindent\textbf{\underline{Maintaining $\mathcal{M}_N$ and Recoloring Non-Edges During a Phase}}\newline 
\paragraph{Maintaining Non-Edge Matchings During a Phase}
Let $\mathcal{M}_N$ denote the non-edge matching on the set of non-edges $\b{E(C)}$ for an almost-clique $C$ obtained the beginning of a phase. By Corollary \ref{corr: non-edgematching}, it follows that $|\mathcal{M}_N|\geq \frac{|\b{E(C)}|}{22\eps\Delta}$. We distinguish two cases depending on $|\mathcal{M}_N|$: 1. $|\mathcal{M}_N|<\eps^2\Delta$ and, 2. $|\mathcal{M}_N|\geq \eps^2\Delta$.

\noindent\textbf{Case 1.} If $|\mathcal{M}_N|<\eps^2\Delta$, we ensure that a large non-edge matching with respect to the total non-edges in $\b{E(C)}$ is dynamically maintained. We do not utilize the dynamic matching algorithm \cite{bhattacharya2016new} which we use to \textit{initialize matchings at the beginning} of every phase because we want: 
\begin{enumerate}
    \item The non-edge matching $\mathcal{M}_N$ to change by at most $O(1)$ per update during a phase so that few vertices are recolored in Step II to obtain $o(n)$ update time.
    \item The update time to recolor dense vertices \textit{during any phase} to be \textit{independent} of the running time and adjustment complexity of the dynamic matching algorithm \cite{bhattacharya2016new} which simplifies analysis.
\end{enumerate} 

As a result the amortized update time guarantee of \cite{bhattacharya2016new} is utilized only towards analyzing amortized costs of re-initializing non-edge matchings at the beginning of a phase. 

Recall that if $|\mathcal{M}_N|<\eps^2\Delta$, the post-processing step must have been carried out such that $\mathcal{M}_N$ corresponds to a \textit{maximal matching}. We use a naive algorithm, $\textsc{Maintain-Matching}$ to maintain a maximal matching on $\b{E(C)}$ throughout a phase under edge updates. It is described as follows.

\noindent{\underline{Algorithm \textsc{Maintain-Matching}}.} On input $(u,v)$ where $u,v\in C$ the algorithm $\textsc{Maintain-Matching}$ proceeds as follows. First, the data structures are updated to reflect the edge update. If $(u,v)$ is an edge insertion, this is a non-edge deletion, in which case if $(u,v)\in \mathcal{M}_N$, then $(u,v)$ is removed from $\mathcal{M}_N$ and matched$[u]\coloneq$matched$[v]\coloneq 0$. For $w\in \{u,v\}$, the algorithm iterates through the list of non-edges $\b{E_C(w)}$ in $C$, and if there exists a non-edge $(w,x)$ such that matched$[x]=0$, $(w,x)$ is added to $\mathcal{M}_N$, and matched$[w]\coloneq$matched$[x]\coloneq 1$. 

 If $(u,v)$ is an edge deletion, this is a non-edge insertion. If matched$[u]=$matched$[v]=0$, $(u,v)$ is added to $\mathcal{M}_N$, and matched$[u]\coloneq$matched$[v]\coloneq 1$. 
 
 The set of newly added edges to $\mathcal{M}_N$ is returned by the algorithm. The pseudo-code is as follows.

\begin{algorithm}[h]
\caption{\textsc{Maintain-Matching}$(u,v)$}
\begin{algorithmic}[1]
        \State $M\leftarrow \emptyset$.
        \If{$(u,v)$ is an edge insertion}\Comment{$(u,v)$ is a non-edge deletion.}
            \State $\b{E(C)}\leftarrow \b{E(C)}\backslash \{(u,v)\}$, $\b{E_C(u)}\leftarrow \b{E_C(u)}\backslash \{(u,v)\}$, $\b{E_C(v)}\leftarrow \b{E_C(v)}\backslash\{(u,v)\}$.
            \If {$(u,v)\in \mathcal{M}_N$} 
                \State $\mathcal{M}_N\leftarrow \mathcal{M}_N\backslash \{(u,v)\}$.
                \State matched$[u]\coloneq$matched$[v]\coloneq 0$.
            \EndIf 
            \For {$w\in \{u,v\}$}
                \If {there exists $(w,x)\in \b{E_C(w)}$ s.t. matched$[w]=0$}
                    \State $\mathcal{M}_N\leftarrow \mathcal{M}_N \cup \{(w,x)\}$, matched$[w]\coloneq$matched$[x]\coloneq 1$, $M\leftarrow M\cup \{(w,x)\}$.
                \EndIf 
            \EndFor 
        \Else \Comment{$(u,v)$ is an edge deletion (non-edge insertion).}
            \State $\b{E(C)}\leftarrow \b{E(C)}\cup \{(u,v)\}$, $\b{E_C(u)}\leftarrow \b{E_C(u)}\cup \{(u,v)\}$, $\b{E_C(v)}\cup \b{E_C(v)}\backslash\{(u,v)\}$.
            \If {matched$[u]=$matched$[v]=0$}
                \State $\mathcal{M}_N\leftarrow \mathcal{M}_N\cup \{(u,v)\}$, matched$[u]\coloneq$matched$[v]\coloneq 1$,  $M\leftarrow M\cup \{(u,v)\}$.
            \EndIf
        \EndIf 
    \textbf{return} $M$.
\end{algorithmic}
\end{algorithm}

\begin{lemma}\label{lemma: maintainmatching}
    \textsc{Maintain-Matching} takes $O(\eps\Delta)$ time and maintains a maximal matching $\mathcal{M}_N$ of size at least $\frac{|\b{E(C)}|}{20\eps\Delta}$ with high probability, throughout a phase. 
\end{lemma}
\begin{proof}
Since every vertex has at most $O(\eps\Delta)$ non-neighbors in $C$ throughout a phase by Lemma \ref{lemma: fd-decompositionphase}, the update time of $\textsc{Maintain-Matching}$ is $O(\eps\Delta)$. Moreover, it follows from the proof of Lemma \ref{lemma: small-matchings-guarantee} that the size of $\mathcal{M}_N$ has size at least $\frac{|\b{E(C)}|}{20\eps\Delta}$ throughout the phase, since $\mathcal{M}_N$ is a maximal matching at any given point in time.      
\end{proof}

\noindent\textbf{Case 2.} If $|\mathcal{M}_N|\geq \eps^2\Delta$ at the beginning of a phase, we do not update $\mathcal{M}_N$ in this phase; the size of $\mathcal{M}_N$ decreases by at most $t=\frac{\eps^2\Delta}{18e^6}$ throughout a phase. This is sufficient for our purpose. By Lemma \ref{lemma: fd-decompositionphase}, $|\b{E(C)}|\leq 22\eps\Delta\eps^2\Delta=22\eps^3\Delta^2$ at the beginning of the phase w.h.p. Since a phase has length $\frac{\eps^2\Delta}{18e^6}$, it follows that $|\b{E(C)}|\leq 22\eps^3\Delta^2+\frac{\eps^2\Delta}{18e^6}\leq 44\eps^3\Delta^2$, for $\eps\geq \frac{1}{396e^6\Delta}$ throughout a phase. Since $|\mathcal{M}_N|\geq (1-\frac{1}{18e^{6}})\eps^2\Delta$ at the end of a phase, it follows that $|\mathcal{M}_N|\geq \frac{|\b{E(C)}|}{50\eps\Delta}$ in this case w.h.p.

The following lemma follows from our preceding discussion.
\begin{lemma}\label{lemma: nonedgematchingsizephase}
A non-edge matching of size at least $\frac{|\b{E(C)}|}{50\eps\Delta}$ can be dynamically maintained for every almost-clique throughout a phase in $O(\eps \Delta)$ update time where $\eps\geq \frac{1}{396e^6\Delta}$, with high probability.
\end{lemma}

We present a subroutine $\textsc{Update-Non-Edges}$ which maintains a non-edge matching satisfying the property in Lemma \ref{lemma: nonedgematchingsizephase} for an almost-clique throughout a phase.

\noindent{\underline{Algorithm \textsc{Update-Non-Edges}}.} 
Note that non-edge matchings are affected only when $u,v\in V_D$ and $u,v$ are in the same almost-clique, denoted by $C$. On an edge update, $(u,v)$, the data structures $N_D(u), N_D(v), \b{E_C(u)}, \b{E_C(u)}, \b{E(C)}$ are updated. The algorithm considers two cases depending on the size of the non-edge matching at the beginning of the current phase: 1) $|\mathcal{M}_N|\geq \eps^2\Delta$, and 2) $\mathcal{M}_N<\eps^2\Delta$. We describe how these cases are handled as follows.

\begin{enumerate}
\item If $(u,v)$ is an edge insertion, and $u,v$ are endpoints of some matched non-edges $(u,w), (v,x)$ where $w\neq v, x\neq u$ then $\mathcal{M}_N$ does not change. If $(u,v)\in \mathcal{M}_N$, $(u,v)$ is removed from $\mathcal{M}_N$, and $u,v$ are added to set $L_O$ which is returned by the algorithm.

If $(u,v)$ is an edge deletion, the matching $\mathcal{M}_N$ does not change. 

\item Algorithm $\textsc{Maintain-Matching}(u,v)$ is invoked which returns a set $M$. Recall that in the case when $(u,v)$ is an edge deletion such that matched$[u]$=matched$[v]=0$, $(u,v)$ is added to $\mathcal{M}_N$, $M=\{(u,v)\}$ is returned and vertices $u,v$ are added to the set $L_I$. 

If $(u,v)$ is an edge insertion such that $(u,v)\in \mathcal{M}_N$, $(u,v)$ is removed from $\mathcal{M}_N$ and the returned set $M$ consists of edges added to $\mathcal{M}_N$ where $|M|\leq 2$. In this case, if either $u$ or $v$ is no longer an endpoint of a matched non-edge in $M$, it is added to $L_O$. In any case, $L_I$ consists of all vertices which are endpoints of matched non-edges in $M$.
\end{enumerate} 

To summarize, $\textsc{Update-Non-Edges}$ returns a tuple $(L_{O}, L_{I}, M)$ where: \begin{enumerate}
    \item $M$ consists of all edges that are added to $\mathcal{M}_N$ after this update.
    \item $L_O$ consists of vertices which cease to be endpoints of a matched non-edge in $\mathcal{M}_N$ after this update. These vertices will be handled by our perfect matching algorithms in Step II.
    \item $L_I$ consists of vertices that are endpoints of newly matched non-edges in $\mathcal{M}_N$ after this update. This set can also include $u$ or $v$.
\end{enumerate} 

It follows that $|L_O|\leq 2$, $|L_I|\leq 4$ and $|M|\leq 2$. 

\begin{algorithm}[H]
\caption{\textsc{Update-Non-Edges$(u,v)$}}
\begin{algorithmic}[1]
\State $C\leftarrow$ almost-clique containing $u,v$.
\State $L_O, L_I, M\leftarrow \emptyset$.
\If{$\mathcal{M}_N>\eps^2\Delta$ at the beginning of this phase}
    \If{$(u,v)$ is an edge insertion} \Comment{$(u,v)$ is a non-edge deletion.}
        \State $\b{E(C)}\leftarrow \b{E(C)} \backslash \{(u,v)\}, \b{E_C(u)}\leftarrow \b{E_C(u)}\setminus \{v\}, \b{E_C(v)}\leftarrow \b{E_C(v)}\setminus \{u\}$.
            \If {$(u,v)\in \mathcal{M}_N$}
                \State $\mathcal{M}_N\leftarrow \mathcal{M}_N \setminus \{(u,v)\}$, $L_O\leftarrow L_O\cup\{u,v\}$.
            \EndIf 
    \Else 
        \State $\b{E(C)}\leftarrow \b{E(C)} \cup \{(u,v)\}, \b{E_C(u)}\leftarrow \b{E_C(u)}\cup \{v\}, \b{E_C(v)}\leftarrow \b{E_C(v)}\cup \{u\}$.
    \EndIf 
\Else 
    \State $b_u\leftarrow $matched$[u]$, $b_v\leftarrow $matched$[v]$.
    \State $M\leftarrow \textsc{Maintain-Matching(u,v)}$.
    \For {$w\in \{u,v\}$}
        \If {$b_w=1$ and matched$[w]=0$}
            \State $L_O\leftarrow L_O \cup \{w\}$.
        \EndIf 
    \EndFor 
    \State $L_I\leftarrow L_I \cup (\cup_{(x,y)\in M} \{x,y\})$.
\EndIf
\State \textbf{return} $(L_O, L_I, M)$.
\end{algorithmic}
\end{algorithm}

\begin{lemma}\label{lemma: update-non-edges}
    Algorithm \textsc{Update-Non-Edges} takes $O(\eps \Delta)$ time and maintains a non-edge matching of size at least $\frac{|\b{E(C)}|}{50\eps\Delta}$ for every almost-clique with high probability.
\end{lemma}
\begin{proof}
    By Lemma \ref{lemma: nonedgematchingsizephase}, and virtue of Algorithm $\textsc{Update-Non-Edges}$ a non-edge matching of size at least $\frac{|\b{E(C)}|}{50\eps\Delta}$ is maintained for every almost-clique at any point in time w.h.p. 

    If $\mathcal{M}_N >\eps^2\Delta$ at the beginning of the current phase, the running time is $O(1)$. Otherwise, the running time is $O(\eps\Delta)$, concluding the proof. 
\end{proof}

\subsubsection{Step II: Coloring via Perfect Matchings}
Given the coloring of sparse vertices and dense vertices which are matched in the non-edge matching, we present our approach to color the remaining dense vertices in each almost-clique. Consider an almost-clique $C$. Recall that we want to ensure that each color $c\in [\Delta+1]$ is used at most twice in $C$. In case $c$ is used to color endpoints of a matched non-edge, exactly two vertices are colored $c$ in $C$. Otherwise, $c$ is assigned to at most one other vertex in $C$. Let $\mathcal{L}=C\backslash (\cup_{(u,v)\in \mathcal{M}_N} \{u,v\})$ denote the set of vertices that are not endpoints of any matched non-edge in the non-edge matching $\mathcal{M}_N$ that is maintained for $C$. Let $\mathcal{R}$ denote the set of colors in $[\Delta+1]$ s.t. $\mathcal{R}=\{i|\,\, i\in [\Delta+1] \,\,\text{s.t. }A_N[i]=1\}$, which are not used to color endpoints of any non-edges in $\mathcal{M}_N$. We maintain a perfect matching $\mathcal{M}_P$ between $\mathcal{L}$ and $\mathcal{R}$ that corresponds to a proper coloring of vertices in $\mathcal{L}$ at any given point in time. 

Let $\glr$ denote bipartite graph on vertex sets $\mathcal{L}$ and $\mathcal{R}$ consisting of all edges of the form $(u,c)$ where $u\in \mathcal{L}, c\in \mathcal{R}$ if $c$ is not currently assigned to any neighbor of $u$ except potentially a neighbor in $\mathcal{L}$. We remark that $\glr$ is not maintained explicitly since this can be prohibitive in terms of time complexity. Instead, we exploit structural properties of the $\glr$ to recompute a perfect matching $\mathcal{M}_P$ before a new phase begins and dynamically maintain it during any phase under \textit{vertex updates}. The matched neighbor $c$ of any vertex $v\in \mathcal{L}$ in $\mathcal{M}_P$ simply corresponds to the color of $v$ at any point, i.e. $c(v)=c$. 

 Our algorithm utilizes a subroutine $\textsc{Match}(\cdot)$ which takes as input a vertex $v\in \mathcal{L}$ in some almost-clique $C$, and recolors $v$ by assigning a proper color to $v$. We give our subroutine \textsc{Match$(v)$} as follows. 

\begin{algorithm}[H]
\caption{$\textsc{Match}(v)$}
    \begin{algorithmic}[1]
        \If{$|\mathcal{M}_N|\geq \frac{\Delta}{10}$}
            \State \textsc{Random-Match}$(v)$.
        \Else 
            \If {$|C|>\Delta$}
                \State \textsc{Match-Large}$(v)$.
            \Else 
                \State \textsc{Match-Small}$(v)$.
            \EndIf 
        \EndIf 
    \end{algorithmic}
\end{algorithm}

Consider the case when a vertex $v\in \mathcal{L}$ has to be recolored. Algorithm \textsc{Match} considers two cases depending on the current size of $\mathcal{M}_N$: i) $|\mathcal{M}_N|\geq\frac{\Delta}{10}$, and ii) $|\mathcal{M}_N|<\frac{\Delta}{10}$. For the former, we present a simple algorithm, \textsc{Random-Match} as follows.

\noindent\underline{Algorithm \textsc{Random-Match}:} On input vertex $v$, algorithm repeatedly samples a color $c\in [\Delta+1]$ until a feasible color in $\mathcal{R}$ is found which is not assigned to any of $v's$ neighbors. 

\begin{algorithm}[H]
\caption{\textsc{Random-Match}$(v)$}
\begin{algorithmic}[1]
    \State $c(v)\leftarrow \perp$.
        \While {$c(v)=\perp$}
            \State Sample a color $c$ uniformly at random from $[\Delta+1]$.
            \If{$c\in \mathcal{R}$ and $N(v)\cap (L(c)\cup L_D(c))=\emptyset$}
                \State $c(v)\leftarrow c$, $\mathcal{M}_P\leftarrow \mathcal{M}_P\cup \{(v,c)\}$.
            \EndIf 
        \EndWhile 
\end{algorithmic}
\end{algorithm}

\begin{lemma}\label{lemma: random-match}
    On input $v$, \textsc{Random-Match} takes $\ti(\frac{n}{\Delta})$ time with high probability and assigns a proper color to $v$ for $\eps<\frac{1}{500}$.
\end{lemma}

\begin{proof}
Since $|\mathcal{M}_N|>\frac{\Delta}{10}$, note that whenever $\varepsilon<\frac{1}{500}$, $|\mathcal{L}|\leq (1+5\eps)\Delta-2|\mathcal{M}_N|\leq 0.81\Delta$. On the other hand, $|\mathcal{R}|\geq (\Delta+1-0.1\Delta)\geq 0.9\Delta$. Moreover, each vertex in $\mathcal{L}$ has at most $5\varepsilon\Delta<0.01\Delta$ neighbors outside $C$ throughout a phase by Lemma \ref{lemma: fd-decompositionphase}. Thus, there are at least $0.9\Delta-0.81\Delta-0.01\Delta=0.08\Delta=\Omega(\Delta)$ colors in $\mathcal{R}$ that $v\in \mathcal{L}$ can be matched to, at any given point in time. Thus, after $O(\log n)$ iterations of the while loop (lines 3-5) of \textsc{Random-Match}, $v$ is assigned a color with high probability by a Chernoff bound. Checking feasibility of any color takes $\ti(\frac{n}{\Delta})$ time since $L(c)$ has size at most $\ti(\frac{n}{\Delta})$ by Lemma \ref{lemma: colorloadsparse} and $L_D(c)$ has size at most $O(\frac{n}{\Delta})$. Thus \textsc{Random-Match} takes a total time of $\ti(\frac{n}{\Delta})$ time w.h.p.
\end{proof}

For the case, when the non-edge matching $\mathcal{M}_N$ is less than $\frac{\Delta}{10}$, we distinguish two cases depending on the size of the almost-clique $C$: i) $|C|>\Delta$ and ii) $|C|\leq \Delta$. We give subroutines $\textsc{Match-Large}$ and $\textsc{Match-Small}$ respectively for these two cases, in the following sections.

\subsubsection{Perfect Matchings for Large Almost-Cliques} \label{sec: large almost cliques}
In this section, we describe how to maintain perfect matchings for \textit{large almost-cliques} $C$, i.e. $|C|=\Delta+k$ for some $k>0$. 
\\\\
\noindent\textbf{Data Structures.} Recall that each vertex $v\in V$ maintains lists of its sparse and dense neighbors $N_S(v)$ and $N_D(v)$ and for any almost-clique $C$, the list $N_C(v)$ containing its neighbors in $C$. 

We maintain an additional data structure $T_C(c)$ for an almost-clique $C$ and all $c\in[\Delta]$, which counts the number of edges with one endpoint in $C$ and the other endpoint $u$ in $V_S$ where $c(u)=c$. We describe how to initialize and maintain these counters at the beginning of and, during a phase respectively.

\noindent\underline{\textit{Initializing Edge Counts}:} The counters $T_C(c)$ are initially 0 for all $c\in [\Delta+1]$, and all almost-cliques $C$ at the beginning of a phase. We iterate over all sparse vertices $v\in V_S$, and for each almost-clique $C$ increment $T_C(c(v))$ by $|N_C(v)|$. This takes $O(\frac{n}{\Delta})$ time per vertex $v$ since there are $O(\frac{n}{\Delta})$ almost-cliques, and a total time of $O(\frac{n^2}{\Delta})$ time to initialize all counters. Amortizing over the length of a phase, which is $\Theta(\varepsilon^2\Delta)$ yields an amortized update time of $O(\frac{n^2}{\eps^2\Delta^2})$ to \textit{initialize} these counters.

\noindent\underline{\textit{Maintaining Edge Counts}:} During any phase, consider the case when a sparse vertex $v$ is recolored from $i$ to $j$. For every almost-clique $C$, we decrement $T_C(i)$ by $|N_C(v)|$ and increment $T_C(j)$ by $|N_C(v)|$. This takes $O(\frac{n}{\Delta})$ time by the upper bound on the number of almost-cliques.

\begin{definition}(Heavy Color) A color $c\in [\Delta+1]$ is said to be a heavy color for an almost-clique $C$ is $T_C(c)>\frac{\Delta}{100}$. A color $c$ is light if it is not heavy.
\end{definition}

Note that by definition, a heavy color $c$ cannot be assigned to at least $\frac{\Delta}{100}$ vertices in $C$ since there exists a sparse neighbor assigned color $c$ for such vertices. The set of all heavy colors for an almost-clique $C$ is denoted by $\mathcal{H}$. Note that set $\mathcal{H}$ depends on counters $T_C(c)$, and can be maintained at any give point in the same asymptotic time required to maintain $T_C(c)$.

The next lemma shows that even if the set of heavy colors $\mathcal{H}$ is excluded from the set of remaining unused colors $\mathcal{R}$ after Step I, there are at least $|\mathcal{L}|$ colors. This allows us to maintain a perfect matching between $\mathcal{L}$ and colors in $\mathcal{R}\backslash\mathcal{H}$ efficiently, as shown later.

\begin{lemma}\label{lemma: colorfulsize}
Let $C$ be an almost-clique, such that $|C|=\Delta+k$ for $k>0$, and let $t$ denote the total number of edges incident to $C$ from vertices outside $C$. Then, $|\mathcal{M}_N|\geq \frac{k-1}{100\varepsilon}+\frac{t}{100\varepsilon\Delta}$. 
\end{lemma}

\begin{proof}
    We first lower bound the number of non-edges in $C$: 
    \begin{align*}
        |\b{E(C)}|&=\binom{\Delta+k}{2}-|E(C,C)| \\
        &\geq \frac{(\Delta+k)(\Delta+k-1)}{2}-\frac{(\Delta+k)\Delta-t}{2} \\
        &=\frac{(\Delta+k)(\Delta+k-1)}{2} -\frac{\Delta(\Delta+k)}{2}+\frac{t}{2} \\
        &=\frac{(\Delta+k)(k-1)}{2}+\frac{t}{2} \\
        &\geq \frac{(k-1)\Delta}{2}+\frac{t}{2}.
    \end{align*}
    Invoking Lemma \ref{lemma: update-non-edges}, it follows that $|\mathcal{M}_N|\geq \frac{k-1}{100\varepsilon}+\frac{t}{100\varepsilon\Delta}$, completing the proof.
\end{proof}

\begin{lemma}\label{lemma: sizeofrhsmatching}
For $\varepsilon\leq \frac{1}{10000}$, $|\mathcal{R}|-|\mathcal{H}|\geq |\mathcal{L}|$ whenever $|C|=\Delta+k$ for $k>0$.
\end{lemma}
\begin{proof}
We have that $|\mathcal{R}|\geq \Delta+1-|\mathcal{M}_N|$ and $\mathcal{L}\leq \Delta+k-2|\mathcal{M}_N|$. Thus, to prove the claim, it suffices to show that $|\mathcal{M}_N|\geq k-1+|\mathcal{H}|$. On the other hand, note that $t\geq \frac{\Delta\mathcal{|H|}}{100}$ by the definition of a heavy color. By Lemma \ref{lemma: colorfulsize}, we have that $|\mathcal{M}_N|\geq \frac{k-1}{100\varepsilon}+\frac{t}{100\varepsilon\Delta}$. Thus, we have that, 
\begin{align*}
|\mathcal{M}_N|\geq \frac{k-1}{100\varepsilon}+\frac{t}{100\varepsilon\Delta} \geq \frac{k-1}{100\varepsilon}+\frac{\mathcal{|H|}}{10000\varepsilon} \geq k-1 +|\mathcal{H}|  
\end{align*}
whenever $\varepsilon\leq\frac{1}{10000}$.
\end{proof}

\noindent\underline{Algorithm \textsc{Match-Large}:} We give our subroutine $\textsc{Match-Large}$ which on input $v$ matches (recolors) $v\in \mathcal{L}$ in some almost-clique $C$.
Our algorithm proceeds as follows. First, an arbitrary unassigned color $c\in \mathcal{R}\backslash \mathcal{H}$ is picked. If $c$ can be assigned to $v$, then $v$ is assigned $c$. Note that this requires scanning $L(c)$ and $L_D(c)$ to check if any neighbor of $v$ outside $C$, is assigned $c$. 

If $v$ cannot be colored $c$, an augmenting path of length 3 in $\mathcal{M}_P$ is found with constant probability as follows. First, the algorithm samples a \textit{random} vertex $w$ in $\mathcal{L}$. Our analysis reveals that a constant fraction of vertices in $\mathcal{L}$ can be assigned color $c$ and thus, $c$ is feasible for $w$ with constant probability. Let $c(w)$ denote the color currently assigned to $w$. Moreover, we show that $c(w)$ is feasible for $v$ with constant probability. The resulting augmenting path is $\langle u,c(w),v,c \rangle$ and the algorithm sets $c(v)\coloneq c(w)$, and $c(w)\coloneq c$. Repeating this process $O(\log n)$ times ensures that a length three augmenting path is found and $v$ is matched in $\mathcal{M}_P$ with high probability. The pseudo-code is given as follows.

\begin{algorithm}[H]
\caption{\textsc{Match-Large}$(v)$}
\begin{algorithmic}[1]
    \State $c(v)\leftarrow \perp$.
    \While {$c(v)=\perp$}
        \State Pick an unassigned color $c$ from $\mathcal{R}\backslash \mathcal{H}$.
        \If{$N(v)\cap (L(c)\cup L_D(c))=\emptyset$}
            \State $c(v)\leftarrow c$, $\mathcal{M}_P\leftarrow \mathcal{M}_P\cup\{(v,c)\}$.
        \Else 
            \State Sample a random vertex $w\in \mathcal{L}$
            \State $c'\leftarrow c(w)$.
            \If{$N(w)\cap (L(c)\cup L_D(c))=\emptyset$ and $N(v)\cap (L(c')\cup L_D(c'))=\emptyset$}
                \State $c(v)\leftarrow c'$, $\mathcal{M}_P\leftarrow \mathcal{M}_P \backslash \{(w,c')\},\mathcal{M}_P\leftarrow \mathcal{M}_P\cup \{(v,c'\}$.
                \State $c(w)\leftarrow c, \mathcal{M}_P\leftarrow \mathcal{M}_P\cup \{(w,c)\}$.
            \EndIf
        \EndIf
    \EndWhile 
\end{algorithmic}
\end{algorithm}

\begin{lemma}\label{lemma: match-large}
    On input $v$, \textsc{Match-Large} takes $\ti(\frac{n}{\Delta})$ time with high probability and assigns a proper color to $v$ for $\eps<\frac{1}{500}$.
\end{lemma}

\begin{proof}
Note that checking if an unassigned color $c$ picked from $\mathcal{R}\backslash \mathcal{H}$ is feasible for $v$ takes $\ti(\frac{n}{\Delta})$ time as before. If $c$ is not feasible for $v$, we argue that $v$ can be colored with constant probability in this iteration of the while loop. Recall from Lemma \ref{lemma: sizeofrhsmatching} that the size of $|\mathcal{R}\backslash \mathcal{H}|$ is at least $|\mathcal{L}|$. The list of colors $\mathcal{R}\backslash \mathcal{H}$ is maintained explicitly, which is straightforward since $\mathcal{R}$ and $\mathcal{H}$ are already maintained by our algorithm. 

Note that since $c\notin \mathcal{H}$, there are at most $\frac{\Delta}{100}$ vertices in $\mathcal{L}$ with a sparse neighbor having color $c$. Note that $|\mathcal{L}|\geq \Delta+k-2|\mathcal{M}_N|\geq 0.8\Delta$ where the latter inequality follows since $|\mathcal{M}_N|<\frac{\Delta}{10}$. It follows that at least $0.79\Delta$ vertices in $\mathcal{L}$ can be assigned color $c\in \mathcal{R}\setminus \mathcal{H}$. In particular, $c$ is feasible for $w$ with constant probability for a random vertex $w\in \mathcal{L}$. 

Next, note that for any vertex $v\in \mathcal{L}$, a large fraction of colors in $\mathcal{R}\backslash \mathcal{H}$ are feasible for $v$ since the number of neighbors of $v$ outside $C$ is at most $5\varepsilon\Delta\leq \frac{\Delta}{100}$ by Lemma \ref{lemma: fd-decompositionphase}. Since $|\mathcal{R}\backslash \mathcal{H}| \geq |\mathcal{L}|$, $v$ has at least $0.79\Delta$ feasible colors in $\mathcal{R}\setminus \mathcal{H}$. Now, since $w$ is a random vertex, it holds that $c(w)$ is feasible for $v$ with constant probability. Thus, the resulting augmenting path $\langle v,c(w),w,c \rangle$ is feasible with constant probability. Checking feasibility of a color for any vertex takes $\ti(\frac{n}{\Delta})$ time; thus, a single iteration of the while loop (lines 2-11) takes $\ti(\frac{n}{\Delta})$ time and $v$ is colored with constant probability. Thus, after $\ti(1)$ iterations, $v$ is colored with high probability by a Chernoff bound. Thus, the total time of \textsc{Match-Large} is $\ti(\frac{n}{\Delta})$ w.h.p.
\end{proof}

Next, we give our subroutine for small almost-cliques.

\subsubsection{Perfect Matchings for Small Almost-Cliques}\label{sec: small almost cliques}
For a small clique $C$ of size at most $\Delta$, we use a different algorithm. We assume for the rest of the section that the size of the non-edge matching $\mathcal{M}_N$ maintained on $\b{E(C)}$ is at most $\frac{\Delta}{10}$. 

Let $A$ be the set of \textit{available colors} in $[\Delta+1]$ that are not assigned to any vertex in an almost-clique $C$. We maintain $A$ explicitly in a simple manner; every time a vertex is recolored in $C$, $A$ is updated (to avoid clutter in our recoloring subroutines, we omit this update to $A$). As before, let $\mathcal{L}$ denote the set of vertices in $C$ that are not endpoints of any matched non-edge in $\mathcal{M_N}$. We give an algorithm $\textsc{Match-Small}$ as follows.

\noindent\underline{Algorithm \textsc{Match-Small}:} To assign a color to vertex $v$ our algorithm proceeds as follows. 

First, a random vertex $u\in \mathcal{L}$ and a random color $c\in A$ are sampled. The feasibility of color $c$ is checked for $u$. If $c$ is not feasible for $u$, this step is repeated.

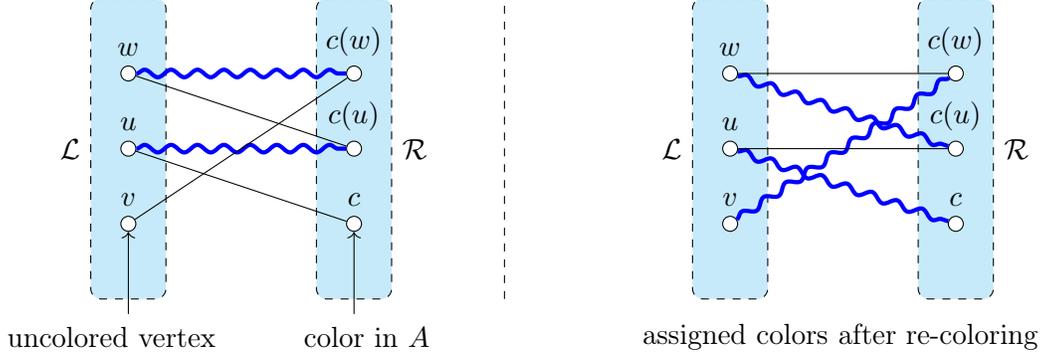
\begin{figure}
    \centering
    \input{figs/small-cliques-two}
    \caption{The length five augmenting path we find for coloring a vertex $v$ in small almost-cliques.}
    \label{fig:small-cliques}
\end{figure}

Let $u\in \mathcal{L}$ and $c$ be the random vertex and color respectively from the last step such that $c$ is feasible for $u$. Next, we sample a random vertex $w\in \mathcal{L}$ and check if $c(u)$ is feasible for $w$. If it is, we then check if $c(w)$ is feasible for $v$. If this is the case, $v$ is colored $c(w)$, $w$ is colored $c(u)$ and $u$ is colored $c$ corresponding to a length five augmenting path. See Figure~\ref{fig:small-cliques}

\begin{algorithm}[h]
\begin{algorithmic}[1]
    \caption{\textsc{Match-Small}$(v)$}
    \State $c(v)\leftarrow \perp$.

    \While{$c(v)=\perp$}
    \State $c\leftarrow \perp$. 
    \While{$c=\perp$}
        \State Sample a random vertex $u\in \mathcal{L}$ uniformly at random.
        \State Sample a random color $c'\in A$.
        \If {$N(u)\cap (L(c')\cup L_D(c'))=\emptyset$}
            \State $c\leftarrow c'$.
        \EndIf
    \EndWhile
    \State Sample a random vertex $w\in \mathcal{L}$.
    \State $c'\leftarrow c(u)$, $c''\leftarrow c(w)$.
    \If {$N(w) \cap ((L_D(c')\cup L(c')\setminus \{u\}))=\emptyset$}
         \If {($N(v)\cap ((L(c'')\cup L_D(c''))\setminus \{w\}))=\emptyset$}
            \State $c(v)\leftarrow c'',\mathcal{M}_P\leftarrow \mathcal{M}_P\cup \{(v,c'')\}$.
            \State $c(u)\leftarrow c,\mathcal{M}_P\leftarrow \mathcal{M}_P\cup \{(u,c)\}, \mathcal{M}_P\leftarrow \mathcal{M}_P\backslash \{(u,c')\}$.
            \State $c(w)\leftarrow c',\mathcal{M}_P\leftarrow \mathcal{M}_P\cup \{(w,c')\}, \mathcal{M}_P\leftarrow \mathcal{M}_P\backslash \{(w,c'')\}$.
        \EndIf 
    \EndIf 
    \EndWhile
\end{algorithmic}
\end{algorithm}

We prove the following lemma.

\begin{lemma}\label{lem:small-cliques}
    On input $v$ where $v$ is in an almost-clique of size $\Delta+1-k$ for some $k\geq 1$, \textsc{Match-Small} takes $\ti(\eps n)$ time with high probability and assigns a proper color to $v$. 
\end{lemma}

We build towards a proof of the Lemma \ref{lem:small-cliques} via the following claims. Let $U\subseteq \mathcal{L}$ denote the set of vertices in $C$ that are currently not assigned a color.

\begin{claim}\label{cl:A-value}
    $|A| = k+|\mathcal{M}_N|+|U|$.
\end{claim}
\begin{proof}
    For endpoints of a non-edge in $\mathcal{M}_N$, a single color is used and the rest of the vertices in $C \setminus U$ are assigned unique colors. Thus, the total number of colors assigned to vertices in $C$ is at most
    $$
    |\mathcal{M}_N| + (|C| - 2|\mathcal{M}_N| - |U|) = |C| - |\mathcal{M}_N| - |U| = \Delta + 1 - k - |\mathcal{M}_N| - |U|.
    $$
    Since there are a total of $\Delta + 1$ colors, the number of available colors is exactly
    $$
    \Delta + 1 - (\Delta + 1 - k - |\mathcal{M}_N| - |U|) = k + |\mathcal{M}_N| + |U|.
    $$
 \end{proof}

 Next, let us bound the number of edges with one endpoint in $D$ and another outside $C$ where $D\subseteq C$.

 \begin{claim}\label{cl:edges-outside-C}
    For any $D \subseteq C$, the total number of edges of $D$ to $V \setminus C$ is at most $|D|k + 100|\mathcal{M}_N| \eps \Delta$.
 \end{claim}
 \begin{proof}
     Recall that we maintain sets $N_C(v)$ and $\b{E_C(v)}$ corresponding to the neighbors and non-edges of any vertex $v$ in $C$, respectively. Since the maximum degree of the graph is $\Delta$, the number of edges with one endpoint in $D$ and another in $V \setminus C$ is at most
     $$
        \sum_{v \in D} \Delta - |N_C(v)| = \sum_{v \in D} \Delta - (|C| - 1 - |\b{E_C(v)}|) = \sum_{v \in D} \Delta - (\Delta - k -|\b{E_C(v)}|) =  k |D| + \sum_{v \in D} |\b{E_C(v)}|.
     $$
     Note that $\sum_{v \in D} |\b{E_C(v)}| \leq 2|\b{E(C)}|$. By Lemma~\ref{lemma: nonedgematchingsizephase}, we know $|\mathcal{M}_N| \geq |\b{E(C)}|/(50\eps \Delta)$. Together, these bounds imply $\sum_{v \in D} |\b{E_C(v)}| \leq 2|\b{E(C)}| \leq 100 |\mathcal{M}_N|\eps \Delta$. Substituting into the inequality above concludes the proof.
 \end{proof}

  Let $T$ be the set of colors not in $A$ and not assigned to any endpoint of a matched non-edge in $\mathcal{M}_N$. In other words, these are colors assigned to vertices in $\mathcal{L}\setminus U$. We call a color $c \in T$ a {\em good} color if at least $\Delta/10$ vertices in $\mathcal{L}$ do not have a neighbor in $V \setminus C$ of color $c$. Let $T_G$ denote the set of good colors, and $\mathcal{L}_G$ denote the set of vertices in $\mathcal{L}$ that are assigned a good color, respectively.  We \textit{do not} maintain the list of good colors explicitly, and they are only used for analysis.

  \begin{definition}(Color availability for a vertex) We say that a color $c\in A$ is available for a vertex $v\in \mathcal{L}$ if no vertex in $N(v)\setminus\mathcal{L}$ is assigned $c$.
  \end{definition}

  \begin{claim}\label{cl:C'-lb}
      $|\mathcal{L}_G| \geq |\mathcal{L}| - o(\Delta)$ for $\eps=o(1)$.
  \end{claim}
  \begin{proof}
      If a vertex of $\mathcal{L}$ is not in $\mathcal{L}_G$, then it is assigned a color that is not good. Call such a color \textit{bad}. By definition, for every bad color $c$, at least $|\mathcal{L}| - \Delta/10$ vertices in $\mathcal{L}$ have a neighbor in $V \setminus C$ assigned color $c$. Since $|C| \geq (1-4\eps)\Delta \geq (1-o(1))\Delta$ by Lemma ~\ref{lemma: fd-decompositionphase} and $|\mathcal{L}| = |C| - 2|\mathcal{M}_N|$, there are at least $(1-o(1))\Delta - 2|\mathcal{M}_N| - \Delta/10 \geq 0.6\Delta$ edges with one endpoint in $\mathcal{L}$ and the other to vertex in $V \setminus C$ which is assigned a particular bad color $c$. 
      
      Since the total number of edges with one endpoint in $C$ and another in $V \setminus C$ is at most $|C|5\eps \Delta \leq 5\eps \Delta^2 = o(\Delta^2)$ by Lemma \ref{lemma: fd-decompositionphase}, we get that the total number of bad colors is at most $o(\Delta^2) / 0.6 \Delta = o(\Delta)$. Hence, we get that $|\mathcal{L}_G| \geq |\mathcal{L}| - o(\Delta)$.
  \end{proof}

  \begin{claim}\label{cl:available-pairs}
    The total number of vertex-color pairs $(u, c)$ where $u \in \mathcal{L}_G$ , $c \in A$, and $c$ is available for $u$ is at least
    $$
     |\mathcal{L}_G||U| + \Omega(\Delta|\mathcal{M}_N|).
    $$
\end{claim}
\begin{proof}
    From Claim~\ref{cl:edges-outside-C}, we know that at most $|\mathcal{L}_G|k+100|\mathcal{M}_N|\eps \Delta$ edges have one endpoint in $\mathcal{L}_G$ and another in $V \setminus C$. Note that a color $c \in A$ is available for a vertex $u \in \mathcal{L}_G$ unless $u$ has a neighbor in $V \setminus \mathcal{L}$ of color $c$. As such, the total number of vertex-color pairs $(u, c)$ where $u \in \mathcal{L}_G$, $c \in A$, and $c$ is available for $u$ is at least
    $$
        |\mathcal{L}_G||A| - \Big( |\mathcal{L}_G|k+100|\mathcal{M}_N|\eps \Delta \Big).
    $$
    Replacing $|A| = k + |\mathcal{M}_N| + |U|$ from Claim~\ref{cl:A-value}, this is at least
    $$
        |\mathcal{L}_G|(k + |\mathcal{M}_N| + |U|) - \Big( |\mathcal{L}_G|k+100|\mathcal{M}_N|\eps \Delta \Big) = |\mathcal{L}_G|(|\mathcal{M}_N| + |U|) - 100\eps \Delta|\mathcal{M}_N|.
    $$
    Since $|\mathcal{L}_G| \geq |\mathcal{L}| - o(\Delta)$ and $|\mathcal{L}| = \Omega(\Delta)$, this is at least $\Omega(\Delta|\mathcal{M}_N|) + |\mathcal{L}_G||U|$.
\end{proof}

\begin{lemma}\label{lemma: pmprobability}
    Let $c\in A$ be a random color picked from $A$ and $u\in \mathcal{L}$ be a random vertex picked from $\mathcal{L}$. Let $U\subseteq \mathcal{L}$ be the set of vertices not currently assigned a color, where $|U|\geq 1$. Then, the probability that $c$ is available for $u$ is at least $\Omega(\frac{|U|}{k+|U|})$.
\end{lemma}

\begin{proof}
    By Claim \ref{cl:available-pairs}, it follows that the total number of vertex color pairs $(u, c)$ where $u \in \mathcal{L}_G$ , $c \in A$, and $c$ is available for $u$ is at least $|\mathcal{L}_G||U| + \Omega(\Delta|\mathcal{M}_N|)$. By Claim \ref{cl:C'-lb}, a random vertex $u$ belongs to $\mathcal{L}_G$ with constant probability. Conditioned on this event, the probability that $c$ is available for $v$, is at least 
    $$
        \frac{|\mathcal{L}_G||U| + \Omega(\Delta |\mathcal{M}_N|)}{|\mathcal{L}_G||A|} \stackrel{\text{Claim~\ref{cl:A-value}}}{=}  \frac{|\mathcal{L}_G||U|}{|\mathcal{L}_G|(k+|\mathcal{M}_N| + |U|)} +  \frac{ \Omega(\Delta |\mathcal{M}_N|)}{|\mathcal{L}_G|(k+|\mathcal{M}_N| + |U|)}.
    $$
 If $|\mathcal{M}_N| > k + |U|$, this quantity is at least $\Omega(1)$. On the other hand, if $|\mathcal{M}_N|<k+|U|$, where $U$ is the set of vertices currently uncolored, then the probability that a random color $c$ is available for $v$ is $\Omega(\frac{|U|}{k+|U|})$ by the first term. 
\end{proof}
The following corollary is immediate from Lemma \ref{lemma: pmprobability}.

\begin{corollary}\label{corr: pmprobability}
    Let $c\in A$ be a random color picked from $A$ and $u\in \mathcal{L}$ be a random vertex picked from $\mathcal{L}$, such that there is at least one vertex in $C$ not assigned any color. Then, the probability that $c$ is available for $u$ is at least $\Omega(\frac{1}{k})$.
\end{corollary}

\begin{lemma}\label{lemma: rtguaranteesmallpm}
    Let $U$ denote the set of vertices that are not assigned any color in a clique $C$ of size at most $\Delta$. Then, \textsc{Match-Small} takes $\ti((1+\frac{k}{|U|})\frac{n}{\Delta})$ update time with high probability..
\end{lemma}

\begin{proof}
We prove that the while loop (lines 2-15) of \textsc{Match-Small} succeeds with constant probability. Thus, after $\ti(1)$ iterations, $v$ is assigned a proper color with high probability by a Chernoff bound.

    From Lemma \ref{lemma: pmprobability}, it follows that for a random vertex $u$ and random color $c$ picked from $\mathcal{L}$ and $\mathcal{A}$ respectively, $c$ is available for $u$ with probability at least $\Omega(\frac{|U|}{k+|U|})$. Note that we can check feasibility of this color in $\widetilde{O}(n/\Delta)$ by Lemma \ref{lemma: colorloadsparse} and the fact that $L_D(c)$ has size at most $O(\frac{n}{\Delta})$. Repeating the process $\widetilde{O}(1+\frac{k}{|U|})$ times, a feasible pair $(u,c)$ can be bound with high probability. Thus, the while loop (lines 4-8) of \textsc{Match-Small} takes total time $\ti((1+\frac{k}{|U|})\frac{n}{\Delta})$ w.h.p.
    
    Let us condition on the event that the vertex $u \in \mathcal{L}_G$, which happens with constant probability by Claim \ref{cl:C'-lb}. Thus, the color $c(u)$ currently assigned to $u$ must be available for at least $\Delta/10$ vertices in $\mathcal{L}$. The algorithm picks a random vertex $w \in \mathcal{L}$ and checks if $c(u)$ is available for it; note that since $w$ is randomly chosen, $c(u)$ is available for $w$ with constant probability. This takes $\ti(\frac{n}{\Delta})$ time.
    
    Next, we analyze the probability that $c(w)$ is available for $v$. Note that since there are at least $\Delta/10$ choices of $w$, and there are at most $O(\epsilon \Delta) = o(\Delta)$ colors that we cannot assign to $v$ due to its neighbors outside $C$, it follows that $c(w)$ is available for $v$ with constant probability. Thus, $c(w)$ is available for $v$ and $c(u)$ is available for $w$, with constant probability. This step takes $\ti(\frac{n}{\Delta})$ time.
    
    Thus, the running time of $\textsc{Match-Small}$ is bounded by $\ti((1+\frac{k}{|U|})\frac{n}{\Delta})$ with high probability. 
\end{proof}

We are now ready to prove Lemma~\ref{lem:small-cliques}.

\begin{proof}[Proof of Lemma \ref{lem:small-cliques}]
Since $k=O(\eps\Delta)$ by the lower bound on the size of $C$ by Lemma \ref{lemma: fd-decompositionphase} and $|U|\geq 1 $, it follows from Lemma \ref{lemma: rtguaranteesmallpm} that the running time of \textsc{Match-Small} is bounded by $\ti(\frac{n}{\Delta}+\eps n)=\ti(\eps n)$.
\end{proof}

\paragraph{Initializing Perfect Matchings For a New Phase}\label{sec: initializepm}
Given our algorithm $\textsc{Match}$ and subroutines $\textsc{Match-Small}$ and $\textsc{Match-Large}$ respectively, we bound the total update time for initializing perfect matchings for all almost-cliques at the beginning of any phase. 

Initializing perfect matchings is simple. We simply iterate over all cliques sequentially. Let $C$ denote the clique in consideration. For all $v\in \mathcal{L}$ where $\mathcal{L}$ corresponds to the vertices in $C$ that are not endpoints of any matched non-edges in $C$, we invoke $\textsc{Match}(v)$. The following lemma bounds the amortized update time to initialize perfect matchings.

\begin{lemma}\label{lemma: recomputepmmatchings}
    The amortized update time to initialize perfect matchings for all almost-cliques, given a proper coloring of sparse vertices and matched non-edges at the beginning of any phase is  $\ti(\frac{n^2}{\eps^2\Delta^2})$ with high probability.
\end{lemma}
\begin{proof}
    Note that the total number of vertices that are matched by our perfect matching algorithms at any point in time is bounded by $n$. Note that for any vertex that is not an endpoint of a matched non-edge in its almost-clique one of the subroutines $\textsc{Random-Match}, \textsc{Match-Small}$ and $\textsc{Match-Large}$ is invoked to assign a proper color.
    
    For almost-cliques for which the non-edge matching $|\mathcal{M}_N|\geq \frac{\Delta}{10}$, it follows by Lemma \ref{lemma: random-match}, that recoloring a single vertex takes $\ti(\frac{n}{\Delta})$ w.h.p. Thus, the total time to match all vertices using Algorithm \textsc{Random-Match} is $\ti(\frac{n^2}{\Delta})$ w.h.p. 
    
    For large almost-cliques of size greater than $\Delta$, it follows by Lemma \ref{lemma: match-large} that, to color a single vertex takes time $\ti(\frac{n}{\Delta})$ and thus, the total update time is bounded by $\ti(\frac{n^2}{\Delta})$ w.h.p. 
    
    On the other hand, for a small almost-clique $C$ let $U$ denote the set of vertices which are currently not assigned a color. Clearly, $|U|\leq \Delta$ and by Lemma \ref{lemma: rtguaranteesmallpm}, an invocation to $\textsc{Match-Small}$ takes $\ti((1+\frac{k}{|U|})\frac{n}{\Delta})$ time w.h.p. Thus, the total update time to match all vertices in $C$ is bounded by $\sum_{|U|=1}^{\Delta}\ti((1+\frac{k}{|U|})\frac{n}{\Delta})=\ti(n+\sum_{|U|=1}^{\Delta}\frac{\eps n}{|U|})=\ti(n+\eps n\ln \Delta)=\ti(n)$.
    For all $O(\frac{n}{\Delta})$ small almost-cliques, this takes a total time of $\ti(\frac{n^2}{\Delta})$ w.h.p.

    Thus, the total time to initialize perfect matchings at the beginning of a phase is $\ti(\frac{n^2}{\Delta})$. Amortized over the length of a phase, this yields $\ti(\frac{n^2}{\eps^2\Delta^2})$ time w.h.p. 
\end{proof}

\subsection{The Final Algorithm}\label{sec:wrap-up}
In this section, we give our full algorithm based on our approaches outlined in the previous sections. Our algorithm utilizes two subroutines, \textsc{initialization} and \textsc{Update-Coloring}, respectively. 

Recall that our algorithm works with a fixed decomposition of $V$ into $V_S$, and $V_D=(C_1, C_2,...,C_{\ell})$ throughout a phase of consisting of $t=\frac{\eps^2\Delta}{18e^6}$ updates as outlined in Section \ref{sec: constadjcomplexity}. After a phase $\mathcal{U}_{i}$ finishes, the \textsc{Initialization} subroutine reflects edge updates in $\mathcal{U}_i$, to yield an updated sparse-dense decomposition that our algorithm works with, for the next phase. The \textsc{Initialization} subroutine is also responsible for \textit{recoloring all vertices} from scratch, based on the updated decomposition. 

The \textsc{Update-Coloring} subroutine, maintains a proper coloring of $G$ after an edge update during a phase. Both subroutines are described as follows. 

\subsubsection{\textsc{Initialization} subroutine}
The \textsc{Initialization} subroutine is responsible for i) initializing data structures and ii) recoloring vertices in $V$ at the beginning of a new phase.

\paragraph{Initializing Data Structures.} 
After a phase $\mathcal{U}_i$ finishes, edge updates in all data structures performed during the phase $\mathcal{U}_i$ are undone (reversed) to yield the states of all data structures at the beginning of a phase. Crucially, we recover the non-edge matching for every almost-clique that is computed by the dynamic matching algorithm \cite{bhattacharya2016new} at the beginning of $\mathcal{U}_i$. This matching corresponds to the one obtained before carrying out the post-processing step (in case of small of matchings to transform it into a maximal matching). This can be accomplished by maintaining a batch of $\Theta(\eps^2\Delta)$ changes that are made to all data structures in any phase. Thus, the total time to recover the states of all data structures at the beginning of a phase $\mathcal{U}_i$ is bounded by $\Theta(\eps^2\Delta)$ and takes $O(1)$ amortized update time.

Next, for each update $(u,v)$ in $\mathcal{U}_i$, we invoke algorithm \textsc{Update-Decomposition}$(\eps, \frac{\eps}{3}, (u,v))$ and the dynamic matching algorithm with $\delta=\frac{1}{6}$ (see Lemma~\ref{alg: bhatt}). This ensures that a non-edge matching based on the updated sparse-dense decomposition accounting for updates in phase $\mathcal{U}_i$ is obtained. By Lemma \ref{lemma: non-edgematchingsize} it follows that these steps combined take $\ti(\frac{1}{\eps^4})$ amortized update time w.h.p. and a matching of size at least $|\b{E(C)|}/(22\eps\Delta)$ can be obtained for all almost-cliques. Corollary~\ref{corr: non-edgematching} implies that accounting for the post-processing step (to transform small non-edge matchings to maximal matchings) yields an overall amortized update time of $\ti(\frac{1}{\eps^4}+\eps n)$.

Thus, the amortized update time for initializing data structures is bounded by $\ti(\frac{1}{\eps^4}+\eps n)$. A subtle point to note is that the update time guarantee of the dynamic matching algorithm (see Lemma \ref{alg: bhatt}) need not necessarily hold in the amortized sense for a sequence of $\widetilde{\Theta}(\eps^2\Delta\cdot \frac{1}{\eps^4})=\widetilde{\Theta}(\frac{\Delta}{\eps^2})$ updates. However, as long as the number of phases is sufficiently large (i.e. the update sequence is long enough) the amortized update time guarantee for the \textsc{Initialization} subroutine holds.

\paragraph{Recoloring Vertices from Scratch.} 
The \textsc{Initialization} subroutine recolors all vertices in $V$ from scratch. Let $V_S$ and $V_D=(C_1, C_2,...,C_{\ell})$ denote the (updated) decomposition of $V$. All vertices are uncolored at this point, i.e. $c(v)=\perp$ for all $v\in V$. By Lemma \ref{lemma: amortizedsparserecoloring}, the amortized update time to recolor sparse vertices is $\ti(\frac{n^2}{\eps^2\Delta^2})$.

Initializing the data structures $T_C(c)$ for all $c\in [\Delta+1]$ for all almost-cliques $C$, takes $O(\frac{n^2}{\eps^2\Delta^2})$ amortized update time as noted in Section \ref{sec: large almost cliques}.

Dense vertices in almost-cliques $C_1, C_2,...,C_{\ell}$ are sequentially colored as follows. Let $C_i$ be an almost-clique. First, endpoints of each matched non-edge are assigned a color using algorithm \textsc{Recolor-Non-Edge} (see \ref{sec: non-edge-matching}). The number of non-edges is at most $O(n)$ and by Corollary \ref{corr: amortizedrecolornonedge}, it follows that the amortized update time to recolor non-edges across all almost-cliques is $\ti(\frac{n^2}{\eps^2\Delta^2})$. 

Finally, we color the remaining uncolored vertices in all almost-cliques $C_1, C_2,..,C_{\ell}$ sequentially as described in Section \ref{sec: initializepm} which takes $\ti(\frac{n^2}{\eps^2\Delta^2})$ amortized update time with high probability.

Thus, recoloring all vertices from scratch takes $\ti(\frac{n^2}{\eps^2\Delta^2})$ time w.h.p.

The following Lemma is immediate from the preceding discussion.
\begin{lemma}\label{lemma: updatetimereinitializationsubroutine}
    The amortized update time for the \textsc{Initialization} is $\ti(\frac{n^2}{\eps^2\Delta^2}+\eps n + \frac{1}{\eps^4})$ w.h.p.
\end{lemma}

\subsubsection{\textsc{Handle-Update} Subroutine}\label{sec: handleupdatesubroutine}
To maintain a proper coloring on $G$ after an edge update $(u,v)$ during a phase $\mathcal{U}_i$, our algorithm utilizes the \textsc{Handle-Update} subroutine. In the following sections we consider the following cases for an edge update $(u,v)$ and give implementations of the \textsc{Handle-Update} subroutine, respectively. 

\begin{enumerate}
    \item $u,v\in V_S$.
    \item $u\in V_S, v\in V_D$.
    \item $u,v\in V_D$.
\end{enumerate} 

\paragraph{Case 1} We assume that the edge update $(u,v)$ is of the form $u,v\in V_S$.
\begin{enumerate}
    \item \underline{$(u,v)$ is an edge deletion:} Vertex $v$ (resp. $u$) is removed from $N_S(u)$ (resp. $N_S(v)$). 
    \item \underline{$(u,v)$ is an edge insertion:} Vertex $v$ (resp. $u$) is added to $N_S(u)$ (resp. $N_S(v)$). 
    
    If $c(u)=c(v)$, vertex $v$ is recolored by invoking Algorithm \textsc{Recolor-Sparse}$(v)$. Let $c$ denote the color assigned to $v$. For all almost-cliques $C$, the counters $T_C(c(u))$ and $T_C(c)$ are updated by decrementing $T_C(c(u))$ by $|N_C(v)|$ and incrementing $T_C(c)$ by $|N_C(v)|$. Next, all neighbors $w$ of $v$ in $L_D(c)$, are recolored. Since $|L_D(c)|=O(\frac{n}{\Delta})$ the number of such neighbors $w$ is bounded by $O(\frac{n}{\Delta})$. 
    
    Consider any neighbor $w$ of $v$ in some almost-clique $C$. Since any almost-clique $C$ has at most two vertices of any color, $v$ has at most two neighbors in $L_D(c)$ which are in $C$. If there is exactly one neighbor $w$, then it must be a matched vertex in the perfect matching maintained for $C$. Otherwise, if $v$ has two neighbors $w,x$ colored $c$ then $(w,x)$ must be a matched non-edge.
    
    If a neighbor $w$ is an endpoint of a matched non-edge $(w,x)$ the subroutine \textsc{Recolor-Non-Edge}$(w,x)$ is invoked. Thereafter, $c$ is added to $\mathcal{R}$. If an edge $(y, c(w))$ if present in the perfect matching $\mathcal{M}_P$, it is removed, $c(y)$ is set to $\perp$ and $\textsc{Match}(y)$ is invoked.
    
    On the other hand, if a neighbor $w$ of $v$ is an endpoint of an edge in the perfect matching $\mathcal{M}_P$, such that $(w,c)\in \mathcal{M}_P$, $(w,c)$ is removed, $c(w)$ is set to $\perp$ and $\textsc{Match}(w)$ is invoked.
\end{enumerate}
\begin{lemma}\label{lemma: c1}
    For Case 1, \textsc{Handle-Update}$(u,v)$ takes $\ti(\frac{n}{\eps^2\Delta}+\frac{n^2}{\Delta^2}+\frac{\eps n^2}{\Delta})$ update time w.h.p.
\end{lemma}

\begin{proof}
    If the update is an edge deletion, only $O(1)$ update time is incurred. If the update is an edge insertion, Algorithm $\textsc{Recolor-Sparse}$ takes $\ti(\frac{n}{\eps^2\Delta})$ time w.h.p by Lemma \ref{lemma: recolorsparse}. Updating counters $T_C(\cdot)$ takes $O(\frac{n}{\Delta})$ time.
    Recoloring of neighbors $w$ of $v$ in an almost-clique takes at most $\ti(\frac{n}{\Delta}+\eps n)$ time by Lemmas \ref{lemma: recolornonedge}, \ref{lemma: match-large} and \ref{lem:small-cliques} w.h.p. Thus, the total update time for all almost-cliques is $\ti(\frac{n}{\Delta}\cdot (\frac{n}{\Delta}+\eps n))=\ti(\frac{n^2}{\Delta^2}+\frac{\eps n^2}{\Delta})$ w.h.p.

    Thus, the total update time for \textsc{Handle-Update} for Case 1 is $\ti(\frac{n}{\eps^2\Delta}+\frac{n^2}{\Delta^2}+\frac{\eps n^2}{\Delta})$ w.h.p.
\end{proof}

\paragraph{Case 2} We assume that the edge update $(u,v)$ is of the form $u\in V_S, v\in V_D$. Let $v$ be in almost-clique $C$. 

\begin{enumerate}
    \item \underline{$(u,v)$ is an edge deletion:} Vertex $v$ (resp. $u$) is removed from $N_D(u)$ (resp. $N_S(v)$). The counter $T_C(c(u))$ is decremented by $|N_C(u)|$. Vertex $v$ is removed from $N_C(u)$. 
    \item \underline{$(u,v)$ is an edge insertion:} Vertex $v$ (resp. $u$) is added to $N_D(u)$ (resp. $N_S(v)$). The counter $T_C(c(u))$ is incremented by $|N_C(u)|$. Vertex $v$ is added to $N_C(u)$. If $c(u)=c(v)$, we recolor vertex $v$ as follows. 
    
    If $v$ is an endpoint of a matched non-edge $(v,w)$ the subroutine \textsc{Recolor-Non-Edge}$(v,w)$ is invoked. Let $c$ denote the color assigned to $v,w$. Thereafter, color $c(u)$ is added to $\mathcal{R}$ and $c$ is removed from $\mathcal{R}$. If an edge $(y, c)$ if present in the perfect matching $\mathcal{M}_P$, it is deleted, $c(y)$ is set to $\perp$ and $\textsc{Match}(y)$ is invoked. 
    
    On the other hand, if $v$ is an endpoint of an edge in the perfect matching $\mathcal{M}_P$, edge $(v,c)\in \mathcal{M}_P$ is deleted, $c(v)$ is set to $\perp$ and $\textsc{Match}(v)$ is invoked.
\end{enumerate}

\begin{lemma}\label{lemma: c2}
    For Case 2, \textsc{Handle-Update}$(u,v)$ takes $\ti(\frac{n}{\Delta}+\eps n)$ update time w.h.p.
\end{lemma}

\begin{proof}
    If the update is an edge deletion, only $O(1)$ update time is incurred. If the update is an edge insertion, the total time taken is at most $\ti(\frac{n}{\Delta}+\eps n)$ time by Lemmas \ref{lemma: recolornonedge}, \ref{lemma: match-large} and \ref{lem:small-cliques} w.h.p.
\end{proof}

\paragraph{Case 3} We assume that the edge update $(u,v)$ is of the form $u,v\in V_D$. Let $v$ be in almost-clique $C$.

\begin{enumerate}
    \item \underline{$(u,v)$ is an edge deletion:} There are two sub-cases.
    \begin{enumerate}
        \item $u\in C$: Vertex $v$ (resp. $u$) is removed from $N_D(u)$ (resp. $N_D(v)$) and $N_C(u)$ (resp. $N_C(v)$). Algorithm \textsc{Update-Non-Edges}$(u,v)$ is invoked and sets $(L_O, L_I, M)$ are returned, where $|M|\leq 1$, $L_O=\emptyset$ (since no vertices which are endpoints of matched non-edges get unmatched), and $|L_I|=2|M|$. For $w\in L_I$, edge $(w,c)\in \mathcal{M}_P$ is deleted and $w$ is removed from $\mathcal{L}$. Thereafter, \textsc{Recolor-Non-Edge}$(w,x)$ is invoked, where $(w,x)$ is a newly matched non-edge in $M$. Let $c'$ denote the color assigned to $w,x$. The color $c'$ is removed from $\mathcal{R}$ and if there exists an edge $(y,c')\in \mathcal{M}_P$, it is removed from $\mathcal{M}_P$, $c(y)$ is set to $\perp$, and $\textsc{Match}(y)$ is invoked.  
        
        \item $u\notin C$: Vertex $v$ (resp. $u$) is removed from $N_D(u)$ (resp. $N_D(v)$). Vertex $v$ is removed from $N_C(u)$, and $u$ is removed from $N_{C'}(v)$ where $C'$ is the almost-clique of $v$.
    \end{enumerate}

    \item \underline{$(u,v)$ is an edge insertion:} There are two sub-cases.
    \begin{enumerate}
        \item $u\in C$: Vertex $v$ (resp. $u$) is added to $N_D(u)$ (resp. $N_D(v)$) and $N_C(u)$ (resp. $N_C(v)$). Algorithm \textsc{Update-Non-Edges}$(u,v)$ is invoked and sets $(L_O, L_I, M)$ are returned, where $|M|\leq 2$, $|L_O|\leq 2$ and $|L_I|=2|M|$. First, for all $w\in L_I \cap \mathcal{L}$, edge $(w,c')$ is removed from $\mathcal{M}_P$, and $w$ is removed from $\mathcal{L}$. Thereafter, \textsc{Recolor-Non-Edge} is invoked on all edges $(w,x)\in M$. 
        Let $c$ denote the color assigned to $w,x$ for any non-edge $(w,x)\in M$ by \textsc{Recolor-Non-Edge}. The color $c$ is removed from $\mathcal{R}$ and if there exists an edge $(y,c)\in \mathcal{M}_P$, it is removed from $\mathcal{M}_P$, $c(y)$ is set to $\perp$, and $\textsc{Match}(y)$ is invoked.
        
        \item $u\notin C$: Vertex $v$ (resp. $u$) is added to $N_D(u)$ (resp. $N_D(v)$). Vertex $v$ is added to $N_C(u)$, and $u$ is added to $N_{C'}(v)$ where $C'$ is the almost-clique of $v$. If $c(u)=c(v)$, we recolor vertex $v$ similarly to Case 2. If $v$ is an endpoint of a matched non-edge $(v,w)$ the subroutine \textsc{Recolor-Non-Edge}$(v,w)$ is invoked. Let $c$ denote the color assigned to $v,w$. Thereafter, $c(u)$ is added to $\mathcal{R}$ and $c$ is removed from $\mathcal{R}$. If an edge $(y, c(v))$ if present in the perfect matching $\mathcal{M}_P$, it is deleted, $c(y)$ is set to $\perp$ and $\textsc{Match}(y)$ is invoked. 
        
        On the other hand, if $v$ is an endpoint of an edge in the perfect matching $\mathcal{M}_P$, edge $(v,c)\in \mathcal{M}_P$ is deleted, $c(v)$ is set to $\perp$ and $\textsc{Match}(v)$ is invoked.
    \end{enumerate}
\end{enumerate}

\begin{lemma}\label{lemma: c3}
    For Case 3, \textsc{Handle-Update}$(u,v)$ takes $\ti(\frac{n}{\Delta}+\eps n)$ update time w.h.p.
\end{lemma}

\begin{proof}
    If $u\notin C$ and $(u,v)$ is an edge deletion, the update time is $O(1)$. If $u\notin C$ and $(u,v)$ is an edge insertion, the update time is $\ti(\frac{n}{\Delta}+\eps n)$ by Lemmas \ref{lemma: recolornonedge}, \ref{lemma: match-large} and \ref{lem:small-cliques} w.h.p.

    If $(u,v)$ is an edge insertion or deletion, and $u\in C$, the call to $\textsc{Update-Non-Edges}$ takes $O(\eps\Delta)$ time w.h.p. by Lemma \ref{lemma: update-non-edges}. There are at most $O(1)$ invocations of subroutines $\textsc{Recolor-Non-Edge}$ and $\textsc{Match}$ thereafter, which takes $\ti(\frac{n}{\Delta}+\eps n)$ time in total w.h.p. by Lemmas \ref{lemma: recolornonedge}, \ref{lemma: match-large} and \ref{lem:small-cliques}. 

    Thus, for Case 3, the total update time is bounded by $\ti(\frac{n}{\Delta}+\eps n)$ w.h.p.
\end{proof}

We conclude by giving a proof of Theorem \ref{thm:mainthm}.
\mainthm*
\begin{proof}
    By Lemma \ref{lemma: updatetimereinitializationsubroutine}, the amortized update time of the \textsc{Initialization} subroutine of our algorithm is $\ti(\frac{n^2}{\eps^2\Delta^2}+\eps n + \frac{1}{\eps^4})$ while the time taken by the \textsc{Handle-Update} subroutine of our algorithm is $\ti(\frac{n}{\eps^2\Delta}+\frac{n^2}{\Delta^2}+\frac{\eps n^2}{\Delta}+\frac{n}{\Delta}+\eps n)=\ti(\frac{n}{\eps^2\Delta}+\frac{n^2}{\Delta^2}+\frac{\eps n^2}{\Delta})$ by Lemmas \ref{lemma: c1}, \ref{lemma: c2} and \ref{lemma: c3}.

    Thus, all in all our algorithm takes $\ti(\frac{n^2}{\eps^2\Delta^2}+\frac{1}{\eps^4}+\frac{\eps n^2}{\Delta})$ since $\eps<1$ and $n\geq \Delta$.
    
    Setting $\eps=\frac{\Delta^{1/5}}{n^{2/5}}$ yields an amortized update time of $\ti(\frac{n^{8/5}}{\Delta^{4/5}} + \frac{n^{14/5}}{\Delta^{12/5}})=\ti(\frac{n^{8/5}}{\Delta^{4/5}})$ for our algorithm.
    
    When $\Delta \leq n^{8/9}$, the naive algorithm which scans all neighbors of any vertex after any edge update to find a feasible color takes $O(n^{8/9})$ time. 
    
    On the other hand, when $\Delta > n^{8/9}$, our algorithm with $\eps=\frac{\Delta^{1/5}}{n^{2/5}}$ takes $\ti(n^{8/9})$ update time. Note that the conditions on $\eps$ in various Lemmas and in particular Lemma \ref{lemma: nonedgematchingsizephase} are satisfied for $\Delta$ sufficiently large, completing the proof.
\end{proof}

\input{decompositionproof}
\newpage

\printbibliography[]

\appendix 
\section{Appendix}
\subsection{Proof of Lemma \ref{lemma: colorsparse}}\label{sec: appendixcolorsparse}
We recall Lemma \ref{lemma: colorsparse}.
\colorsparse*
 We build towards a proof of the aforementioned properties via the following lemmas.

\begin{claim}\label{claim: degreeinU}
    Let $G[U]$ denote the subgraph induced on $\eps$-sparse vertices in $U$ such that each vertex in $V_S$ is included in probability $\frac{1}{2}$. Then for any $\eps$-sparse vertex $v$ with at least $\frac{9}{10}\Delta$ neighbors in $V_S$, the number of neighbors of $v$ in $U$ is in $[\frac{2}{5}\Delta, \frac{11}{20}\Delta]$ with high probability.
\end{claim}

\begin{proof}
    Fix a $\eps$-sparse vertex $v$, and let $d_U(v)$ denote the number of neighbors of $v$ in $U$. Let $E[d_U(v)]$ denote the expected number of neighbors of $v$ in $U$. If $v$ has at least $\frac{9}{10}\Delta$ neighbors in $V_S$, then $\frac{9}{20}\Delta \leq E[d_U(v)]\leq \frac{1}{2}\Delta$. By a standard Chernoff bound, it follows that $$\Pr\left[|d_U(v)-E[d_U(v)|] \geq \frac{5}{100}\Delta\right] \leq e^{-\frac{-9\Delta}{6000}} \leq \frac{1}{n^d}$$ for an arbitrarily large constant $d>0$ whenever $\Delta\geq 670d\ln n$.  
\end{proof}

We utilize the fact that the number of non-edges in the neighborhood of any vertex $v\in V_S$ is at least $\eps^2\binom{\Delta}{2}$. By introducing `ghost' vertices as neighbors of $v$ in $V_S$ such that the number of neighbors of $v$ in $V_S$ is exactly $\Delta$ (as in the proofs in \cite{ack}), it is immediate by Claim \ref{claim: degreeinU} that the number of non-edges of the form $(u,w)$ where $u,w\in U$ is at least $\frac{2}{9}\eps^2\binom{\Delta}{2}$. Note that the introduction of ghost vertices is not done by our algorithm, and is merely a tool for analysis (since we ignore edges and non-edges incident to dense neighbors of $v$). We assume that $\varepsilon\leq \frac{1}{5000}$, $\alpha\geq 5000^3$, and $(\Delta+1)>\frac{\alpha\log n}{\varepsilon^2}$.
\begin{proof}[Proof of Lemma \ref{lemma: colorsparse}]
The lower bound on $|A(v)|$ follows from Lemma \ref{lemma: one-shot} and is immediate by taking a union bound over all vertices (noting that $\Delta \gg \sqrt{n}$ as otherwise the problem can be trivially solved in $\Theta(\Delta) = O(\sqrt{n})$ update-time). 

    As a warm-up to the proof, we first claim an upper bound on the \textit{expected} running time of $\textsc{Greedy-Coloring}(V_S)$ on the set $V_S$ without running \textsc{One-Shot-Coloring} at all. The fact that for all $c\in [\Delta+1]$, $E[|L(c)|]=\frac{n}{\Delta+1}$ is not hard to observe. For the running time analysis, let $r(v)$ denote the rank of a vertex $v$ in the permutation $\pi$ such that the number of neighbors $w$ for which $\pi(w)<\pi(v)$ is exactly $r(v)-1$. Thus, for a vertex $v$ with rank $r(v)$, the probability that a random color chosen from $[\Delta+1]$ can be assigned to $v$, is at least $\frac{\Delta+1-r(v)+1}{\Delta+1}=\frac{\Delta-r(v)}{\Delta+1}$. Let $Z_i$ denote the geometric random variable which is $1$ with probability $\frac{\Delta-i}{\Delta+1}$. It follows that $E[Z_i]=\frac{\Delta+1}{\Delta-i}$, implying that if $r(v)=i$, the expected number of iterations of the while loop of \textsc{Greedy-Coloring} (lines 4-7) is at most $\frac{\Delta+1}{\Delta-i}$. The probability that $r(v)=i$ for any $i$ is $\frac{1}{d(v)}$ for any vertex $v$. Let $d(v)$ denote the degree of $v$. Thus, the expected number of while loop iterations for any vertex $v$ is $\frac{\Delta+1}{d(v)}\sum_{i=1}^{d(v)}\frac{1}{\Delta-i+2}=O(\Delta\ln \Delta \frac{1}{d(v)})$. For all vertices $v$, the total number of iterations is bounded by $O(\Delta \ln \Delta)\sum_{v\in V}\frac{1}{d(v)}=O(\Delta\ln \Delta)\cdot n(\frac{1}{n\Delta/n})=O(\Delta \ln \Delta)\cdot \frac{n}{\Delta}=\frac{n\ln \Delta}{\Delta}$.
    Each iteration of the while loop takes $O(\frac{n}{\Delta})$ time in expectation by the expected bound on the list sizes. Thus, in expectation, the running time of $\textsc{Color-Sparse}(V_S)$ is $O(\frac{n^2\ln \Delta}{\Delta})$. 

    Now, we analyze the running time of $\textsc{Color-Sparse}$. Note that $\textsc{One-Shot-Coloring}$ runs in $\ti(\frac{n^2}{\Delta})$ time with high probability. This follows since a total of $\widetilde{O}(n)$ colors are sampled, and so each color is sampled $\ti(\frac{n}{\Delta})$ times with high probability, thus $|L(c)|=\ti(\frac{n}{\Delta})$. To check whether a color can be assigned to a vertex takes time $O(|L(c)|)=\ti(\frac{n}{\Delta})$ for a total time of $\ti(\frac{n^2}{\Delta})$. 

    For the running time of $\textsc{Greedy-Coloring}$ we first observe that for any vertex $v$ with number of sparse neighbors at most $\frac{9}{10}\Delta$, the number of iterations of the while loop (lines 4-7) of \textsc{Greedy-Coloring} is $\ti(1)$ before a feasible color is assigned to $v$. This holds since, a random color is feasible with probability at least $\frac{1}{10}$, and after $\ti(1)$ random colors have been sampled, one must be feasible for $v$ w.h.p. Thus, for the remainder of the proof, we bound the number of colors sampled for vertices with degree greater than $\frac{9}{10}\Delta$. 
    
    Fix a vertex $v$, with at least $\frac{9}{10}\Delta$ sparse neighbors. Let $r(v)=i$ denote the rank of $v$ among its remaining $d_r(v)$ uncolored neighbors after \textsc{One-Shot-Coloring} finishes, where $i\in [1, d_r(v)+1]$. When $\textsc{Greedy-Coloring}$ considers $v$, the number of colors that are feasible for $v$ is at least $\Delta+1-(|N_S(v)|-d_r(v))-(i-1)=\Delta+2-|N_S(v)|+d_r(v)-i$. On the other hand, by Claim \ref{claim: degreeinU}, $d_r(v)\geq \frac{9}{10}\Delta-\frac{11}{20}\Delta \geq \frac{7}{20}\Delta$ and $d_r(v)\leq \Delta-\frac{9}{20}\Delta=\frac{11}{20}\Delta$ w.h.p. It follows that the total number of colors sampled before $v$ is assigned a feasible color by \textsc{Greedy-Coloring} is $$\ti\left(\frac{\Delta+1}{\Delta+2-|N_S(v)|+d_r(v)-i}\right)=\ti\left(\frac{\Delta}{d_r(v)-i+2}\right).$$
    
    To obtain a high probability bound, we analyze the following stochastic process. Consider each vertex in $S\coloneq V_S\setminus U'$ picking one of $b=\Theta(\frac{\Delta}{\log n})$ bins $B_1, B_2,...,B_b$ uniformly at random. Once this process is finished, we consider vertices in the following order: go over all bins $B_i$ in the order of increasing index $i$, and within each bin $B_i$, consider vertices in bin $B_i$ in random order. Let $r(v)$ be the rank of a vertex in this order, where $i\in [1,d_r(v)+1]=[1,\Theta(\Delta)]$. We note that by virtue of this random process, the resulting final ordering of vertices corresponds to a random permutation. It can be shown by a standard application of a Chernoff bound that the number of vertices in each bin is $\Theta(\frac{n\log n}{\Delta})$ with high probability. Here $n=|V_S|$. By a similar calculation, each vertex has at most $\Theta(\log n)$ neighbors in each bin with high probability.

    Let us consider a grouping of vertices in geometric order of their ranks. Let $s=\frac{11}{20}\Delta$. Let $G_i$ denote the group of vertices which have ranks in the range $[s-2^{i+1}+1, s-2^{i}]$, where $i\in \{0,1,2,...,\Theta(\log \Delta)\}$. First, we note that for a vertex $v$ in $G_i$, i.e. for a vertex of rank at least $s-2^{i+1}+1$, it must be in the last $O(2^{i+1})$ bins with high probability, by a Chernoff bound. This follows since there are at least $s-2^{i+1}$ neighbors of $v$ before $v$ in the ordering, each bin contains at $\Theta(\log n)$ neighbors of $v$ with high probability, and conditioned on this, $v$ must be in the last $\Theta(\frac{\Delta}{\log n}-\frac{\Delta-2^{i+1}}{\log n})=\Theta(\frac{2^{i+1}}{\log n})$ bins. Thus, w.h.p, $v$ is in the last $\Theta(2^{i+1})$ bins. Any bin is assigned at most $\frac{n\log n}{\Delta}$ vertices with high probability, thus it follows that the number of vertices with rank at least $s-2^{i+1}+1$ is bounded by $\ti(2^{i+1}\frac{n}{\Delta})$. 

    Consider a vertex in $G_i$. The number of iterations of the while loop, i.e. number of colors sampled before a feasible color is found with constant probability for $v$ is at most $\frac{\Delta+1}{\Delta+1-(s-2^i)}=O(\frac{\Delta+1}{2^i})$, and $\ti(\frac{\Delta}{2^i})$ with high probability. Thus, the total number of iterations for vertices in $G_i$ is bounded by $\ti(\frac{\Delta}{2^i}\cdot 2^{i+1}\frac{n}{\Delta})=\ti(n)$ with high probability. Taking a union bound over all $O(\log \Delta)$ groups, yields that the total number of colors drawn throughout Algorithm $\textsc{Color-Sparse}$ is $\ti(n)$ with high probability. 

    Since a total of $\ti(n)$ colors are drawn uniformly at random from $[\Delta+1]$, it follows that the number of times a particular color is drawn is no more than $\ti(\frac{n}{\Delta})$ w.h.p. Thus, the length of lists $L(c)$ for all $c\in [\Delta+1]$ is bounded by $\ti(\frac{n}{\Delta})$ with high probability.

    As a result, each iteration of the while loop takes $\ti(\frac{n}{\Delta})$ w.h.p. Thus, the total running time of \textsc{Color-Sparse} is bounded by $\ti(\frac{n}{\Delta})\cdot \ti(n)=\ti(\frac{n^2}{\Delta})$.
\end{proof}

\end{document}

%% file: figs/clique-decomposition.tex
\begin{tikzpicture}
    % Draw bounding circles behind each clique
    \draw[fill=cyan!20, dashed] (0,0) circle (2.1cm);
    \draw[fill=cyan!20, dashed] (5,0) circle (2.1cm);

    % Define the style for the vertices
    \tikzstyle{vertex}=[circle, draw, fill=white, inner sep=0pt, minimum size=15pt, font=\footnotesize]
    \tikzstyle{clique}=[circle, draw, fill=gray!30, inner sep=0pt, minimum size=15pt, font=\footnotesize]
    
    % Coordinates for the vertices in the first clique
    \node[clique] (A1) at (0:1.6) {};
    \node[clique] (A2) at (60:1.6) {};
    \node[clique] (A3) at (120:1.6) {};
    \node[clique] (A4) at (180:1.6) {};
    \node[clique] (A5) at (240:1.6) {};
    \node[clique] (A6) at (300:1.6) {};

    % Coordinates for the vertices in the second clique
    \node[clique] (B1) at ($ (5,0) + (0:1.6) $) {};
    \node[clique] (B2) at ($ (5,0) + (60:1.6) $) {};
    \node[clique] (B3) at ($ (5,0) + (120:1.6) $) {};
    \node[clique] (B4) at ($ (5,0) + (180:1.6) $) {};
    \node[clique] (B5) at ($ (5,0) + (240:1.6) $) {};
    \node[clique] (B6) at ($ (5,0) + (300:1.6) $) {};

    % Coordinates for the scattered vertices
    \node[vertex] (C1) at (2.5, 2.8) {};
    \node[vertex] (C2) at (7.5, 2.8) {};
    \node[vertex] (C3) at (2.5, -2.8) {};
    \node[vertex] (C4) at (7.5, -2.8) {};

    % Draw edges within the first clique
    \draw (A1) -- (A2);
    \draw (A1) -- (A3);
    \draw (A1) -- (A4);
    \draw (A1) -- (A6);
    \draw (A2) -- (A4);
    \draw (A2) -- (A5);
    \draw (A2) -- (A6);
    \draw (A3) -- (A4);
    \draw (A3) -- (A5);
    \draw (A3) -- (A6);
    \draw (A4) -- (A5);
    \draw (A4) -- (A6);
    \draw (A5) -- (A6);

    % Draw edges within the second clique
    \draw (B1) -- (B3);
    \draw (B1) -- (B4);
    \draw (B1) -- (B5);
    \draw (B1) -- (B6);
    \draw (B2) -- (B3);
    \draw (B2) -- (B4);
    \draw (B2) -- (B5);
    \draw (B2) -- (B6);
    \draw (B3) -- (B4);
    \draw (B3) -- (B5);
    \draw (B4) -- (B6);
    \draw (B5) -- (B6);

    % Draw a few edges between the two cliques
    \draw (A1) -- (B3);
    \draw (A5) -- (B4);

    % Draw edges from scattered vertices to the cliques
    \draw (C1) -- (A1);
    \draw (C1) -- (A2);
    \draw (C2) -- (B2);
    \draw (C2) -- (C1);
    \draw (C3) -- (A6);
    \draw (C3) -- (B4);
    \draw (C4) -- (B1);
    \draw (C4) -- (B5);
    \draw (C4) -- (B6);
    \draw (C4) -- (C3);

\end{tikzpicture}

%% file: figs/large-clique.tex
\begin{tikzpicture}
    % Define the rectangles
    \node[draw, rectangle, rounded corners, minimum width=1cm, minimum height=3cm, fill=cyan!20, dashed, fill opacity=0.3, label=left:$\mathcal{L}$] (L) at (0, 1) {};
    \node[draw, rectangle, rounded corners, minimum width=1cm, minimum height=3cm, fill=cyan!20, dashed, fill opacity=0.3, label=right:$\mathcal{R} \setminus \mathcal{H}$] (R) at (3, 1) {};
    
    % Define the vertices
    \node[draw, fill=white, circle, inner sep=2pt, label=above:$w$] (w) at (0, 1.5) {};
    \node[draw, fill=white, circle, inner sep=2pt, label=above:$v$] (v) at (0, 0) {};
    
    \node[draw, fill=white, circle, inner sep=2pt, label=above:$c(w)$] (cw) at (3, 1.5) {};
    \node[draw, fill=white, circle, inner sep=2pt, label=above:$c$] (c) at (3, 0) {};
    
    % Draw the edges
    \draw (c) -- (w);
    \draw[ultra thick, blue, decorate, decoration={snake, amplitude=0.5mm}] (w) -- (cw);
    \draw (cw) -- (v);
\end{tikzpicture}

%% file: figs/small-cliques.tex
\begin{tikzpicture}
    % Define the rectangles
    \node[draw, rectangle, rounded corners, minimum width=1cm, minimum height=3.5cm, fill=cyan!20, dashed, fill opacity=0.3, label=left:$\mathcal{L}$] (L) at (0, 1) {};
    \node[draw, rectangle, rounded corners, minimum width=1cm, minimum height=3.5cm, fill=cyan!20, dashed, fill opacity=0.3, label=right:$\mathcal{R}$] (R) at (3, 1) {};
    
    % Define the vertices
    \node[draw, fill=white, circle, inner sep=2pt, label=above:$w$] (w) at (0, 2-0.2) {};
    \node[draw, fill=white, circle, inner sep=2pt, label=above:$u$] (u) at (0, 1-0.2) {};
    \node[draw, fill=white, circle, inner sep=2pt, label=above:$v$] (v) at (0, 0-0.2) {};
    
    \node[draw, fill=white, circle, inner sep=2pt, label=above:$c(w)$] (cw) at (3, 2-0.2) {};
    \node[draw, fill=white, circle, inner sep=2pt, label=above:$c(u)$] (cu) at (3, 1-0.2) {};
    \node[draw, fill=white, circle, inner sep=2pt, label=above:$c$] (c) at (3, 0-0.2) {};
    
    % Draw the edges
    \draw (c) -- (u);
    \draw[ultra thick, blue, decorate, decoration={snake, amplitude=0.5mm}] (u) -- (cu);
    \draw (cu) -- (w);
    \draw[ultra thick, blue, decorate, decoration={snake, amplitude=0.5mm}] (w) -- (cw);
    \draw (cw) -- (v);
\end{tikzpicture}

%% file: decomposition.tex
\section{A Fully Dynamic Sparse-Dense Decomposition}\label{sec: fullydynamicsparsedensedecomp}
In this section, we give a fully dynamic algorithm to maintain a sparse-dense decomposition of a dynamic graph $G=(V,E)$ under edge updates. 

In a sparse-dense decomposition, $V$ is partitioned into a set $V_{S}$ of \textit{sparse} vertices and a set $V_D$ of \textit{dense} vertices. The set $V_D$ is further partitioned into a collection of cliques $C_1, C_2, \ldots, C_k$, where for any $i\in [k]$, $C_i$ is an \textit{almost}-clique, that is, $|C_i|=\Theta(\Delta)$ and $|E(C_i)|\geq (1-\Omega(\varepsilon))\Delta^2$. The graph decomposition is parametrized by a small constant $\varepsilon>0$ and maintained \textit{independently} of the coloring. In the following, we introduce some definitions to formalize our decomposition.

\begin{definition}[\textbf{$\varepsilon$-friend edges} \cite{hss}]\label{defn: friendedge} Given $0 < \varepsilon<1$, an edge $(u,v)\in E$ is an $\varepsilon$-friend edge if $$|N(u)\cap N(v)|\geq (1-\varepsilon)\Delta.$$ A vertex $u$ (resp. $v$) is an $\varepsilon$-friend of $v$ (resp. $u$) if $(u,v)$ is an $\varepsilon$-friend edge.
\end{definition}

\begin{definition}[\textbf{$\varepsilon$-dense vertices} \cite{hss}] Given $0< \varepsilon <1$, a vertex is $\varepsilon$-dense if it has at least $(1-\varepsilon)\Delta$ $\varepsilon$-friends in $V$. A vertex is $\varepsilon$-sparse if it is not $\varepsilon$-dense.
\end{definition}

Let $F_{\varepsilon}\subseteq E$ denote the set of $\varepsilon$-friend edges of $G$. Let $F_{\varepsilon}(v)\subseteq N(v)$ denote the set of all $\varepsilon$-friends of $v$. Let $V^{\text{dense}}_{\varepsilon}$ denote the set of all $\varepsilon$-dense vertices and $V^{\text{sparse}}_{\varepsilon}\coloneq V\backslash V^{\text{dense}}_{\varepsilon}$ denote the set of all $\varepsilon$-sparse vertices. Note that the number of edges whose endpoints are neighbors of any $\varepsilon$-sparse vertex $v$ is given by $|E(N(v))|\leq \binom{\Delta}{2}-\frac{\varepsilon\Delta\cdot \varepsilon\Delta}{2}\leq (1-\varepsilon^2)\binom{\Delta}{2}$. This follows because any vertex in $V$ which is \textit{not} an $\varepsilon$-friend of $v$ has at most $(1-\varepsilon)\Delta$ common neighbors with $v$ and the number of such vertices is at least $\varepsilon\Delta$. 

Let $G_{\varepsilon}=(V^{\text{dense}}_{\varepsilon}, F_{\varepsilon})$ be the induced subgraph of $G$ on the set of $\varepsilon$-dense vertices, which contains only $\varepsilon$-friend edges in $F_{\varepsilon}$, i.e. $F_{\varepsilon}=\{(u,v)|\, u,v\in V^{\text{dense}}_{\varepsilon}, (u,v)\in F_{\varepsilon}\}$. 
The following lemma (reproduced from \cite{hss}) gives properties of the decomposition of $G_{\eps}$ into almost-cliques $C_1, C_2,...,C_k$.  
%\sbcomment{Refer to Reed's instead of HSS and ACK?}

\begin{lemma}[\cite{hss}]\label{lemma: hss}
    For any graph $G=(V,E)$, and any $0\leq \varepsilon<\frac{1}{5}$, let $G_{\varepsilon}=(V^{\text{dense}}_{\varepsilon}, F_{\varepsilon})$ denote the induced subgraph of $G$ on the set of $\varepsilon$-dense vertices and the set of $\varepsilon$-friend edges. Then, there exists a decomposition of $G_{\varepsilon}$ into connected components $C_1,...,C_k$ satisfying $V_{\varepsilon}=C_1\cup C_2...\cup C_k$ such that for any $i\in [k]$: \begin{enumerate}
        \item $|C_i|\leq (1+3\varepsilon)\Delta$.
        \item Each vertex in $C_i$ has at most $\varepsilon\Delta$ neighbors (in $G$) in $V\backslash C_i$.
        \item Each vertex in $C_i$ has at most $3\varepsilon\Delta$ non-neighbors (in $G$) in $C_i$. 
        \item For any $u,v\in C_i$, $|N(u)\cap N(v)|\geq (1-2\varepsilon)\Delta$.
    \end{enumerate}
\end{lemma}

To utilize this decomposition for fully dynamic $(\Delta+1)$-coloring, we require additional properties. In particular, a lower bound on $|C_i|$ of $\Omega(\Delta)$ for all $i$ gives an upper bound on the number of neighbors outside $C_i$ for any vertex in $C_i$. Furthermore, it is crucial to limit the number of non-edges across all almost-cliques which change after a single edge update to the graph; this effectively limits the change in the coloring maintained by our algorithm. We present the first fully dynamic algorithm which maintains a sparse-dense decomposition satisfying these additional properties.
 
 For a vertex $v$ in an almost-clique $C$, let $\b{E_C(v)}$ denote the total number of non-edges of $v$ in $C$, i.e. $\b{E_C(v)}=\{(u,v)| \, u\in C\backslash N(u)\}$. Let $\b{E(C)}=\bigcup_{v\in C} \b{E_C(v)}$ denote the set of all non-edges in $C$. 

\begin{definition}(Adjustment Complexity of Non-Edges) The adjustment complexity of non-edges is the total number of changes to the list of non-edges across all almost-cliques after a single edge update to $G$.
\end{definition}

The following theorem summarizes the main result of this section.
\begin{restatable}{theorem}{fddecomposition}
\label{fd-decomposition}
 For any graph $G=(V,E)$, and any constant $0< \varepsilon<\frac{3}{50}$, there exists a fully dynamic randomized algorithm which maintains a graph decomposition of $G$ in $O(\frac{\ln n}{\varepsilon^4})$ worst-case update time against an adaptive adversary. The algorithm maintains a partition of $V\coloneqq V_S \cup V_D$ such that at all times $V_S\subseteq \vs_{4\varepsilon/3}$ and $V_D\subseteq \vd_{11\eps/4}$. 
    
    Furthermore, $V_D$ is partitioned into vertex-disjoint \textit{almost-cliques} $V_D\coloneqq C_1\cup C_2\cup...C_{\ell}$ such that after processing any edge update, the following properties for all $i\in [\ell]$ hold with high probability: 
    \begin{enumerate}
        \item $(1-4\varepsilon)\Delta\leq |C_i|\leq (1+10\varepsilon)\Delta$.
        \item Each vertex in $C_i$ has at least $(1-4\varepsilon)\Delta$ neighbors in $C_i$.
%        \item Each vertex in $C_i$ has at most $4\varepsilon\Delta$ neighbors outside $C_i$. \rrcomment{implied by previous item.  do we need it?}
        \item The adjustment complexity of non-edges is $O(\frac{1}{\varepsilon^4})$.
    \end{enumerate}
\end{restatable}

Our fully dynamic algorithm to maintain the sparse dense decomposition and the proof of Theorem \ref{fd-decomposition} is deferred to Section \ref{sec: fdalgorithmdecomposition}.

%% file: figs/small-cliques-two.tex
\begin{tikzpicture}
    % Left figure
    \begin{scope}
        % Define the rectangles
        \node[draw, rectangle, rounded corners, minimum width=1cm, minimum height=4cm, fill=cyan!20, dashed, fill opacity=0.5, label=left:$\mathcal{L}$] (L) at (0, 1) {};
        \node[draw, rectangle, rounded corners, minimum width=1cm, minimum height=4cm, fill=cyan!20, dashed, fill opacity=0.5, label=right:$\mathcal{R}$] (R) at (3, 1) {};
        
        % Define the vertices
        \node[draw, fill=white, circle, inner sep=2pt, label=above:$w$] (w) at (0, 2) {};
        \node[draw, fill=white, circle, inner sep=2pt, label=above:$u$] (u) at (0, 1) {};
        \node[draw, fill=white, circle, inner sep=2pt, label=above:$v$] (v) at (0, 0) {};
        
        \node[draw, fill=white, circle, inner sep=2pt, label=above:$c(w)$] (cw) at (3, 2) {};
        \node[draw, fill=white, circle, inner sep=2pt, label=above:$c(u)$] (cu) at (3, 1) {};
        \node[draw, fill=white, circle, inner sep=2pt, label=above:$c$] (c) at (3, 0) {};
        
        % Draw the edges
        \draw (c) -- (u);
        \draw[ultra thick, blue, decorate, decoration={snake, amplitude=0.5mm}] (u) -- (cu);
        \draw (cu) -- (w);
        \draw[ultra thick, blue, decorate, decoration={snake, amplitude=0.5mm}] (w) -- (cw);
        \draw (cw) -- (v);
        
        % Add an arrow with text for vertex v
        \draw[->] (0, -1.2) -- (v);
        \node at (1.3, -1.5) [left] {uncolored vertex};
        
        % Add an arrow with text for vertex c
        \draw[->] (3, -1.2) -- (c);
        \node at (2.2, -1.5) [right] {color in $A$};

        \node at (6.7, -1.5) [right] {assigned colors after re-coloring};
    \end{scope}
    
    % Divider
    \draw[dashed] (5, -1) -- (5, 3);

    % Right figure
    \begin{scope}[xshift=8cm]
        % Define the rectangles
        \node[draw, rectangle, rounded corners, minimum width=1cm, minimum height=4cm, fill=cyan!20, dashed, fill opacity=0.5, label=left:$\mathcal{L}$] (L) at (0, 1) {};
        \node[draw, rectangle, rounded corners, minimum width=1cm, minimum height=4cm, fill=cyan!20, dashed, fill opacity=0.5, label=right:$\mathcal{R}$] (R) at (3, 1) {};
        
        % Define the vertices
        \node[draw, fill=white, circle, inner sep=2pt, label=above:$w$] (w) at (0, 2) {};
        \node[draw, fill=white, circle, inner sep=2pt, label=above:$u$] (u) at (0, 1) {};
        \node[draw, fill=white, circle, inner sep=2pt, label=above:$v$] (v) at (0, 0) {};
        
        \node[draw, fill=white, circle, inner sep=2pt, label=above:$c(w)$] (cw) at (3, 2) {};
        \node[draw, fill=white, circle, inner sep=2pt, label=above:$c(u)$] (cu) at (3, 1) {};
        \node[draw, fill=white, circle, inner sep=2pt, label=above:$c$] (c) at (3, 0) {};
        
        % Draw the edges
        \draw[ultra thick, blue, decorate, decoration={snake, amplitude=0.5mm}] (v) -- (cw);
        \draw[ultra thick, blue, decorate, decoration={snake, amplitude=0.5mm}] (w) -- (cu);
        \draw[ultra thick, blue, decorate, decoration={snake, amplitude=0.5mm}] (u) -- (c);
        \draw (w) -- (cw);
        \draw (u) -- (cu);
    \end{scope}
    
\end{tikzpicture}

%% file: decompositionproof.tex
\section{Fully Dynamic Algorithm for Sparse-Dense Decomposition}\label{sec: fdalgorithmdecomposition}
We recall Theorem \ref{fd-decomposition}.
\fddecomposition*
The rest of this section is devoted to proving Theorem \ref{fd-decomposition}. We first start with an informal overview of the algorithm in Section~\ref{sec: tech-decomposition}, and formalize it in the forthcoming sections.

\subsection{High Level Overview of Our Algorithm}\label{sec: tech-decomposition}
We give an informal overview of our sparse-dense decomposition in this section.

Our algorithm keeps track of the sparsity of vertices and `friend-ness' of edges under edge updates. We give simple subroutines which rely on random sampling to determine whether an edge is an $\eps$-friend edge or a vertex is $\eps$-dense. By utilizing standard concentration arguments, we give high probability bounds on their correctness. A major observation which is exploited by our algorithm is that \textit{sparsity} of a vertex $v$ is fairly \textit{insensitive} to \textit{edge updates}. In other words, sparsity of a vertex $v$ changes noticeably if $v$ or its neighbors are incident to a sufficiently large number of edge updates. For example, suppose $v$ is $\eps$-dense at some point. Then, for $v$ to become $2\eps$-sparse, it must lose at least $\eps\Delta$ of its $\eps$-friends. This can happen via the following types of updates. Firstly, an edge deletion $(u,v)$ where $u$ is an $\epsilon$-friend of $v$ decreases the number of $\eps$-friends of $v$ by 1. Secondly, an $\eps$-friend $u$ of $v$ might lose $\eps$-`friend-ness' if it becomes a $2\eps$-friend. This happens when the number of common neighbors of $u$ and $v$ decreases by at least $\Omega(\eps\Delta)$. Thus, for $\eps\Delta$ friends of $v$ which were previously $\eps$-friends and are now $2\eps$-friends, there must be at least $\epsilon^2\Delta^2$ edge updates in the neighborhood of $v$. In either case, by assigning a credit of $\ti(\frac{1}{\eps^4})$ to endpoints of every updated edge, we show that $v$ has a credit of $\ti(\frac{\Delta}{\eps^2})$ available after it becomes $2\eps$-sparse. Since the source of all credits is the credit that is assigned to endpoints of an updated edge, some of this credit on $v$ may be contributed by its neighbors. Our algorithm ensures that whenever a vertex accumulates a credit of $\ti(\frac{\Delta}{\eps^2})$, it is used to update various data structures, recompute `friend-ness' of incident edges and leave a credit on $v's$ neighbors. Our overall charging scheme is intricate, and involves a careful analysis (see Section \ref{sec: fdalgtomaintainfriendedges} for details).

In addition to maintaining the list of sparse and dense vertices $V_S$ and $V_D$ respectively, we also require a \textit{fully dynamic} decomposition of $V_D$ into almost cliques $(C_1, C_2,...,C_{\ell})$. Additionally, we want all almost-cliques to satisfy desirable properties at all times such as size bounds and small adjustment complexity after any edge update. Instead of using a fully dynamic connectivity algorithm as a black-box to maintain these almost-cliques induced by friend edges, which could blow up the adjustment complexity in general, we give an alternative approach to maintain almost-cliques which ensures low adjustment complexity. 

At a high level, our algorithm works as follows. Let us first consider the case when a vertex $v$ becomes $\eps$-dense. If $v$ is the first vertex in its neighborhood to become $\eps$-dense, we create an almost-clique containing $v$ and its $(1-\eps)\Delta$, $\eps$-friends. However, if $v$ already has an $\eps$-friend $u$ in $V_D$, our algorithm moves all $\epsilon$-friends of $v$ (which could be in $V_S$) together with $v$, to $u$'s almost-clique. Next, consider when a vertex becomes $\Theta(\eps)$-sparse. In this case, $v$ is moved to $V_S$ and data structures are updated. However, over a period of time such moves to $V_S$ may cause almost-cliques to violate the desired lower bound of $\Omega(\Delta)$. To get around this, we maintain a size invariant (in addition to other key invariants--see Section \ref{sec: fullydynamicdecomposition}). Whenever an almost-clique shrinks beyond a certain threshold size, our algorithm \textit{collapses} the whole clique and handles sufficiently dense vertices in the collapsed clique individually. Such vertices can potentially join other almost-cliques thereafter. Our charging argument shows that a total credit of $\Theta(\Delta)$ is available for a vertex when it moves to $V_S$ or $V_D$. Additionally, when an almost-clique collapses, a total credit of $\Theta(\Delta^2)$ is available.

All in all, we exploit various sparsity properties and the structure of almost-cliques, together with intricate charging arguments to obtain $\ti(\frac{1}{\eps^4})$ update time. Our algorithm ensures that the desired properties of the sparse-dense decomposition are maintained at any given point in time, which allows us to efficiently maintain a $(\Delta+1)$-coloring on top of it.

\subsection{Subroutines: Friend Edges and Dense Vertices}
Before providing the full algorithm, we describe some subroutines which we will utilize extensively in our main algorithm. Throughout this section, let $\varepsilon=\varepsilon'+\tau$ for a small constant $\tau>0$ where $\tau\leq\varepsilon'$ is an accuracy parameter that will be set later to enable desired bounds on the running time. Note that by Definition \ref{defn: friendedge}, an $\varepsilon'$-friend edge is also a $\varepsilon$-friend edge.

\subsubsection{Determining Friend Edges} For any vertex $v$, we maintain a set $N_{\varepsilon}(v)$ such that any vertex $u$ is in $N_{\varepsilon}(v)$ if and only if $(u,v)$ is an $\varepsilon$-friend edge with high probability. We give a randomized subroutine, \textsc{Determine-Friend}, which on input a pair $(u,v)$ and parameters $\varepsilon>0, \tau>0$ determines whether $u$ and $v$ share at least $(1-\varepsilon + \tau)\Delta$ common neighbors, with high probability. Subsequently, the lists $N_{\varepsilon}(u)$ and $N_{\varepsilon}(v)$ are updated. 
\begin{algorithm}[h]
\caption{\textsc{Determine-Friend}$((u,v), \varepsilon,\tau)$}
\begin{algorithmic}[1]
\State $\varepsilon'\leftarrow \varepsilon-\tau$.
\State Sample $k$ neighbors $w_{1}, w_{2},...,w_{k}$ uniformly and independently at random from $N(u)$.
\State Set $Z_j=1$ if $w_{j}\in N(u)\cap N(v)$, and $0$ otherwise.
\If {$T=\frac{\Delta}{k}\sum_{j=1}^k Z_j \geq (1-(\varepsilon-\frac{\tau}{2}))\Delta$}
    \State $N_{\varepsilon}(v)\leftarrow N_{\varepsilon}(v)\cup \{u\}$, $N_{\varepsilon}(u)\leftarrow N_{\varepsilon}(u)\cup \{v\}$.
\Else 
    \State $N_{\varepsilon}(v)\leftarrow N_{\varepsilon}(v)\backslash \{u\}$, $N_{\varepsilon}(u)\leftarrow N_{\varepsilon}(u)\backslash \{v\}$.
\EndIf 
\end{algorithmic}
\end{algorithm}

\begin{lemma}\label{lemma: is-friend}
  Given edge $(u,v)$ and parameters $\varepsilon>0, \tau>0$ as input, $\mathtt{Determine}$-$\mathtt{Friend}$ satisfies the following with probability at least $1-\frac{1}{n^c}$, whenever $k\geq \frac{12c\ln n}{\tau^2}=\Omega(\frac{\ln n}{\tau^2})$, where $c>0$ is an arbitrarily large constant:
  \begin{itemize}
      \item[{\bf (a)}] If $(u,v)$ is a $\varepsilon'$-friend edge, i.e. $|N(u)\cap N(v)|\geq (1-\varepsilon')\Delta$, then $u$ (resp., $v$) is added to $N_{\varepsilon}(v)$ (resp., $N_{\varepsilon}(u)$). 
      \item[{\bf (b)}] If $(u,v)$ is not a $\varepsilon$-friend edge, i.e. $|N(u)\cap N(v)|<(1-\varepsilon)\Delta$, then $u$ (resp., $v$) is removed from $N_{\varepsilon}(v)$ (resp., $N_{\varepsilon}(u)$).
  \end{itemize}  
\end{lemma}

\noindent\begin{proof}%[Proof of Lemma \ref{lemma: is-friend}] 
    We begin by noting that $E[Z_j]=\Pr[w_j\in N(u)\cap N(v)]=\frac{1}{\Delta}|N(u)\cap N(v)|$. Let $Z=\sum_{j=1}^k E[Z_j]$. Thus, $E[Z]=\sum_{j=1}^k E[Z_i]=\frac{k}{\Delta}|N(u)\cap N(v)|$. Moreover, $T=\frac{\Delta}{k}Z$ and by linearity of expectation, $E[T]=\frac{\Delta}{k}E[Z]=|N(u)\cap N(v)|$. 
    
    Let $E_1$ denote the event that $\textsc{Determine-Friend}$ removes $u$ (resp. $v$) from $N_{\varepsilon}(v)$ (resp. $N_{\varepsilon}(u)$) when $|N(u)\cap N(v)|\geq (1-\varepsilon')\Delta$. Let $E_2$ denote the event that $\textsc{Determine-Friend}$ adds $u$ (resp. $v$) to $N_{\varepsilon}(v)$) (resp. $N_{\varepsilon}(u)$) when $|N(u)\cap N(v)|\leq (1-\varepsilon'-\tau)\Delta=(1-\varepsilon)\Delta$.  We give a lower bound on $k$ such that the probability is upper bounded by $\frac{1}{n^c}$.
    
   We establish part (b) first.  By Lines~4--7 of the algorithm, note that $\Pr[E_1]\leq \Pr[T\geq (1-\varepsilon'-\frac{\tau}{2})\Delta]$ (assuming that $|N(u)\cap N(v)| < (1-\varepsilon)\Delta]$). For any $\delta>0$, we have
    \begin{align*}
    \Pr[T-E[T]\geq \delta E[T]] =\Pr\left[\frac{\Delta}{k}Z-\frac{\Delta}{k}E[Z] \geq \frac{\delta\Delta}{k}E[Z]\right]
    =\Pr[Z-E[Z]\geq \delta E[Z]]
    \end{align*}
    Applying a standard Chernoff bound for i.i.d. Bernoulli random variables (see \cite{mitzenmacher2017probability}), the above probability is at most $e^{-\frac{\delta^2E[Z]}{3}}$. Setting $\delta=\frac{\tau}{2(1-\varepsilon)}$ and noting that $E[Z]=\frac{k}{\Delta}(1-\varepsilon)\Delta=k(1-\varepsilon)$, we obtain
    \[
        \Pr[T-E[T]\geq \delta E[T]] = \Pr[Z-E[Z]\geq \delta E[Z]]
                         \leq e^{-\frac{\tau^2}{4(1-\varepsilon)^2}\cdot\frac{k(1-\varepsilon)}{3}} 
                       = e^{-\frac{\tau^2k}{12(1-\varepsilon)}}.
                                \]
    By setting $k\geq \frac{12c(1-\varepsilon)\ln n}{\tau^2}$, we conclude that $\Pr[\neg E_1]\leq \frac{1}{n^c}$, thus completing the proof for part (b). 
    
    Similarly, we establish part (a) of the lemma.  Suppose 
    $|N(u)\cap N(v)|\ge(1-\varepsilon')\Delta$.  By Lines~4--7, $E_2$ happens exactly when $T\leq (1-\varepsilon'-\frac{\tau}{2})\Delta$, which is equivalent to $T \leq (1 + \delta) E[T]$ for $\delta = \frac{\tau}{2(1-\varepsilon')}$.  Again, by applying a standard Chernoff bound for i.i.d.\ Bernoulli random variables, we derive
    %By a similar calculation as above and noting that $\Pr[E_2]\leq \Pr[T\leq (1-\varepsilon'-\frac{\tau}{2})\Delta]$,  we have that for $\delta\geq\frac{\tau}{2(1-\varepsilon')}$,
    \[
        \Pr[E_2] = \Pr[T\leq (1+\delta)E[T]] = \Pr[Z\leq(1+\delta)E[Z]]
                                \leq e^{-\frac{\tau'^2}{4(1-\varepsilon')^2}\cdot \frac{k(1-\varepsilon')}{3}} = e^{-\frac{\tau'^2k}{12(1-\varepsilon')}}.
    \]
    By setting $k\geq \max\left\{\frac{12c(1-\varepsilon)\ln n}{\tau^2},\frac{12c(1-\varepsilon')\ln n}{\tau^2}\right\}$, we obtain $\Pr[E_2]\leq \frac{1}{n^c}$. Setting $k=\frac{12c\ln n}{\tau^2}$, we obtain the desired upper bound on $\Pr[E_2]$, completing the proof of part (b).
\end{proof}

\subsubsection{Determining Dense Vertices}
For every vertex $v$, we maintain boolean variables is-dense($v,d)$ which is 1 if $v$ is $d$-dense and 0 otherwise for various values of $d$, specified later.  We give a subroutine $\textsc{Determine-Dense}$ which, given as input a $(\varepsilon-\tau)$-dense vertex $v\in V$ and parameters $\varepsilon>0, \tau>0$, determines if $v$ is $\varepsilon$-dense. It is implemented as follows: for each neighbor $u\in N(v)$, run $\textsc{Determine-Friend}((u,v),\varepsilon, \tau)$. If $|N_{\varepsilon}(v)|\geq (1-\varepsilon)\Delta$, is-dense$(v,\varepsilon)$ is set to 1 and 0 otherwise. We also maintain the set $V_{\varepsilon}$ of $\varepsilon$-dense vertices containing all vertices in $V$ for which is-dense$(v,\varepsilon)=1$, i.e. all vertices which are at most $\varepsilon$ dense.

\begin{algorithm}[h]
\caption{$\textsc{Determine-Dense}(v,\varepsilon,\tau)$}
\begin{algorithmic}[1]
    \For{each vertex $u\in N(v)$}
        \State Run $\textsc{Determine-Friend}((u,v),\varepsilon, \tau)$.
\EndFor 
    \If{$|N_{\varepsilon}(v)|\geq (1-\varepsilon)\Delta$}
        \State is-dense$(v,\varepsilon)=1$.
        \State $V_{\varepsilon}\leftarrow V_{\varepsilon}\cup \{v\}$.
    \Else 
        \State is-dense$(v,\varepsilon)=0$.
        \State $V_{\varepsilon}\leftarrow V_{\varepsilon}\backslash \{v\}$.
    \EndIf
\end{algorithmic}
\end{algorithm}

%\rrcomment{If the guarantee is in terms of $\varepsilon$, why include $\tau$ in the input to the subroutine?}
\begin{lemma}\label{lemma: densevert}
    Given an $(\varepsilon-\tau)$-dense vertex $v$ (hence, $\varepsilon$-dense by definition) and parameters $\varepsilon, \tau$ as input, \textsc{Determine-Dense} takes $O(\frac{\Delta\ln n}{\tau^2})$ time and correctly determines if $v$ is $\varepsilon$-dense 
     with probability at least $1-\frac{1}{n^{c-1}}$. Here, $1-\frac{1}{n^c}$ is the success probability of \textsc{Determine-Friend} on inputs $(u,v), \varepsilon$ and $\tau$ as in Lemma \ref{lemma: is-friend}.
\end{lemma}
\begin{proof}
    The running time of $\textsc{Determine-Friend}((u,v), \varepsilon,\tau)$ is $O(\frac{\ln n}{\tau^2})$. Thus, \textsc{Determine-Dense} takes $O(\frac{\Delta\ln n}{\tau^2})$ time.
    Note that any call to $\textsc{Determine-Friend}((u,v), \varepsilon,\tau)$ correctly determines whether $(u,v)$ is a $(\varepsilon-\tau)$-friend edge with probability at least $1-\frac{1}{n^c}$ by Lemma \ref{lemma: is-friend}. If $v$ has at least $(1-(\varepsilon-\tau))\Delta<\Delta\leq n$ such friends $u$, then the probability that any of these calls incorrectly determines friend status is at most $n\cdot\frac{1}{n^c}=\frac{1}{n^{c-1}}$. Thus, it follows that if $v$ is $(\varepsilon-\tau)$-dense, then $\textsc{Determine-Dense}(v,\varepsilon,\tau)$ is correct with probability least $1-\frac{1}{n^{c-1}}$. 
    
    By a similar argument, if $v$ is $\varepsilon$-sparse then for any edge $(u,v)$ which is not a $\varepsilon$-friend edge, the probability that $\textsc{Determine-Friend}((u,v), \varepsilon,\tau)$ correctly determines this is $1-\frac{1}{n^c}$ by Lemma \ref{lemma: is-friend}. The probability that this holds for all friend edges $(u,v)$ which are not $\varepsilon$ friend edges is at least $1-\frac{1}{n^{c-1}}$. Thus, if $v$ is $\varepsilon$-sparse, then \textsc{Determine-Dense}$(v,\varepsilon,\tau)$ sets is-dense$(v,\varepsilon)=0$. 
\end{proof}
\eat{
\begin{corollary}\label{corollary: notdense}
If a call to \textsc{Determine-Dense}$(v,\varepsilon, \tau)$ does not add $v$ to $V_{\varepsilon}$, then $v$ is $(\varepsilon-\tau)$-sparse with probability at least $1-\frac{1}{n^{c-1}}$.
\end{corollary}

\begin{proof}
  The proof follows from Lemma \ref{lemma: densevert} and the definition of dense vertices. If $v$ was $(\varepsilon-\tau)$-dense then $v$ is added to $V_{\varepsilon'}$ with probability at least $1-\frac{1}{n^{c-1}}$. Thus, if $v$ is not added to $V_{\varepsilon'}$ it is $(\varepsilon-\tau)$-sparse with high probability.
\end{proof}}

\subsubsection{Fully Dynamic Algorithm to Maintain Friend Edges}\label{sec: fdalgtomaintainfriendedges}
 Given the subroutines in the previous section for determining $\varepsilon$-friend edges and $\varepsilon$-dense vertices, we give a fully dynamic algorithm $\textsc{Maintain-Friends}((u,v),\varepsilon,\tau)$ which takes parameters $\tau>0, \varepsilon>0$ as input and, maintains for all vertices $v$ sets $N_{i\varepsilon}(v)$ for $i\in \{1,2,3\}$.  To simplify notation, we use $N_{i\varepsilon}(v)$ and $N_{i}(v)$ interchangeably.  By definition, $N_{3}(v)\subseteq N_{2}(v)\subseteq N_{1}(v)$. Our algorithm satisfies the following properties at any point in time. 
 
 \begin{enumerate}
     \item If $u\in N_{i\varepsilon}(v)$ for any $v\in V$, $i\in \{1,2,3\}$ then $(u,v)$ is a $(i\varepsilon+\tau)$ friend. 
     \item If is-dense($v,i\varepsilon)=1$ for any $v\in V,$ $i\in \{1,2,3\}$ then $v$ is $(i\varepsilon+\tau)$-dense. 
 \end{enumerate}

We now describe the high level ideas of our algorithm. For a vertex $v$, an edge update is a \textit{direct} update to $v$ if it is incident to $v$. Data structures for $v$ and neighbors $u\in N(v)$ are updated after every $\frac{\tau\Delta}{8}$ direct updates to $v$ by invoking subroutine \textsc{Update}$(v)$. An update is called a Type 1 update for $v$ if $\frac{\tau\Delta}{8}$ direct updates to $v$ have taken place since the last call to \textsc{Update}$(v)$. For all $v\in V$, counters direct$(v)$ and indirect$(v)$ are maintained, which count the number of direct and \textit{indirect} updates, defined below since the last call to \textsc{Update}$(v)$. Whenever direct$(v)=\frac{\tau\Delta}{8}$ (i.e. a Type 1 update for $v$), the procedure \textsc{Update}($v$) is called. After $\textsc{Update}(v)$ is executed, direct$(v)$ is reset to $0$ and counters indirect$(u)$ are incremented for all $u\in N(v)$. 

An update which leads to indirect$(u)$ being incremented as a result of a Type 1 update to some neighbor $v\in N(u)$ is called an \textit{indirect} update. Whenever indirect$(v)=\frac{\tau\Delta}{8}$ for any vertex $v$, procedure $\textsc{Update}(v)$ is called. Such an update is called a Type 2 update for $v$ (note that this is \textit{not} a direct update). Thereafter, indirect($v$) is reset to 0. 

For both Type 1 and Type 2 updates for any vertex $v$, $\textsc{Update}(v)$ is called. The crucial difference between the two is that, following a Type 2 update for $v$ the counters indirect$(u)$ are \textit{not} incremented for any $u\in N(v)$. On the other hand, following a Type 1 update for $v$, the counters indirect$(u)$ are incremented for all neighbors $u\in N(v)$. 

In the following, we give the pseudo-code of our algorithm $\textsc{Maintain-Friends}$, and subroutine \textsc{Update}. 

\begin{algorithm}[h]
\caption{\textsc{Maintain-Friends}$((u,v), \varepsilon, \tau)$}
    \begin{algorithmic}[1]
        \State Update $N(u), N(v)$ accordingly. 
        \State direct$(u)\leftarrow$direct$(u)+1$, direct$(v)\leftarrow$direct$(v)+1$.
        \For{$i\in \{1,2,3\}$}
            \If{$(u,v)\in E$} \Comment{The case when $(u,v)$ is an edge insertion.}
            \State \textsc{Determine-Friend}$((u,v), i\varepsilon, \frac{\tau}{2})$.
            \Else \Comment{The case when $(u,v)$ is an edge deletion.}
            \State $N_{i}(u)\leftarrow N_{i}(u)\backslash \{v\}$.
             \State $N_{i}(v)\leftarrow N_{i}(v)\backslash \{u\}$.
             \EndIf
        \EndFor 
        \State $U\leftarrow \emptyset$.
        \For{$w\in \{u,v\}$ s.t. direct$(w)=\frac{\tau\Delta}{8}$} 
            \State $\textsc{Update}(w)$. \Comment{$(u,v)$ is a Type 1 update for $w$.}
            \State direct$(w)=0$.
            \State $U\leftarrow U \cup \{w\}$.
        \EndFor
        \For {$z \in \bigcup\limits_{y\in U} N(y)$}
            \State indirect$(z)\leftarrow$indirect$(z)+1$. 
            \If{indirect$(z)=\frac{\tau\Delta}{8}$}
               \State $\textsc{Update}(z)$.  \Comment{$(u,v)$ is a Type 2 update for $z$.}
               \State $U\leftarrow U\cup \{z\}$.
               \State indirect$(z)\leftarrow 0$.
            \EndIf 
        \EndFor
        \State \textbf{return} $U$.
\end{algorithmic}
\end{algorithm}

\begin{algorithm}[h]
    \caption{\textsc{Update}($v$)}
    \begin{algorithmic}[1]
        \For {$i\in \{1,2,3\}$}
            \State \textsc{Determine-Dense}$(v,i\varepsilon, \frac{\tau}{2})$. 
        \EndFor 
    \end{algorithmic}
\end{algorithm}

\eat{We say that a vertex $v$ is at least $d$-sparse if it has at most $(1-d)$, $d$-friends. We prove the following theorem.
}
\begin{lemma}\label{lemma: maintain-friends}
    Algorithm $\textsc{Maintain-Friends}$ takes $O(\frac{\ln n}{\tau^4})$ amortized update time and, invoking Algorithm $\textsc{Maintain-Friends}((u,v), \varepsilon, \tau)$ after every edge update $(u,v)$ satisfies the following properties for any vertex $v\in V$, $i\in \{1,2,3\}$ with high probability at any time:
    \begin{itemize}
        \item For any vertex $v$, if $u\in N_i(v)$ then $u$ is a $(i\varepsilon+\tau)$-friend of $v$.
        \item For any vertex $v$, if $u$ is a $(i\varepsilon - \tau)$-friend of $v$, then $u \in N_i(v)$.
        \item For any vertex $v$, if $v\in V_i$ then $v$ is $(i\varepsilon+\tau)$-dense.
        \item For any vertex $v\notin V_i$, $v$ is $(i\varepsilon-\frac{3\tau}{4})$-sparse.
        \end{itemize}
\end{lemma}

\begin{proof}
Fix $c>0$ to be a large constant in Lemma \ref{lemma: is-friend} so that all calls to \textsc{Determine-Friend} correctly determine friend statuses. We prove the theorem statement conditioned on this event.

\noindent\textit{Analysis of running time.} We first prove the running time guarantee of $\textsc{Maintain-Friends}$. Updating data structures of endpoints of the updated edge takes $O(1)$ time, while a single call to $\textsc{Determine-Friend}$ takes $O(\frac{\ln n}{\tau^2})$ time by Lemma \ref{lemma: is-friend}. To bound the total update time, we use a charging argument-- each direct update $(u,v)$ is charged an amount $\Theta(\frac{\ln n}{\tau^4})$. We show this is sufficient to pay for Type 1 and Type 2 updates.

Our charging scheme works as follows. Both endpoints $u,v$ of an updated edge $(u,v)$ receive a credit of $\Theta(\frac{\ln n}{\tau^4})$. Consider a Type 1 update for $v$. After $\frac{\tau\Delta}{8}$ direct updates to $v$, the total credit is at least $\Theta(\frac{\Delta\ln n}{\tau^3})$. This credit is used to pay for the cost of \textsc{Update}$(v)$, which is $\Theta(\frac{\Delta\ln n}{\tau^2})$. Moreover, for all neighbors $w\in N(v)$, a credit of $\Theta(\frac{\ln n}{\tau^3})$ is assigned and indirect$(w)$ is incremented. As a result, a Type 1 update has an amortized update time of $O(\frac{\Delta \ln n}{\tau^4})$.

Next, we analyze the amortized cost incurred of a Type 2 update for any $v$. Note that whenever indirect$(v)$ is incremented, some neighbor $w$ is involved in a Type 1 update and gives a credit of $\Theta(\frac{\ln n}{\tau^3})$ to $v$. Thus, for a Type 2 update for $v$ there is a total credit of at least $\frac{\tau\Delta}{8}\cdot \Theta(\frac{\ln n}{\tau^3})=\Theta(\frac{\Delta\ln n}{\tau^2})$. This credit is used to pay for the cost of $\textsc{Update}(v)$ which is $\Theta(\frac{\Delta\ln n}{\tau^2})$ by Lemma \ref{lemma: densevert}. Note that for a Type 2 update to a vertex $v$, we do not leave a credit for any neighbors. Thus, the cost to handle Type 2 updates is completely paid for by charging every direct update $\Theta(\frac{\ln n}{\varepsilon^4})$. This yields an amortized update time of $\Theta(\frac{\ln n}{\tau^4})$.

\noindent\textit{Proof of Correctness.}
We begin by noting that for any vertex $v$, the sum direct$(v)+$indirect$(v)$ is at most $\frac{\tau\Delta}{4}$. Let $N_{i}(v)$ for any $i\in \{1,2,3\}$ be the maintained lists. For an edge insertion $(u,v)$, \textsc{Determine-Friend}$((u,v), i\varepsilon, \frac{\tau}{2})$ is called and $u$ is added to $N_{i}(v)$ iff $|N(u)\cap N(v)|\geq (1-i\varepsilon)\Delta$ by the second guarantee given by Lemma \ref{lemma: is-friend}. On the other hand, for an edge deletion $(u,v)$, $u$ is removed from $N_{i}(v)$ for all $i$. 

We argue that a vertex $u\in N_{i}(v)$, then $|N(u)\cap N(v)|\geq (1-(i\varepsilon+\tau))\Delta$-- hence, proving that $u$ is a $(i\varepsilon+\tau)$-friend of $v$. This is clearly the case when either $\textsc{Update}(v)$ or $\textsc{Update}(v)$ is invoked or \textsc{Determine-Friend} is called when $(u,v)$ is inserted at which point $|N(u)\cap N(v)|\geq (1-i\varepsilon)\Delta$ by Lemma \ref{lemma: is-friend}. For, $|N(u)\cap N(v)|$ to decrease by at least $\frac{\tau\Delta}{4}$, there must be at least $\frac{\tau\Delta}{4}$ direct updates incident to at least one of $u$ or $v$. But any time there are at least $\frac{\tau\Delta}{8}$ direct updates incident to either $u$ or $v$, the subroutine $\textsc{Determine-Friend}((u,v),i\varepsilon, \frac{\tau}{2})$ is called. Thus, between any two consecutive invocations of $\textsc{Determine-Friend}((u,v), i\varepsilon, \frac{\tau}{2})$ there are at most $\frac{\tau\Delta}{4}$ vertices in $|N(u)\cap N(v)|$ which could be deleted. Thus, for any vertex $u\in N_{i}(v)$, it holds that $|N(u)\cap N(v)|\geq (1-i\varepsilon)\Delta-\frac{\tau\Delta}{4}\geq (1-(i\varepsilon+\frac{\tau}{4}))\Delta>(1-(i\varepsilon+\tau))\Delta$.

We next argue that if $u$ is a $(i\varepsilon-5\tau/4)$-friend of $v$, then $u \in N_i(v)$.  If $u$ is a $(i\varepsilon-5\tau/4)$-friend of $v$, then $|N(u)\cap N(v)|\geq (1-(i\varepsilon-5\tau/4))\Delta$.  When $(u,v)$ was inserted, \textsc{Determine-Friend} is called.  If at that time $|N(u)\cap N(v)|\geq (1-(i\varepsilon-\tau))\Delta$, then by Lemma \ref{lemma: is-friend}, $u$ is added to $N_i(v)$ with high probability.  If not, then $|N(u)\cap N(v)|$ has increased by at least $\frac{\tau \Delta}{4}$ since then.  As we argued above, for this to happen, there must be at least $\frac{\tau\Delta}{8}$ direct updates incident to either $u$ or $v$, implying that the subroutine $\textsc{Determine-Friend}((u,v),i\varepsilon, \frac{\tau}{2})$ is called at some time between the period when $|N(u)\cap N(v)|$ increases from $(1-(i\varepsilon-\tau))\Delta$ to $(1-(i\varepsilon-5\tau/4))\Delta$.   During that call, by Lemma~\ref{lemma: is-friend}, $u$ is added to $N_i(v)$ with high probability and remains in that list.

Next, we prove that for any vertex $v$, if the variable is-dense$(v,i\varepsilon)=1$ then $v$ is $(i\varepsilon+\tau)$ dense. This is true at the time when is-dense($v,i\varepsilon)$ is set to 1 by subroutine \textsc{Determine-Dense}. Between any two consecutive invocations of \textsc{Update}($v$) due to Type 1 updates for $v$, the number of friends in any list $N_{i}(v)$ changes by at most $\frac{\tau\Delta}{8}$. On the other hand, the number of friends in any list $N_{i}(v)$ may only change by at most $\frac{\tau\Delta}{8}$ between any two Type 2 updates for $v$. Thus, between any two calls to $\textsc{Determine-Dense}(v,i\varepsilon,\frac{\tau}{2})$ the number of friends that changes is at most $\frac{\tau\Delta}{4}$. Thus, between any two calls, $v$ has at least $(1-(i\varepsilon+\frac{\tau}{4}))\Delta>(1-(i\varepsilon+\tau))\Delta$, $(i\varepsilon+\tau)$ friends, i.e. $v$ is  $(i\varepsilon+\tau)$-dense. 

Finally, we prove that if vertex $v\notin V_i$, then $v$ is $(i\varepsilon-\frac{3\tau}{4})$-sparse. This is true immediately after the call to $\textsc{Determine-Dense}(v, i\varepsilon, \frac{\tau}{2})$ by Lemma \ref{lemma: densevert} since $v$ must have less than $(1-(i\varepsilon-\frac{\tau}{2}))\Delta$ friends. Similar to the above argument, between any two calls to \textsc{Determine-Dense}$(v,i\varepsilon,\frac{\tau}{2})$, the number of friends can increase by at most $\frac{\tau\Delta}{4}$. Thus, between any two calls to \textsc{Determine-Dense}$(v,i\varepsilon, \frac{\tau}{2})$, $v$ has at most $(1-i\varepsilon+\frac{\tau}{2})\Delta+\frac{\tau\Delta}{4}= (1-(i\varepsilon-\frac{3\tau}{4}))\Delta$, $(i\varepsilon-\frac{3\tau}{4})$ friends implying that $v$ is $(i\varepsilon-\frac{3\tau}{4})$-sparse. 
\end{proof}

\subsection{Algorithm for Maintaining Sparse-Dense Decomposition}\label{sec: fullydynamicdecomposition} Given the subroutines in the previous sections, we present an algorithm \textsc{Update-Decomposition}$\allowbreak((u,v), \varepsilon$,$\tau)$ which maintains a sparse-dense decomposition satisfying the properties outlined in Theorem \ref{fd-decomposition}. As before, $\varepsilon$ and $\tau$ are input parameters.  Our algorithm utilizes the subroutine \textsc{Maintain-Friends} to handle any edge update, which maintains the sets $N_i(v)$ for all $v\in V$ containing the set of $(i\varepsilon+\tau)$ friends of $v$, and the sets $V_i$ of $(i\varepsilon+\tau)$-dense vertices respectively for all $i\in \{1,2,3\}$. For notational convenience $c_i=i\varepsilon+\tau$, such that $c_i-c_{i-1}=\varepsilon$ for $i\in \{2,3\}$. 
\\ \\
\noindent\textbf{High Level Overview. }
 Our algorithm  maintains a sparse-dense decomposition of the graph $G=(V,E)$ by partitioning $V$ into $V_S$ and $V_D$ such that $V_S\subseteq V\backslash V_1$. By Lemma \ref{lemma: maintain-friends}, it follows that any vertex not in $V_1$ is $(\varepsilon-\frac{3\tau}{4})$-sparse. That is, $V_S\subseteq \vs_{(\varepsilon-\frac{3\tau}{4})}$, is a subset of vertices in $G$ that are $(\varepsilon-\frac{3\tau}{4})$-sparse at any given point in time. Moreover, we maintain $V_D\subseteq V_3$, i.e. $V_D\subseteq \vd_{(3\varepsilon+\tau)}$ is a subset of vertices in $G$ that are $c_3=(3\varepsilon+\tau)$-dense at any given point in time. Furthermore, our algorithm maintains almost-cliques on $V_D$ induced by $c_3$-friend edges. 

Our algorithm for maintaining a sparse-dense decomposition achieves its objectives by ensuring that four invariants hold at the start of every step.  Two invariants concern vertices while the remaining two concern almost-cliques.
\newcommand{\Density}{{\mbox{{\textsc{Density}}}}}
\newcommand{\Friendship}{{\mbox{{\textsc{Friendship}}}}}
\newcommand{\Size}{{\mbox{{\textsc{Size}}}}}
\newcommand{\Connectedness}{{\mbox{{\textsc{Connectedness}}}}}
\begin{itemize}
\item \Density$(v)$:  If $v$ is in $V_D$, $v$ is in $V_3$; otherwise, $v$ is not in $V_1$.
\item \Friendship$(v)$:  If $v$ is in almost-clique $C$, then $v$ has at least $(1 - c_3)\Delta$ neighbors from $N_3(v)$ that are currently in $C$ or were members of $C$ at some time since $C's$ creation.
\item \Size$(C)$: Almost-clique $C$ has size at least $(1 - c_3)\Delta$ and at most $(1 + 3c_3)\Delta$.
\item \Connectedness$(C)$: The $c_3$-friend edges in almost-clique $C$ form a spanning connected subgraph of $C$.
\end{itemize}

 \newcommand{\DenseMove}[1]{{\mbox{{\textsc{Dense-Move}}}}$(#1)$}
 \newcommand{\DenseMoveNoArgs}{{\mbox{{\textsc{Dense-Move}}}}}
 \newcommand{\SparseMove}[1]{{\mbox{{\textsc{Sparse-Move}}}}$(#1)$}
  \newcommand{\SparseMoveNoArgs}{{\mbox{{\textsc{Sparse-Move}}}}}

 Our algorithm determines when to move vertices between $V_S$ and $V_D$ so as to maintain the invariants.  Starting with an empty graph, all vertices are in $V_S$, and $V_D\coloneqq \emptyset$.  After any edge update, if there is a vertex $v$ in $V_S$ that is inserted to $V_1$, we call \DenseMove{v}, which operates as follows: 
 \begin{itemize}
 \item If $v$ has a $c_1$-friend $u$ in $V_D$, move $v$ and all of the $c_1$-friends of $v$ that are in $V_S$ to $V_D$ into the same almost-clique as $u$, updating data structures as necessary. (Type 1 dense move)

 \item Otherwise, move $v$ and all its $c_1$-friends from $V_S$ to $V_D$, forming a new almost-clique, updating data structures as necessary. (Type 2 dense move)
 \end{itemize}
 Similarly, after any edge deletion, if there is a vertex $v$ in $V_D$ that is not in $V_3$ or has fewer than $(1 - c_3)\Delta$ neighbors from $C$ in $N_3(v)$ (i.e., violating \Density\ or \Friendship\ invariants), then we call \SparseMove{v}, which operates as follows:
 \begin{itemize}
     \item Move $v$ move from $V_D$ to $V_S$, updating data structures as necessary.
 \end{itemize}
Finally, for a suitably small constant $\nu > 0$, when at least $\nu \Delta$ vertices from an almost-clique $C$ have been moved from $V_D$ to $V_S$ (i.e., $\Omega(\Delta)$ calls have been made to \SparseMoveNoArgs\ for vertices in $C$), we move all vertices in $C$ to $V_S$, delete $C$, and reconstruct a new almost-clique (if necessary) by calling \DenseMoveNoArgs\ on the subset of these vertices that are in $V_1$.  This ensures that the \Size\ and \Connectedness\ invariants continue to hold.

The analysis of our algorithm uses a charging argument for vertex moves. This yields $\ti(\frac{1}{\varepsilon^4})$ amortized update time. Crucially, we ensure that the number of non-neighbors of every vertex $v\in V_D$ in its almost-clique has \textit{low-adjustment complexity}. Finally, using the standard technique of periodically rebuilding data structures, we point out how to obtain $\ti(\frac{1}{\varepsilon^4})$ \textit{worst-case} update time, incurring only a constant factor increase in update time. 
\\
\\
\noindent\textbf{Data Structures.} We maintain additional data structures. Each vertex $v$ maintains lists $N_S(v)$ and $N_D(v)$ corresponding to its neighbors in $V_S$ and $V_D$ respectively, s.t. $N(v)=N_S(v)\cup N_D(v)$. The algorithm maintains a list of all almost-cliques and each vertex $v$ maintains a pointer to the almost-clique $C$ it is part of. For all $v\in V_D$, $v$ maintains a list $\b{E_C(v)}$ of \textit{non-edges} in its almost-clique $C$ and a list $N'(v) = N_3(v) \cap C$. For all $v\in V$ and all almost cliques $C$, we maintain the list $N_C(v)$ containing neighbors of $v$ in $C$. 
\\
\\
\noindent\textbf{The Full Algorithm.} We present algorithm \textsc{Update-Decomposition} and subroutines \DenseMoveNoArgs\ and \SparseMoveNoArgs\ in the following. Recall that on any update, the procedure \textsc{Maintain-Friends} returns a set $U$ of vertices for which \textsc{Update} is invoked. For a vertex $v\in U$, we call $v$ an \textit{affected} vertex if $v$ is added or removed from one of $V_1$, $V_2$ or $V_3$ (this can be done by labelling a vertex $v$ in $U$ if is-dense$(v,i\varepsilon)$ changes for any $i\in \{1,2,3\}$ when \textsc{Update}($v$) is called). Our algorithm processes the set of all affected vertices, since these vertices may need to be moved. We say $v$ is part of a sparse (resp. dense) move if it is moved from $V_D$ to $V_S$ (resp. $V_S$ to $V_D$). 

\newcommand{\sparseMoveCount}[1]{\sigma(#1)}

\begin{algorithm}[h]
\caption{\textsc{Update-Decomposition}$(\varepsilon, \tau, (u,v))$}
\begin{algorithmic}[1]
        \State $U\leftarrow \textsc{Maintain-Friends}((u,v), \varepsilon,\tau)$. 
        \If {$u,v\in V_D$ and are in the same almost-clique $C$}
            \State Update $N_C(u), N_C(v), N'_C(u), N'_C(v), \b{E_C(u)}, \b{E_C(v)}$.
        \EndIf 
        \If{$(u,v)$ is an edge insertion} 
            \For{$w\in U$}
                \If{$w\in V_1$ and $w\in V_S$} Call \DenseMove{w}.
                \eat{\If{$w\in V_1$ and $w\notin V_1$ prior to this update}}
                \EndIf 
            \EndFor
        \EndIf
        \If{$(u,v)$ is an edge deletion}
            \State ${\cal C} \leftarrow \emptyset$. $U' \leftarrow \emptyset$.
            \For{$w\in U$}
                \State Let $C$ be the almost-clique of $w$.
                \If{$w \in V_D$ and ($w\notin V_3$ or $|N'(w)| \le (1 - c_3)\Delta$)}
                    \State   Increment $\sparseMoveCount{C}$.\If{$\sparseMoveCount{C} \ge \alpha \tau \Delta$} ${\cal C} \leftarrow {\cal C} \cup \{C\}$.
                    \Else \, Call \SparseMove{w}.
                    \EndIf
                \EndIf 
            \EndFor
            \For{$C \in {\cal C}$}
                \For{$w \in C$}
                    \State $V_S = V_S \cup \{w\}$. $V_D = V_D \setminus \{w\}$.  $U' = U' \cup \{w\}$.
                    \State Update $N_S(w), N_D(w)$, delete $\b{E_C(w)}, N'(w)$, and \textbf{for} $z\in N(w)$: delete $N_C(z)$.
                \EndFor
                \State Remove $C$ from list of almost-cliques.
            \EndFor
            \For{$w \in U'$}
                \If{$w\in V_1$ and $w \in V_S$} Call \DenseMove{w}.
                \EndIf
            \EndFor
        \EndIf 
\end{algorithmic}
\end{algorithm}

Algorithm $\textsc{Update-Decomposition}$ on edge update $(u,v)$ calls the subroutine $\textsc{Maintain-Friends}$ which is responsible for maintaining sets $V_1, V_2, V_3$ of $c_1, c_2, c_3$ dense vertices and $N_1(v), N_2(v), N_3(v)$ of $c_1,c_2,c_3$ friends of $v$, for all $v\in V$. The only time a vertex changes its dense status is when \textsc{Update} is invoked. Thus, the set $U$ returned by \textsc{Maintain-Friends} contains the list of all vertices that could potentially be added or removed from one of $V_1, V_2, V_3$. 

If the edge update is an insertion, then 
we check for each vertex $w\in U$ whether a dense move is necessitated. If $w \in U$ was not in $V_1$ before the current update and it is now in $V_1$ yet still in $V_S$, then by Lemma \ref{lemma: maintain-friends}, $w$ is $(\varepsilon+\tau)$-dense. Subsequently, the subroutine \DenseMove{w} is called, and $w$ is moved to $V_D$.  We execute the preceding step for each $w \in U$.  As mentioned above,  \DenseMoveNoArgs\ considers Type 1 and Type 2 dense moves, adding relevant vertices to an existing almost-clique in the former case and creating a new almost-clique in the latter case.

If the edge update is a deletion, then we 
we check for each vertex $w\in U$ whether a sparse move is necessitated.  If $w$ is in $V_D$ as part of almost-clique $C$ but either is not in $V_3$ or has at most $(1 - c_3)\Delta$ neighbors in $N_3(w) \cap C$, then either $w$ is $(\varepsilon - 3\tau/4)$-sparse by Lemma~\ref{lemma: maintain-friends} or it does not have a sufficient number of $c_3$-friends in $C$.  In this scenario, we would like to move $w$ to $V_S$, for which we invoke a call to \SparseMoveNoArgs.  Such a move may, in fact, trigger other sparse moves from the $C$, eventually possibly leading to $C$ becoming empty.  To accommodate such a sequence of moves, we keep track of the number of sparse moves associated with each almost-clique since its inception by maintaining a quantity $\sigma(C)$.  As soon as this number exceeds $\nu \Delta$, for a suitably small constant $\nu>0$, we move all of its vertices to $V_S$ and then consider each vertex that is sufficiently dense for a dense move, possibly creating a new almost-clique.

\eat{On the other hand, if $w \in V_3$ prior to this update, and $w\notin V_3$ after $\textsc{Maintain-Friends}$ is called, by Lemma \ref{lemma: maintain-friends}, it follows that $v$ is $(3\varepsilon-\frac{3\tau}{4})$-sparse, and hence, $c_3$-sparse. Subsequently, the subroutine $\textsc{Move}(w)$ is called and $w$ is moved to $V_S$. }  

We describe how the subroutine \DenseMove{v} works for a vertex $v$. First, the sets $V_D, V_S$ are updated. Next, we check if there exists a $c_1=(\varepsilon+\tau)$-friend $u$ of $v$ in $V_D$. If it is, then, $v$ joins the almost-clique containing $u$ (a Type 1 dense move). Thereafter, all $c_1$-friends of $v$ in $N_1(v)\backslash V_D$, i.e. $c_1$-friends in $V_S$, are added to $C$. We update all data structures of vertices in $U\cup C$, where $U$ is the set of vertices moved to $V_D$ in this step.  For a Type 2 dense move, i.e. when all of $v's$ $c_1$-friends in $N_1(v)$ are in $V_S$, a new almost-clique $C$ containing $v$ and all $c_1$ friends in $N_1(v)$ is created. Data structures of all vertices in $C$ are updated. 

For a sparse move for $v$ (\SparseMove{v}), the sets $V_S, V_D$ are updated and $v$ is removed from almost-clique $C$. Data structures of all remaining vertices in $U$ and $v's$ neighbors are updated.  This completes the description of \textsc{Update-Decomposition} and the subroutines \DenseMoveNoArgs\ and \SparseMoveNoArgs.
%Finally, if $|N_C(u)|<(1-c_3)\Delta$ for any vertex $u\in C$, $u$ is moved from $V_D$ to $V_S$.

\begin{algorithm}[ht]
    \caption{\DenseMove{v} \label{alg: Dense Move}}
    \begin{algorithmic}[1]
    %\Comment{Set $U$ stores vertices which will be possibly moved.}
    \State $V_D\leftarrow V_D\cup \{v\}$, $V_S\leftarrow V_S\backslash\{v\}$.
    \If {there exists $u\in V_D$ s.t. $u\in N_{1}(v)$}\Comment{Type 1 dense move}
        \State Let $C\leftarrow$ almost-clique of $u$.
        \State $C\leftarrow C\cup \{v\}$, $U\leftarrow \{v\}$.
        \State \textbf{for} all $u\in N_1(v)\backslash V_D$: $C\leftarrow C\cup \{u\}$, $U\leftarrow U\cup \{v\}$. 
        \State \textbf{for} all $w\in U\cup C$: Update $N_S(w), N_D(w), \b{E_C(w)}, N_C(w), N'(w)$.
        \State \textbf{for} all $w\in U$: \textbf{for} all $z\in N(w)$: Update $N_C(z)$.
    \Else \Comment{Type 2 dense move}
        \State Create almost-clique $C$ and add to list of almost-cliques. 
        \State $C\leftarrow \{v\} \cup N_1(v)$, $U\leftarrow \{v\} \cup N_1(v)$. $\sparseMoveCount{C} \leftarrow 0$.
        \State \textbf{for} all $w\in U$: Update $N_S(w), N_D(w), \b{E_C(w)}, N'(w)$, and \textbf{for} all $z\in N(w)$: Update $N_C(z)$.
    \EndIf 
 \end{algorithmic}
\end{algorithm}

\begin{algorithm}[ht]
    \caption{\SparseMove{v}}
    \begin{algorithmic}[1]
        \State $U\leftarrow \emptyset$. \Comment{Set $U$ stores vertices which will be possibly moved.}
        \State $C\leftarrow $ almost-clique of $v$.
        \State $U\leftarrow N_C(v)$.
        \State $V_S\leftarrow V_S\cup \{v\}$, $V_D\leftarrow V_D\backslash\{v\}$.
        \State $C\leftarrow C\backslash \{v\}$.
        \State \textbf{for} all $u\in N(v)$: $N_S(u)\leftarrow N_S(u)\cup \{v\}$, $N_D(u)\leftarrow N_D(u)\backslash \{v\}$ and $N_C(u)\leftarrow N_C(u)\backslash \{v\}$.
        \State \textbf{for} all $(u,v)\in \b{E_C(v)}$: $\b{E_C(u)}\leftarrow \b{E_C(u)}\backslash \{(u,v)\}$. 
        \State $\b{E_C(v)}\leftarrow \emptyset$.
        \eat{\State \textbf{for} all $u\in U$: if $N_C(u)<(1-c_3)\Delta$ \textbf{then} call \SparseMove{u}.}
    \end{algorithmic}
\end{algorithm}

\eat{More formally, our algorithm maintains the following Invariant at any given point in time. We call this the $\textsc{Dense}$ Invariant.

\noindent{\underline{\textsc{Dense} Invariant:} For any vertex $v\in V_D$: i) $|N_C(v)|\geq (1-c_3)\Delta$, where $C$ denotes the almost-clique of $v$ and, ii) $v$ is $c_3$-dense.}

\begin{observation}\label{obs: lowerboundcsize}
    If \textsc{Dense-Invariant} holds, then $|C|\geq (1-c_3)\Delta$.
\end{observation}

The next section is devoted to the analysis of algorithm $\textsc{Update-Decomposition}$.
}

\subsection{Analysis}
We prove several lemmas towards a proof of Theorem \ref{fd-decomposition}.  We start with an observation about $\sparseMoveCount{C}$ for any almost-clique $C$.

\begin{observation}
\label{obs:sparse move count}
At any point in time, for any almost-clique $C$, $\sparseMoveCount{C}$ is at most $\nu \Delta$.  Furthermore, if an update is an edge insertion, then $\sparseMoveCount{C}$ does not change during the processing of the update.
\end{observation}
\begin{proof}
    The proof of the second claim is immediate from the fact that only dense moves are executed during the processing of an edge insertion; for any existing almost-clique $C$, $\sparseMoveCount{C}$ does not change during a dense move, and for a new almost-clique, $\sparseMoveCount{C}$ is set to zero and remains so until the end of the step.  
    
    The proof of the first claim is by induction on the number of steps.  The base case is trivial.  Suppose the claim is true at the start of a step.  If the step is an edge insertion, then the induction step follows from the first claim.  Otherwise, the induction step holds since whenever $\sparseMoveCount{C}$ exceeds $\nu \Delta$, $C$ is deleted, and any new almost-clique $C'$ formed during the step has $\sparseMoveCount{C'} = 0$ at the end of the step.
\end{proof}

We next prove that if the \Density\ and \Friendship\ invariants hold, then for a Type 1 dense move for vertex $v$ which becomes $c_1$-dense, all its $c_1$-friends in $V_D$ must be in the same almost-clique for suitably small $\varepsilon$. 

\begin{lemma}
\label{lemma: unique-dense-almost-clique}
If the \Density\ and \Friendship\ invariants are satisfied for every vertex in $V_D$ immediately preceding a  call to \DenseMove{v} and $v$ has at least one $c_1$-friend in $V_D$ (i.e., the condition in line~2 of Algorithm~\ref{alg: Dense Move} is satisfied), then all of the $c_1$-friends of $v$ in $V_D$ reside in
at exactly one almost-clique, assuming $\varepsilon < \frac{1}{8}-\frac{\tau}{2} - \frac{\nu}{4}$.
\end{lemma}
\begin{proof}
Suppose for the sake of contradiction that $v$ has two $c_1$ friends $u, w$ in distinct almost-cliques $C_1$ and $C_2$.
   Consider $C_1$. By the \Friendship\ invariant for $u$ and Observation~\ref{obs:sparse move count}, $u$ has at least $(1-c_3)\Delta - \nu \Delta$ neighbors from $N_3(u)$ in $C_1$.  By Lemma~\ref{lemma: maintain-friends}, this implies that 
   $u$ has $(1-c_3 - \nu)\Delta$ $(c_3 + \tau)$-friends in $C_1$.
   Since, $u$ and $v$ share at least $(1-c_1)\Delta$ neighbors, at least $(1-(c_1+c_3+\nu))\Delta$ of neighbors of $u$ must also be neighbors of $v$ since the maximum degree is bounded by $\Delta$. Moreover, each such common neighbor $w\in N(v)\cap N(u)$ must have at least $(1-(c_1+c_3+\nu))\Delta$ common neighbors with $v$, since the maximum degree is bounded by $\Delta$ and $v$ (resp. $w$) has at least $(1-c_1)\Delta$ (resp. $(1-c_3)\Delta$) common neighbors with $u$. Thus, $(v,w)$ is a $(c_1+c_3+\nu)$-friend edge. 
   
We repeat the same argument for $C_2$, and observe that $v$ has at least $(1-(c_1+c_3+\nu))\Delta$ distinct neighbors in $C_1$, and $(1-(c_1+c_3+\nu))\Delta$ distinct neighbors in $C_2$, for a total of at least $2(1-(c_1+c_3+\nu))\Delta$ neighbors in total. This is a contradiction for $\varepsilon<\frac{1}{8}-\frac{\tau}{2}-\frac{\nu}{4}$. 
\end{proof}

\eat{
\begin{lemma}\label{lemma: unique-dense-almost-clique}
    Let \textsc{Dense} Invariant be maintained at any given point in time, and consider a vertex $v$ which becomes $c_1$-dense. If $v$ has at least one neighbor $u\in V_D$, then there exists a unique almost-clique $C$ containing all $c_1$-friends of $v$ whenever $\varepsilon<\frac{1}{8}-\frac{\tau}{2}$. 
\end{lemma}

\begin{proof}
   Suppose $v$ has exactly one $c_1$-friend in $V_D$. Then we are done. For contradiction, suppose $v$ has two $c_1$ friends $u, w$ in distinct almost-cliques $C_1$ and $C_2$, and assume that \textsc{Dense} Invariant holds.
   
   Consider $C_1$. By \textsc{Dense} Invariant which holds for $u$, $u$ has $(1-c_3)\Delta$, $c_3$ friends in $C_1$. Since, $u,v$ share at least $(1-c_1)\Delta$ neighbors, at least $(1-(c_1+c_3))\Delta$ of neighbors of $u$ must also be neighbors of $v$ since the maximum degree is bounded by $\Delta$. Moreover, each such common neighbor $w\in N(v)\cap N(u)$ must have at least $(1-(c_1+c_3))\Delta$ common neighbors with $v$, since the maximum degree is bounded by $\Delta$ and $v$ (resp. $w$) has at least $(1-c_1)\Delta$ (resp. $(1-c_3)\Delta$) common neighbors with $u$. Thus, $(v,w)$ is a $(c_1+c_3)$-friend edge.
   
We repeat the same argument for $C_2$, and observe that $v$ has at least $(1-(c_1+c_3))\Delta$ distinct neighbors in $C_1$, and $(1-(c_1+c_3))\Delta$ distinct neighbors in $C_2$, for a total of at least $2(1-(c_1+c_3))\Delta$ neighbors in total. This is a contradiction for $\varepsilon<\frac{1}{8}-\frac{\tau}{2}$.
\end{proof}
}

The next lemma shows that any vertex which joins $V_D$ has at least $(1-2c_1)\Delta$ $2c_1$-friends in its almost-clique. 

\begin{lemma}\label{lemma: dense-nbrsatmove}
    Let $\tau < \varepsilon<\frac{1}{8}-\frac{\tau}{2} - \frac{\nu}{4}$.  Then, if \Density\ and \Friendship\ invariants hold for every vertex in $V_D$ before a call to \DenseMoveNoArgs, then any vertex that moves as a result of the call has at least $(1 - 2c_1)\Delta$, $2c_1$-friends in its almost-clique. 
\end{lemma}
\begin{proof}
Let $v$ be a vertex that moves during a \DenseMoveNoArgs\ call, and let $C$ be the almost-clique $v$ joins after the move. We will prove that $v$ has at least $(1-2c_1)\Delta$, $2c_1$-friends in $C$. 
\eat{Since $2c_1=2\varepsilon+2\tau< 3\varepsilon+\tau$ whenever $\varepsilon> \tau$, this would imply that $v$ has at least $(1 - c_3)\Delta$ $c_3$-friends in $C$ after the move, completing the proof the lemma.}

We first consider the case where $v$ moves to $V_D$ as a result of a call to \DenseMove{v}. 
 There are two sub-cases.  First is where $v$ and all its neighbors in $N_1(v)$ form a new almost-clique $C$.  In this sub-case, by Lemma~\ref{lemma: maintain-friends}, $v$ is $c_1$-dense and every vertex in $N_1(v)$ is a $c_1$-friend of $v$; since $v$ is $c_1$-dense, it has at least $(1 - c_1)\Delta$, $c_1$-friends, all of which are in $C$.  The second sub-case is where $v$ and all of its neighbors in $N_1(v) \cap V_S$ move to the almost-clique $C$ containing a $c_1$-friend $u$ of $v$.  By Lemma \ref{lemma: unique-dense-almost-clique}, it holds that after the move all of $v$'s $c_1$-friends in $V_D$ are in the same almost-clique $C$.  Thus, in both sub-cases, $v$ has at least $(1-c_1)\Delta$, $c_1$-friends in $C$, yielding the desired claim of the lemma.

We next consider the case when $v$ is moved from $V_S$ to $V_D$ when \DenseMove{u} is called for a neighbor $u$ of $v$ in $V_1$; $v$ is moved to $V_D$ and joins $u's$ almost-clique $C$. By the argument in the first case, $u$ has at least $(1-c_1)\Delta$ $c_1$-friends in $C$. Let $w$ be any $c_1$-friend of $u$. Then $|N(w)\cap N(v)|\geq (1-(c_1+c_1))\Delta=(1-2c_1)\Delta$ since the maximum degree is bounded by $\Delta$. Since there are at least $(1-c_1)\Delta>(1-2c_1)\Delta$ such friends, $v$ has at least $(1-2c_1)\Delta$, $2c_1$-friends in $C$, completing the desired claim and the proof of the lemma.
\end{proof}

\begin{lemma}\label{lemma: denseinvariantaftermove}
    If $\tau < \varepsilon < \frac{1}{8} - \frac{\tau}{2} - \frac{\nu}{4}$, then the \Density\ and \Friendship\ invariants are maintained immediately after a call to \DenseMoveNoArgs{}. 
\end{lemma}

\begin{proof}
Consider any call \DenseMove{v}.  For the sake of the proof, assume that \Density\ and \Friendship\ invariants hold immediately preceding the call to \DenseMove{v}.  Since only $v$ and its $c_1$-friends in $V_S$ are moved from $V_S$ to $V_D$, both the invariants continue to hold for vertices that are not moved.  Since $v$ is in $V_1$, it satisfies the \Density\ invariant.  Furthermore, every $c_1$-friend of $v$ is $c_2$-dense and hence is in $V_3$ by Lemma~\ref{lemma: maintain-friends}, thus also satisfying the \Density\ invariant.  By Lemma~\ref{lemma: dense-nbrsatmove}, all the vertices that move also satisfy the \Friendship\ invariant.
\end{proof}

\begin{lemma}
    \label{lemma: connectedness invariant}
    If the \Density, \Friendship, and \Connectedness\ invariants hold preceding a call to \DenseMoveNoArgs, then the \Connectedness\ invariant is maintained after the call.
\end{lemma}
\begin{proof}
Consider any call \DenseMove{v}.  For the sake of the proof, assume that \Density\ and \Friendship\ invariants hold immediately preceding the call to \DenseMove{v}.  Since only $v$ and its $c_1$-friends in $V_S$ are moved from $V_S$ to $V_D$, the \Connectedness\ invariant continues to hold for all cliques other than the one containing $v$.  We consider two cases.  First is where $v$ and all its neighbors in $N_1(v)$ form a new almost-clique $C$.  Clearly, $v$ and its $c_1$-friend edges from a spanning tree for $C$, completing the proof for this case.  The second case is where $v$ and all of its neighbors in $N_1(v) \cap V_S$ move to the almost-clique $C$ containing a $c_1$-friend $u$ of $v$.  By our assumption, before the move, the $c_3$-friend edges in $C$ spanned the vertices in $C$.  Since $(u,v)$ is a $c_1$-friend edge and all of the other vertices that move have a $c_1$-friend edge to $v$, it follows that after the move, the $c_3$-friend edges in $C$ span all the vertices in $C$, thus completing the proof. 
\end{proof}

Next, we upper bound the size of any almost-clique $C$ by upper bounding the size of $\b{E_C(v)}$ for any $v\in C$. The proofs of the next two lemmas are similar to the proofs of Lemmas 3.4, 3.9 in \cite{hss}.

\begin{lemma} \label{lemma: densecommonnbrs}  Let $\varepsilon<\frac{1}{15}-\frac{\tau}{3}$. Then, if \Density, \Friendship, and \Connectedness\ invariants hold, then 
for any two vertices $u,v\in V_D$ that are in the same almost-clique $C$, $|N_C(u)\cap N_C(v)|\geq (1-2c_3)\Delta$. 
\end{lemma}
\begin{proof}
    Suppose the \Density, \Friendship, and \Connectedness\ invariants hold.  Then by the \Connectedness\ invariant, for any two vertices $(u,v)$ in $C$, there exists a path $u_0\coloneqq u, u_1\coloneqq v, u_2,\ldots,u_k=v$, where $(u_i, u_{i-1})$ is a $c_3$-friend edge. We prove that $|N(u)\cap N(u_i)|\geq (1-2c_3)\Delta$ at any given point in time for $i\in [k]$. The base case follows since $(u,u_1)$ is a $c_3$-friend edge. By the induction hypothesis $|N(u)\cap N(u_{i})|$ has a $2c_3$-friend edge. Thus,
    \begin{align*}
        |N(u)\cap N(u_{i+1})|&\geq |N(u)\cap N(u_{i}) \cap N(u_{i+1})| \\
                            &=|N(u_i)\cap N(u_{i+1})| + |N(u_i)\cap N(u)|- |(N(u_i)\cap N(u_{i+1}))\cup (N(u_i)\cap N(u))| \\
                            &\geq |N(u_i)\cap N(u_{i+1})| + |N(u_i)\cap N(u)| -|N(u_i)| \\
                            &\geq (1-c_3)\Delta+(1-2c_3)\Delta-\Delta \\
                            &=(1-3c_3)\Delta.
    \end{align*}
    Since $u$ and $u_{i+1}$ are both $c_3$-dense, they have at most $c_3\Delta$ neighbors each which are not $c_3$-friends. Thus, the number of common friends of $u$ and $u_{i+1}$ are at least $|N(u)\cap N(u_{i+1})|-2c_3\Delta\geq (1-3c_3)\Delta-2c_3\Delta=(1-5c_3)\Delta$. In particular, for $c_3=3\varepsilon+\tau<\frac{1}{5}$, $|N(u)\cap N(u_{i+1})|> 0$, so they have at least one common neighbor $x\in C$ which is $c_3$-dense (by the \Density\ invariant) such that $|N(x)\cap N(u)|\geq (1-c_3)\Delta$, and $|N(x)\cap N(u_{i+1})|\geq (1-c_3)\Delta$. Since the maximum degree is $\Delta$, it follows that, 
    \begin{align*}
        |N(u)\cap N(u_{i+1})|\geq (1-c_3)\Delta+(1-c_3)\Delta -\Delta =(1-2c_3)\Delta
    \end{align*}
    completing the proof.
\end{proof}

\begin{lemma}\label{lemma: non-degree}
If \Density, \Friendship, and \Connectedness\ invariants hold, then for any $v\in V_D$ such that $v$ is in almost-clique $C$, $|\b{E_C(v)}|\leq 3c_3\Delta$, whenever $\varepsilon<\frac{1}{15}-\frac{\tau}{3}$.
\end{lemma}
\begin{proof}
Let $v$ be in almost-clique $C$ and consider $u\in N(v)$. By Lemma \ref{lemma: densecommonnbrs}, it follows that any two vertices $v,u\in C$ have at least one common neighbor $w$ where $u\notin N(v)$. We count the number of length two paths of the form $v,w,u$, denoted by $P$. For any $u\in C\backslash N(v)$, there at most $|N(v)\cap N(u)|$ possible choices for $w$. By Lemma \ref{lemma: densecommonnbrs}, $|N(v)\cap N(u)|\geq (1-2c_3)\Delta$, so that $P=\sum_{u\in C\backslash N(v)} |N(v)\cap N(u)|\geq |E_N(v)|(1-2c_3)\Delta$. On the other hand, summing over all intermediate vertices $w$ on the $v-u$ path, we have that, 
\begin{align*}
    P=\sum_{w\in N(v)}|N(w)\cap (C\backslash N(v))| &\leq \sum_{w\in N(v)}|N(w)\backslash N(v)| \\
            &=\sum_{w\in N_3(v)} |N(w)\backslash N(v)|+\sum_{w\in N(v)\backslash N_3(v)} |N(w)\backslash N(v)| \\
            &\leq \sum_{w\in N_3(v)} c_3\Delta + \sum_{w\in N(v)\backslash N_3(v)} \Delta \\ 
            &\leq (1-c_3)\Delta\cdot c_3\Delta+ c_3\Delta^2 \\
            &=(2-c_3)c_3\Delta^2.
\end{align*}
Combining the inequalities yields $|\b{E_C(v)}|\leq \Delta \frac{c_3(2-c_3)}{1-2c_3}$, which is at most $3c_3\Delta$ for $\varepsilon<\frac{1}{15}-\frac{\tau}{3}$.
\end{proof}

\begin{corollary}\label{corr: ubalmost-cliquesize}
    For $\varepsilon<\frac{1}{15}-\frac{\tau}{3}$, 
    if \Density, \Friendship, and \Connectedness\ invariants hold, then
    any almost-clique on vertices in $V_D$ has size at most $(1+3c_3)\Delta$.
\end{corollary}
\begin{proof}
    For any $v\in C$, note that $|C|=|C\backslash N(v)|+|C\cap N(v)|\leq |E_N(v)|+|N(v)|\leq (1+ 3c_3)\Delta$ where the final inequality follows from Lemma \ref{lemma: non-degree}.
\end{proof}

\begin{lemma}
\label{lem:connectedness}
    For $\varepsilon < \frac{1}{15} - \frac{\tau}{3}$, if the \Friendship\ invariant holds for every vertex and the \Size\ invariant holds for every almost-clique, then the \Connectedness\ invariant holds for every almost-clique.
\end{lemma}
\begin{proof}
    Fix an almost-clique $C$.  By the \Friendship\ invariant, every vertex $v$ of $C$ has at least $(1 - c_3)\Delta$ neighbors in $N_3(v) \cap C$.  By Lemma~\ref{lemma: maintain-friends}, this implies that every $v \in C$ has at least $(1 - c_3 - \nu)\Delta$ $c_3$-friends in $C$.  Since $\varepsilon < \frac{1}{15} - \frac{\tau}{3}$, $c_3 = 3\varepsilon + \tau < \frac{1}{5} - \nu$ for a sufficiently small constant $\nu>0$.  Therefore, $2(1-c_3-\nu) > 1$, implying that any two vertices $u$ and $v$ in $C$ share a neighbor $w$ which is a $c_3$-friend for both.  This establishes the \Connectedness\ invariant for $C$.  
\end{proof}

In the following, we establish three of the four invariants after the decomposition is updated following an edge deletion.

\begin{lemma}
\label{lemma: update-delete}
Following the completion of \textsc{Update-Decomposition} after an edge deletion, the \Density, \Friendship, and \Size\ invariants continue to hold.
\end{lemma}
\begin{proof}
After the deletion of an edge $(u,v)$, $U$ is a subset of vertices for which there is a change the density or the friendship of an incident edge.  For every other vertex $v$, the \Density\ invariant continues to hold since the membership in $V_3$ or $V_1$ is not impacted, and the \Friendship\ invariant continues to hold since the number of neighbors from $N_3(v)$ that lie in its almost-clique can only decrease by the number of sparse moves during this step.   

For any vertex $w$ in $U$, if $w$ is in $V_S$, \SparseMove{w} is not called.  On the other hand, if $w \in U \cap V_D$ and $w$ either violates the \Density\ or \Friendship\ invariant, we move $w$ to $V_S$, thus ensuring that it satisfies the \Density\ and \Friendship\ invariants after the sequence of calls to \SparseMoveNoArgs.  All of the vertices affected by these sparse moves belong to the almost-cliques in ${\cal C}$.  

Consider an almost-clique $C$ in ${\cal C}$.  If the number of sparse moves in $C$ since its creation is at most $\alpha \tau \Delta$, it follows that all of the vertices in $C$ but not in $U$ continue to satisfy the \Density\ and \Friendship\ invariants.  This is because for any such vertex $v$, its membership in $V_3$ has not changed (ensuring the \Density\ invariant) and the number of neighbors in $N_3(v) \cap C$ is at least $(1 - c_3 - \nu)\Delta$ .  Furthermore, the \Size\ invariant holds since (a) vertices have only been removed from $C$, implying the upper bound, and (b) the size is at least $(1 - c_1)\Delta - \nu \Delta$, which is at least the lower bound.  The \Density, \Friendship, and \Size\ invariants in turn imply the \Connectedness\ invariant by Lemma~\ref{lem:connectedness}.  

We now consider the case where $\sigma(C)$ exceeds $\nu \Delta$. Almost-clique $C$ is deleted by moving all its vertices to $V_S$, and then issuing a sequence of \DenseMoveNoArgs\ calls to vertices that are in $V_1$.  By Lemmas~\ref{lemma: denseinvariantaftermove} and~\ref{lemma: connectedness invariant}, it follows that the \Density, \Friendship, and \Connectedness\ invariants hold for each such $C$.  This in turn yields the \Size\ invariant by Corollary~\ref{corr: ubalmost-cliquesize}.
\end{proof}

We now have all the ingredients to establish the four invariants before every update.
\begin{lemma}
\label{lemma: invariants}
The \Density, \Friendship, \Connectedness, and \Size\ invariants hold before every edge update.
\end{lemma}
\begin{proof}
The proof is by induction on the number of edge updates.  The induction base case holds for an empty graph since every vertex is in $V_S$.  Consider any edge update.  If it is an edge insertion, then the algorithm issues a (possibly empty) sequence of dense moves.  By Lemmas~\ref{lemma: denseinvariantaftermove} and~\ref{lemma: connectedness invariant}, \Density, \Friendship, and \Connectedness\ invariants hold after every dense move.  By Corollary~\ref{corr: ubalmost-cliquesize}, the \Size\ invariant also holds.  Thus, the four invariants hold after the edge update is processed, completing the induction step for edge insertion.

If the edge update is an edge deletion, by Lemma~\ref{lemma: update-delete}, the \Density, \Friendship, and \Size\ invariants continue to hold after the completion of \textsc{Update-Decomposition}.  By Lemma~\ref{lem:connectedness}, the \Connectedness\ invariant also continues to hold, thus completing the proof of the lemma.
\end{proof}
The following Lemma establishes that we can maintain a sparse-dense decomposition using Algorithm $\textsc{Update-Decomposition}$ in $O(\frac{\ln n}{\tau^4})$-amortized update time. 

\begin{lemma}\label{lemma: Update-Decomposition}
    Let $\varepsilon-\tau>\frac{\tau}{4}$. Then, Algorithm \textsc{Update-Decomposition}$(\varepsilon, \tau)$ takes $O(\frac{\ln n}{\tau^4})$ amortized update time. 
\end{lemma}

\begin{proof}
We adopt a charging scheme similar to the one in the proof of Lemma \ref{lemma: maintain-friends}. 
Each edge update $(u,v)$ is charged an amount $C\frac{\ln n}{\tau^4}$ such that both $u$ and $v$ receive a credit of $\Theta(\frac{\ln n}{\tau^4})$, where $C$ is a sufficiently large constant.  We refer to this credit as \emph{edge credit}, which is the ultimate source for paying for all of the algorithm's costs and for any other credits created for future steps.  By Lemma~\ref{lemma: densevert}, the cost of \textsc{Determine-Dense} is $\Theta(\Delta \frac{\ln n}{\tau^2})$.  By Lemma~\ref{lemma: is-friend}, the cost of \textsc{Determine-Friend} is $\Theta(\frac{\ln n}{\tau^2})$.  Both of these costs are paid for by edge credits.  By setting the constant $C$ in the edge credit sufficiently large, we assume that (i) every call to \textsc{Determine-Dense} for vertex $v$ leaves a \emph{dense credit} of $\Theta(\Delta \frac{\ln n}{\tau^2})$ on $v$, and (ii) every call to \textsc{Determine-Friend} for vertex $v$ leaves a \emph{friend credit} of $\Theta(\Delta \frac{\ln n}{\tau^2})$ on $v$.

We will now prove that the edge, dense, and friend credits together are sufficient to pay for all of the costs of \textsc{Update-Decomposition}.  In our argument, we will work with two other kinds of credits, which we refer to as \emph{move credit} and \emph{clique credit}.  We maintain the invariant that any vertex in $V_S$ that is in $V_2$ has a move credit of $\Omega(\Delta \ln n)$; we refer to this as the move credit invariant.  The clique credit is assigned to an almost-clique and is used to pay for the cost of reconstructing it after a sparse move sequence.  Again, note that both the move and clique credits ultimately originate from edge credits. 

The cost of lines 2-3 of \textsc{Update-Decomposition} is $O(1)$.  We now analyze the calls to \DenseMoveNoArgs\ in lines~6 and~21.  Consider \DenseMove{v} for a vertex $v$.  By lines~6 and~21, it follows that at the time of the call $w$ is in $V_1$ and $V_S$.  The running time of \DenseMove{v} is $O(k \Delta \ln n)$, where $k$ is the number of vertices in $N_1(v) \setminus V_D$ since lines~1--5 of \DenseMoveNoArgs\ take time $O(\ln n)$ and lines~6--12 take time  
$O(k \Delta \ln n)$.  By the move credit invariant, $w$ and each vertex in $N_1(v) \setminus V_D$ has a move credit of $\Omega(\Delta \ln n)$.  We thus have a total move credit of $O(k \Delta \ln n)$, which we use to pay for the call to \DenseMove{v}.

We now analyze lines~7--19 of \textsc{Update-Decomposition}.  Lines~7--12 maintain relevant data structures and take time proportional to $|U|$.  Consider the call \SparseMove{w}.  We first note that the total running time of \SparseMoveNoArgs\ is $O(\Delta \ln n)$.  Now, \SparseMove{w} is called made only if the almost-clique $C$ containing $w$ has $\sigma(C) < \alpha\tau\Delta$.  In this case, there are two sub-cases.  The first is where $w$ is not in $V_3$.  Since $w$ was in $V_3$ prior to the edge update (by the \Density\ invariant of Lemma~\ref{lemma: invariants}), a call to \textsc{Determine-Dense} for $w$ has been made, yielding an $\Omega(\Delta \ln n)$ dense credit, which we use to pay for (i) the running time of \SparseMoveNoArgs, (ii) leaving a move credit on $w$ to maintain the move credit invariant, and (iii) leaving an $\Omega(\Delta \ln n)$ clique credit on $C$.  The second sub-case is where $w$ has lost $\Omega(\Delta)$ $c_3$-friends; by Lemma~\ref{lemma: dense-nbrsatmove}) $w$ had at least $(1 - 2c_1)\Delta$ $2c_1$-friends at the time $w$ entered $V_D$, but has fewer than $(1 - c_3)\Delta$ $c_3$-friends at this time.  Since $c_3 > 2c_1$ and since fewer than $(c_3 - 2c_1)\Delta/2$ vertices in $C$ have moved from $V_D$ to $V_S$ since the creation of the almost-clique, it follows that either at least $(c_3 - 2c_1)\Delta/4$ edges of $w$ have been deleted or at least $(c_3 - 2c_1)\Delta/4$ neighbors of $w$ are no longer in $N_3(w)$. 
 We use the $\Omega(\Delta \ln n)$ edge credit at $w$ in the former situation and the $\Omega(\Delta \ln n)$ friend credit at $w$ in the latter situation to pay for (i) the cost of the sparse move, (ii) leaving a move credit of $\Omega(\Delta \ln n)$ on $w$ to maintain the move credit invariant, and (iii) adding $\Omega(\Delta \ln n)$ clique credit to $C$. 

Finally, we consider the scenario where $C$ is a clique for which $\sigma(C)$ exceeds $\Omega(\Delta)$ during the course of this update.  In this scenario, we have $\Omega(\Delta^2 \ln n)$ clique credit on $C$.  We use this to pay for the cost incurred in lines 15--18 of \textsc{Update-Decomposition} and for leaving a move credit on all the vertices that are being moved from $V_D$ to $V_S$ before issuing any calls to \DenseMoveNoArgs, as necessary.  This completes the proof.
\end{proof}

The following corollary bounds the amortized \textit{adjustment complexity} of non-edges across all almost-cliques i.e. the average \textit{total} number of non-edges which change across all almost-cliques after an edge update. The proof follows from the proof of Lemma \ref{lemma: Update-Decomposition}, and uses essentially identical charging arguments, and we omit it.
\begin{corollary}\label{corr: adjcomplexitynonedges}
   The amortized adjustment complexity of non-edges obtained by Algorithm \textsc{Update-} \textsc{Decomposition} is $O(\frac{1}{\eps^4})$
\end{corollary}
Finally, we conclude this section by giving a proof of our main technical result summarized in Theorem~\ref{fd-decomposition}.

\begin{proof}[Proof of Theorem \ref{fd-decomposition}]
We first show that the theorem holds for an amortized update time guarantee of $\tilde{O}(\frac{1}{\varepsilon^4})$. Then, we show that the same properties of the decomposition can be recovered in $\tilde{O}(\frac{1}{\varepsilon^4})$ worst-case update time.

    Let $\tau=\frac{\varepsilon}{3}$. This satisfies the condition of Lemma \ref{lemma: Update-Decomposition}. Moreover, we choose $\varepsilon<\frac{3}{50}$ to satisfy the conditions in Corollary \ref{corr: ubalmost-cliquesize} and Lemmas \ref{lemma: unique-dense-almost-clique}, \ref{lemma: dense-nbrsatmove}, \ref{lemma: densecommonnbrs} and \ref{lemma: non-degree}.

    By Lemma~\ref{lemma: invariants}, all of the four invariants hold after the processing of every edge update.
    By the \Size\ invariant, we have that any almost-clique $C_i$ has size at most $(1+3(3\varepsilon+\tau))\Delta=(1+10\varepsilon)\Delta$, and at least $(1-c_3)\Delta=(1-(3\varepsilon+\frac{\varepsilon}{3}))=(1-\frac{10\varepsilon}{3})\Delta\geq (1-4\varepsilon)\Delta$.  By the \Density\ invariant, it follows that $V_S\subseteq \vs_{(\varepsilon+\tau)}\coloneq \vs_{\frac{4\eps}{3}}$. Similarly, $V_D\subseteq \vd_{(3\varepsilon-\frac{3\tau}{4})}\coloneq\vd_{\frac{11\eps}{4}}$.  The \Friendship\ invariant guarantees that every vertex $v$ in an almost clique $C$ has at least $(1-c_3)\Delta\geq (1-4\varepsilon)\Delta$ neighbors in $C$; it also follows that the number of neighbors of $v$ outside $C$ is at most $4\varepsilon\Delta$.

    By Corollary \ref{corr: adjcomplexitynonedges}, the adjustment complexity of non-edges is $O(\frac{1}{\varepsilon^4})$. 

Finally we note that, while it suffices to work with amortized running time and adjustment complexity guarantees for the purpose of obtaining a sublinear (in $n$) update time algorithm for $(\Delta+1)$ coloring, we can transform the aforementioned amortized guarantees to worst-case. This is done by employing the standard technique of \textit{periodically rebuilding} data structures (see \cite{chan2021more, baswana2016dynamic}), by incurring only a constant factor increase in the running time. Our algorithm can be run on a small parameter, say $\frac{\eps}{2}$ as opposed to the parameter $\eps$ it receives as input. Subroutines such as \textsc{Update} and \textsc{Determine-Friend} can be invoked whenever direct or indirect counters for any vertex $v$ change by a smaller value--say $\frac{\tau\Delta}{16}$. The update sequence can be partitioned into epochs of length $\frac{\varepsilon\Delta}{16}$, where the time taken to reflect edge updates in a previous epoch is amortized over the length of the current epoch. Since periodically rebuilding data structures to convert amortized bounds to worst-case bounds is a well-known technique, we omit further details.
\end{proof}